\documentclass[aps,pra,superscriptaddress,amsmath,amssymb,twocolumn,notitlepage,nofootinbib]{revtex4-1}

\usepackage[T1]{fontenc}
\usepackage[utf8]{inputenc}

\usepackage[dvipsnames]{xcolor}
\usepackage{url}
\PassOptionsToPackage{%
    bookmarksnumbered=true, 
    breaklinks=true,
    unicode=false, 
    pdfstartview={FitH}, 
    pdftitle={}, 
    pdfnewwindow=true, 
    colorlinks=true, 
    allcolors=MidnightBlue!70!black!15!blue}{hyperref}
\usepackage{hyperref}

\usepackage[caption=false]{subfig}
\usepackage{graphicx} 
\usepackage{epstopdf}
\usepackage{array}
\usepackage{verbatim}
\usepackage{amsmath,amsfonts,amssymb,amscd,mathtools,bbm}
\usepackage{amsthm}
\usepackage{tabularx}
\usepackage{stmaryrd}
\usepackage{enumerate}
\usepackage[ruled]{algorithm2e}
\usepackage{wasysym}
\usepackage{braket}
\usepackage{mathrsfs}
\usepackage{microtype}

\usepackage{newtxtext,newtxmath}
\usepackage[normalem]{ulem}

\theoremstyle{plain}
\newtheorem{thm}{Theorem}
\newtheorem{cor}[thm]{Corollary}
\newtheorem{lem}[thm]{Lemma}
\newtheorem{pro}[thm]{Proposition}

\theoremstyle{definition}
\newtheorem{defn}[thm]{Definition}
\newtheorem{remark}[thm]{Remark}

\newcommand{\nc}{\newcommand}

\nc{\R}{R} 
\nc{\Rl}{\underline{R}} 
\nc{\Rs}{R^{s}} 
\nc{\Rsl}{\underline{R}^{s}} 

\nc{\eq}[1]{(\hyperref[eq:#1]{\ref*{eq:#1}})}
\nc{\eqrange}[2]{Eqs. (\eq{#1}--\eq{#2})}
\renewcommand{\sec}[1]{\hyperref[sec:#1]{Section~\ref*{sec:#1}}}
\nc{\thrm}[1]{\hyperref[thrm:#1]{Theorem~\ref*{thrm:#1}}}
\nc{\lemm}[1]{\hyperref[lemm:#1]{Lemma~\ref*{lemm:#1}}}
\nc{\prop}[1]{\hyperref[prop:#1]{Proposition~\ref*{prop:#1}}}
\nc{\corr}[1]{\hyperref[corr:#1]{Corollary~\ref*{corr:#1}}}
\nc{\fig}[1]{\hyperref[fig:#1]{~\ref*{fig:#1}}}
\nc{\deff}[1]{\hyperref[deff:#1]{~\ref*{deff:#1}}}

\nc{\A}{\mathcal{A}}
\nc{\C}{\mathcal{C}}
\nc{\E}{\mathcal{E}}
\nc{\N}{\mathcal{N}}
\nc{\U}{\mathcal{U}}
\nc{\T}{\mathcal{T}}
\nc{\Q}{\mathcal{Q}}
\nc{\G}{\mathcal{G}}
\nc{\D}{\mathcal{D}}
\nc{\I}{\mathcal{I}}
\nc{\F}{\mathcal{F}}
\nc{\f}{\mathscr{f}}
\renewcommand{\H}{\mathcal{H}}
\nc{\mL}{\mathcal{L}}
\nc{\rl}{\rangle\langle}
\nc{\mg}{\mathcal{G}}
\nc{\M}{\mathcal{M}}
\renewcommand{\O}{\mathcal{O}}
\nc{\B}{\mathcal{B}}
\nc{\K}{\mathcal{K}}
\renewcommand{\S}{\mathcal{S}}
\nc{\X}{\mathcal{X}}
\nc{\Y}{\mathcal{Y}}
\nc{\Z}{\mathcal{Z}}
\nc{\mbF}{\mathbb{F}}
\nc{\mbI}{\mathbb{I}}
\nc{\mbM}{\mathbb{M}}
\nc{\mbN}{\mathbb{N}}

\nc{\hn}[1]{\|#1\|^H_{1\rightarrow 1}}
\nc{\ro}[1]{\langle\!\langle#1|}
\nc{\V}{\mathcal{V}}
\renewcommand{\*}{\textup{*}}
\DeclareMathOperator{\cone}{cone}
\DeclareMathOperator{\conv}{conv}
\DeclareMathOperator{\sint}{int}
\DeclareMathOperator{\cl}{cl}
\DeclareMathOperator{\ran}{ran}
\nc{\cleq}{\preceq}
\nc{\cgeq}{\succeq}
\nc{\cle}{\prec}

\nc{\tth}[0]{\textsuperscript{th}}
\nc{\st}[0]{\textsuperscript{st}}
\nc{\nd}[0]{\textsuperscript{nd}}
\nc{\rd}[0]{\textsuperscript{rd}}

\nc{\RR}{\mathbb{R}}
\nc{\CC}{\mathbb{C}}

\DeclareMathOperator{\Tr}{Tr}

\nc{\Favg}{\overline{F}}
\nc{\Paulis}{{\matheu P}}
\nc{\Clifs}{{\matheu C}}
\nc{\Hilb}{{\matheu H}}
\nc{\id}{\mathbbm{1}}

\let\mathscr\relax
\DeclareFontFamily{U}{mathc}{}
\DeclareFontShape{U}{mathc}{m}{it}%
{<->s*[1.03] mathc10}{}
\DeclareMathAlphabet{\mathscr}{U}{mathc}{m}{it}

\nc{\sop}[1]{{\mathcal #1}}
\nc{\PL}[1]{{#1}^{P\!L}}
\nc{\BHn}{{{\mathcal B}({\mathcal H}^{\otimes n})}}

\nc{\ketbra}[2]{|{#1}\rangle\!\langle{#2}|}
\nc{\kket}[1]{|{#1}\rangle\!\rangle}
\nc{\bbra}[1]{\langle\!\langle{#1}|}
\nc{\bbrakket}[2]{\langle\!\langle{#1}|{#2}\rangle\!\rangle}
\nc{\kketbbra}[2]{|{#1}\rangle\!\rangle\!\langle\!\langle{#2}|}

\nc{\no}{\nonumber\\}
\nc{\ba}{\begin{eqnarray}}
\nc{\ea}{\end{eqnarray}}
\nc{\bann}{\begin{eqnarray*}}
\nc{\eann}{\end{eqnarray*}}
\nc{\bal}{\begin{equation}\begin{aligned}}
\nc{\eal}{\end{aligned}\end{equation}}
\nc{\dm}[1]{\ketbra{#1}{#1}}
\nc{\bs}[1]{\boldsymbol{#1}}
\nc{\cbraket}[1]{\abs{\braket{#1}}}

\nc{\txr}[1]{{\color{red!85!black} #1}}
\nc{\txb}[1]{{\color{blue!65!white} #1}}

\nc{\note}[1]{{\color{blue!65!red} #1}}

\newcolumntype{L}[1]{>{\raggedright}p{#1}}
\newcolumntype{C}[1]{>{\centering}p{#1}}
\newcolumntype{R}[1]{>{\raggedleft}p{#1}}
\newcolumntype{D}{>{\centering\arraybackslash}X}

\nc{\TT}{\mathcal{L}}

\nc{\tf}{\tilde{F}}

\renewcommand{\*}{\textup{*}}
\nc{\<}{\left\langle}
\renewcommand{\>}{\right\rangle}
\DeclarePairedDelimiter\abs{\lvert}{\rvert}%
\nc{\norm}[1]{\left\|#1\right\|}
\let\V\V
\let\C\C

\nc{\NN}{\mathbb{N}}

\nc{\proj}[1]{\ket{#1}\!\bra{#1}}

\nc{\sbar}{\;\rule{0pt}{9.5pt}\right|\;}

\nc{\lset}{\left\{\left.}
\nc{\rset}{\right\}}
\nc{\ve}{\varepsilon}
\nc{\lsetr}{\left\{\,}
\nc{\rsetr}{\right.\right\}}
\nc{\sbarr}{\,\rule{0pt}{9.5pt}\left|\;}

\let\txb\relax
\let\note\relax

\begin{document}

\title{Framework for resource quantification in infinite-dimensional general probabilistic theories}


\author{Ludovico Lami}
\thanks{These authors contributed equally to this work.}
\affiliation{Institut f\"ur Theoretische Physik und IQST, Universit\"at Ulm, Albert-Einstein-Allee 11, D-89069 Ulm, Germany}
\email{ludovico.lami@gmail.com}
\author{Bartosz Regula}
\thanks{These authors contributed equally to this work.}
\affiliation{School of Physical and Mathematical Sciences, Nanyang Technological University, 637371, Singapore}
\email{bartosz.regula@gmail.com}
\author{Ryuji Takagi}
\affiliation{Center for Theoretical Physics and Department of Physics, Massachusetts Institute of Technology, Cambridge, Massachusetts 02139, USA}
\affiliation{School of Physical and Mathematical Sciences, Nanyang Technological University, 637371, Singapore}
\author{Giovanni Ferrari}
\affiliation{Dipartimento di Fisica e Astronomia Galileo Galilei,
Universit\`a degli studi di Padova, via Marzolo 8, 35131 Padova, Italy}
\affiliation{Institut f\"ur Theoretische Physik und IQST, Universit\"at Ulm, Albert-Einstein-Allee 11, D-89069 Ulm, Germany}


\begin{abstract}%
Resource theories provide a general framework for the characterization of properties of physical systems in quantum mechanics and beyond. Here, we introduce methods for the quantification of resources in general probabilistic theories (GPTs), focusing in particular on the technical issues associated with infinite-dimensional state spaces. 
We define a universal resource quantifier based on the robustness measure, and show it to admit a direct operational meaning: in any GPT, it quantifies the advantage that a given resource state enables in channel discrimination tasks over all resourceless states. We show that the robustness acts as a faithful and strongly monotonic measure in any resource theory described by a convex and closed set of free states, and can be computed through a convex conic optimization problem. 

Specializing to continuous-variable quantum mechanics, we obtain additional bounds and relations, allowing an efficient computation of the measure and comparison with other monotones. We demonstrate applications of the robustness to several resources of physical relevance: optical nonclassicality, entanglement, genuine non-Gaussianity, and coherence. In particular, we establish exact expressions for various classes of states, including Fock states and squeezed states in the resource theory of nonclassicality and general pure states in the resource theory of entanglement, as well as tight bounds applicable in general cases. 
\end{abstract}

\maketitle


\section{Introduction}

The success of quantum mechanics in fields such as computation, communication, and information processing owes to the fact that certain properties of quantum systems, dubbed \emph{resources}, can be exploited to enable significant advantages over purely classical methods in practical tasks. A central aim in the investigation of such properties is their \emph{quantification}, that is, the development of ways to measure and compare the resources contained in quantum systems~\cite{chitambar_2019}. This can be done in two seemingly different ways. On the one hand, we can identify tasks of particular importance and ask: how useful is the given state in performing this task?~\cite{horodecki_2012,brandao_2015,chitambar_2019,takagi_2019-2} On the other hand, we can approach resource quantification from a very general, abstract perspective, and establish resource measures which apply to broad classes of resources~\cite{delrio_2015,coecke_2016,regula_2018,kostecki_2019,takagi_2019,chitambar_2019,gonda_2019}. The latter approach is attractive as it can reveal common features connecting very different physical phenomena and is immediately applicable in a variety of settings, but it often lacks a direct relation with the operational significance of resources. The recent years have seen remarkable progress in connecting the two approaches, showing that resource measures which directly quantify practical advantages can be defined in broad classes of resource theories~\cite{brandao_2015,anshu_2018-1,takagi_2019-2,uola_2019-1,takagi_2019,sparaciari_2020,liu_2019,regula_2020,kuroiwa_2020,seddon_2020,ducuara_2019,uola_2019-2}.

In the search for a universal axiomatic characterization of resources, it is then appealing to ask whether these general quantitative methods can apply even in physical theories beyond quantum mechanics~\cite{delrio_2015,coecke_2016}. The framework of \textit{general probabilistic theories} (GPTs)~\cite{ludwig_1985,hartkamper_1974,davies_1970,lami_2018-1} provides a formalism which encompasses both quantum and classical probability theory among a myriad of more general theories. \txb{Its generality makes it an ideal candidate to study, for instance, hypothetical non-orthodox modifications of quantum mechanics~\cite{barnum_2014, vandam_2013}, or extensions of it that include novel physical phenomena --- the archetypal example being gravity~\cite{Mielnik-general-quantum, galley_2020}. 
Fortunately, resources can be studied also at this level of generality}~\cite{Chiribella_2015entanglement, Chiribella_2017microcanonical, kostecki_2019, takagi_2019, lami_2017, aubrun_2019, aubrun_2019b}. In particular, Ref.~\cite{takagi_2019} extended the connection between resource quantification and advantages in operational tasks also to resource theories in GPTs.

A major issue with the aforementioned approaches is that they typically apply only to finite-dimensional state spaces, or they place significant restrictions on the resources in consideration. For instance, previous works on resources in infinite-dimensional quantum mechanics are limited to the restricted Gaussian framework~\cite{lami_2018,lami_2020} or make strong technical assumptions such as compactness of the relevant sets~\cite{kuroiwa_2020}. This prevents us from being able to use them in the description of some of the most fundamental physical settings, such as quantum optical systems or GPTs in general Banach spaces. An extension of the known results is non-trivial: the previously studied methods explicitly make use of the simplified topological structure of finite-dimensional spaces, and the optimization problems which define the relevant resource measures have not been characterized in the considerably more complex infinite-dimensional setting.

In this work, we develop a general method of quantifying resources in infinite dimensions, applicable both to quantum mechanics and broader GPTs, and directly connected with the performance advantage in practical tasks. The framework is based on a measure called the \emph{robustness}, which found use in a variety of finite-dimensional settings in quantum mechanics~\cite{vidal_1999, harrow_2003, brandao_2008-1, brandao_2010, napoli_2016, regula_2018, anshu_2018-1, takagi_2019-2, takagi_2019, liu_2019, regula_2020, seddon_2020} and in GPTs~\cite{takagi_2019}, but hitherto has not been considered in infinite-dimensional spaces. In particular, we introduce a variant of the robustness applicable to infinite-dimensional theories and show that it directly quantifies the maximal advantage that a given state provides in a family of channel discrimination tasks. This extends from the finite-dimensional cases a deep connection between resource quantification and the fundamental operational tasks of discrimination. The results apply to any resource theory described by a convex and closed set of states in any GPT, ensuring immediate applicability to the vast majority of physical settings of interest. Our methods rely on a novel extension of the optimization methods which underlie the robustness measure to arbitrary Banach spaces, explicitly considering in detail the issues which arise in such a generalization, and obtaining a characterization of properties such as strong duality.  We further use our results to show that the robustness satisfies all of the requirements commonly desired from a valid resource measure~\cite{vidal_2000,chitambar_2019}, including faithfulness and strong monotonicity.

Specializing to continuous-variable quantum mechanics, we introduce methods to simplify the evaluation of the robustness in several cases, and show that the measure can be computed analytically for some of the most relevant examples of states in resource theories of nonclassicality, entanglement, coherence, and genuine non-Gaussianity. We in particular establish exact expressions or {tight} bounds for: Fock states, squeezed states, and cat states in the resource theory of nonclassicality, as well as all pure states and an important mixed state based on the Hilbert operator in the resource theories of entanglement and coherence. We compare the robustness to a related resource measure often employed in finite dimensions, the standard robustness~\cite{vidal_1999,takagi_2019}, and show that the latter is not a well-behaved monotone in important continuous-variable resource theories such as entanglement, nonclassicality, and coherence. Specifically, we provide examples of states where our robustness measure is finite and well-behaved but the standard robustness diverges to infinity, and show in particular that this behavior affects most physically accessible states in the resource theory of nonclassicality.

All of our findings provide evidence for the suitability of our infinite-dimensional robustness measure as a universal resource quantifier in general resource theories, satisfying desirable properties and computable in many cases.

This work also serves as the companion to the paper~\cite{our_main}, which deals with resource quantification in continuous-variable quantum mechanics. Here we provide a derivation and extended discussion of the results stated in~\cite{our_main}, along with several additional developments specific to quantum theory. The general framework in GPTs can be thought of as a generalization of the concepts introduced in~\cite{our_main}.

The organization of the paper is as follows. Section~\ref{sec:prelim} contains a comprehensive introduction to all of the relevant concepts, including GPTs, resource theories, and resource quantification. In Section~\ref{sec:defining_robustness}, we discuss the problems concerning defining a robustness measure in infinite dimensions, and establish a characterization of the relevant conic optimization problems. We show in Section~\ref{sec:robustness_properties_monotone} that our proposed variant of the robustness indeed satisfies all of the requirements for a valid resource measure in any convex resource theory. Section~\ref{sec:discrimination} establishes the robustness as the {figure of merit in} channel discrimination tasks. In Section~\ref{sec:quantum} we discuss simplifications which occur when the GPT is chosen to be infinite-dimensional quantum mechanics, establishing several bounds and expression for the robustness, as well as characterizing strong duality. Finally, in Section~\ref{sec:examples} we explicitly apply the robustness to the resource theories of nonclassicality (Sec.~\ref{sec:nonclassicality}), entanglement (Sec.~\ref{sec:entanglement}), genuine non-Gaussianity (Sec.~\ref{sec:nongaussianity}), and coherence (Sec.~\ref{sec:coherence}), obtaining a multitude of resource-specific results.


\section{Preliminaries}\label{sec:prelim}

We first introduce some basic notation that will be used throughout the paper.

Given a real topological vector space $\V$, we will use $\V\*$ to denote its continuous dual space, that is, the space of continuous linear functionals $F\colon \V \to \RR$. We write $\< F, x \> = F(x)$ for any $x \in \V, F \in \V\*$. 

A set $\S$ is called convex if $t x + (1-t) y \in \S$ for any $x,y \in \S$ and $t\in[0,1]$, and it is called a cone if $x \in \S \Rightarrow \lambda x \in \S$ for all $\lambda \in \RR_+$. We will use $\conv(\S) = \lset \sum_{i=1}^n c_i x_i \sbar x_i \in \S,\; c_i \in \RR_+,\; \sum_i c_i = 1 \rset$ to denote the convex hull of $\S$, that is, the smallest convex set which contains $\S$; analogously, $\cone(\S) = \lset \lambda x \sbar x \in \S,\; \lambda \in \RR_+ \rset$ will denote the conic hull of $\S$, and $\cl(\S)$ the closure of $\S$, that is, the smallest closed set which contains $\S$. For two sets $\S, \mathcal{R}$, we use $\S + \mathcal{R} = \lset s + r \sbar s \in \S,\; r \in \mathcal{R} \rset$ to denote their Minkowski sum. 
Given a set $\S \subseteq \V$, its dual cone $\S\* \subseteq \V\*$ is given by $\lset Y \in \V\* \sbar \< Y, x \> \geq 0 \; \forall x \in \S \rset$.

As is standard in optimization theory, we will take $\inf \emptyset = \infty$ and $\sup \emptyset = -\infty$. We will allow algebraic operations involving infinity whenever it leads to no ambiguities, with $c + \infty = \infty$ and $c \infty = \operatorname{sgn}(c) \infty$ for any  $c \in \RR$. We use $\log$ to denote the logarithm to the base 2 and $\ln$ for the natural logarithm.


\subsection{General probabilistic theories}

We provide a brief introduction to the formalism of GPTs based on standard references~\cite{ludwig_1983,ludwig_1985,davies_1970}; see also \cite[Ch.\,1-2]{lami_2018-1} for a modern rigorous introduction to the concepts. GPTs can be thought of a generalization of the formalism of quantum mechanics, aiming to recover a characterization of the probabilistic and statistical aspects of the manipulation of physical systems in an axiomatic manner. 
While some literature considers only the technically more elementary finite-dimensional case, here we set out to describe the GPT machinery in full generality. The universality of the formalism described below rests upon a result by Ludwig~\cite[Ch.~IV, Thm.~3.7]{ludwig_1985}, who deduced it from first principles --- we expand on that below.

Our setting will be a \textit{base norm Banach space}, namely, a type of ordered Banach space in which the order and the norm structures are intimately related with each other. To understand this notion better, let us discuss how a base norm space can be constructed. The basic object of interest is an ordered vector space $\V$, equipped with the \textit{positive cone} $\C$, which induces a partial order on $\V$ according to $x \cleq_\C y \iff y - x \in \C$. 
For $(\V,\cleq_\C)$ to be an ordered vector space, we need $\C$ to be a convex cone (i.e., $\lambda \C + \mu\C \subseteq \C$ for all $\lambda,\mu\geq 0$) that is also \textit{pointed} (i.e., $\C \cup (-\C) = \{0\}$). In what follows, we will also assume that $\C$ is spanning (i.e., $\C - \C = \V$).

The (algebraic) dual $\V\*$ of an ordered vector space $\V$ can itself be turned into an ordered vector space by the choice of the dual cone $\C\*$ as the positive cone of $\V\*$. To define a GPT we also need to single out a strictly positive functional $U\in \C\*$ called the \textit{unit effect}. Here, strict positivity means that $\<U,x\>\geq 0$ for all $x\in \C$, with equality if and only if $x=0$. The order unit can be used to construct the \textit{state space} $\Omega\coloneqq \{ x\in \C\,|\, \braket{U,x}=1\}$ as well as the function $\norm{\cdot}_\Omega:\V \to \RR$ defined by
\begin{equation} \label{base norm} \begin{aligned}
  \norm{x}_\Omega &\coloneqq \inf \lset \lambda_+\! +\! \lambda_- \sbar x = \lambda_+ \omega_+\! -\! \lambda_- \omega_-,\; \lambda_{\pm} \in \RR_+,\; \omega_{\pm} \in \Omega \rset\\
  &= \inf\lset \< U,x_++x_- \>\sbar x=x_+-x_-,\; x_\pm \in \C \rset .
\end{aligned}\end{equation}
In the above setting, it is not difficult to verify that $\norm{\cdot}_\Omega$ is always at least a seminorm. For $\V$ to be a Banach space, $\norm{\cdot}_\Omega$ needs to be a norm which makes $\V$ complete\footnote{Namely, such that all Cauchy sequences converge. Here, a sequence $(x_n)_n$ of elements of $\V$ is said to be \textit{Cauchy} if for all $\epsilon>0$ one can find an integer $N$ such that $\norm{x_n-x_m}_\Omega \leq \epsilon$ for all $n,m\geq N$, and to converge if there exists $x\in \V$ such that $\lim_{n\to\infty} \norm{x_n-x}_\Omega = 0$.}, in which case we refer to $\V$ as a base norm Banach space with base $\Omega$. We can then without loss of generality assume that $\C$ and $\Omega$ are both closed sets with respect to the topology induced by the base norm~\cite{ellis_1966}.

Hereafter, we will always assume that $\V$ is a base norm Banach space, and we will use the topology induced by the norm $\norm{\cdot}_\Omega$ unless stated otherwise.

In order to axiomatically define the concept of a measurement, we need to consider linear functionals $E: \Omega \to \RR$ whose values correspond to the probability of obtaining a certain measurement outcome. 
Under the so-called no restriction hypothesis~\cite{ludwig_1985,barrett_2007,chiribella_2010}, which we will take as an axiom, the set of physically implementable measurement functionals is then formed by all possible continuous linear functionals $E: \Omega \to [0,1]$, which is precisely the set of functionals $E$ such that $0 \cleq E \cleq U$. Any such functional is called an \textit{effect}, and a collection of effects $\{E_i\}$ such that $\sum_i E_i = U$ will be called a \textit{measurement}, since its outcomes sum up to a valid probability distribution.

With the above notation, the base norm can be recast into the following dual form:
\begin{equation} \label{base norm dual} \begin{aligned}
  \norm{x}_\Omega  &= \sup \lset \< E, x \> \sbar -U \cleq_\C E \cleq_\C U \rset.
\end{aligned}\end{equation}
The dual norm on $\V\*$ is known as the \textit{order unit norm}
\begin{equation}\begin{aligned}
  \norm{Y}^\circ_\Omega &= \sup \lset \< Y, x \> \sbar \norm{x}_{\Omega} \leq 1 \rset\\
  &= \sup \lset \abs{\< Y, \omega \>} \sbar \omega \in \Omega \rset. 
\end{aligned}\end{equation}
Accordingly, we will refer to $\V\*$ as an order unit Banach space.

Before we move on, let us devote a moment to look at quantum theory from this new perspective. Given a Hilbert space $\mathcal{H}$, we can consider the base norm Banach space of all trace class operators on $\mathcal{H}$, ordered by the cone of positive semidefinite operators. Choosing the unit effect to be simply the trace, we see that the corresponding state space is formed by all density operators, i.e., positive semidefinite operators with trace $1$. The base norm turns out to coincide with the trace norm $\norm{\cdot}_1$, and the dual space with the order unit space of all bounded operators on $\mathcal{H}$, equipped with the operator norm $\norm{\cdot}_\infty$. It can be noticed that measurements in the GPT sense generalize the concept of positive operator-valued measure (POVM) elements from quantum mechanics, and that a quantum measurement is simply identified with any valid POVM.

The final component of any GPT are the physical transformations, or channels, between states. This issue is typically more difficult to approach in an axiomatic way~\cite{edwards_1971,edwards_1972}, making it dependent on the particular setting {under} consideration. However, we will not need to assume anything about the set of physical transformations save for two of the weakest assumptions: first, that the identity transformation is physical (i.e., doing nothing is allowed), and second, that any physical transformation $\Lambda : \V \to \V'$ always maps a state $\omega \in \Omega$ into a valid state $\Lambda(\omega) \in \Omega'$ in the output state $\V'$. These general assumptions guarantee that our results apply to any physical GPT in consideration.

To consider probabilistic state transformations, we will employ the general notion of an \textit{instrument}~\cite{davies_1970}. This corresponds to a collection of transformations $\{\Lambda_i\}$ which are unnormalized channels, in the sense that each $\Lambda_i$ maps a state $\omega$ to another state $\omega_i'$ with some probability $p_i$, and thus the post-selected output of the transformation can be written as $\Lambda_i(\omega) = p_i \omega_i'$. To ensure that the overall transformation is physical, we will consider an instrument to be any set $\{\Lambda_i\}_i : \V \to \V'$ such that $\Lambda_i(\omega) \in \C'$ for all $\omega \in \Omega$, and $\sum_i \Lambda_i$ is a normalization-preserving transformation, i.e., $\<U', \sum_i \Lambda_i(\omega) \> = \<U, \omega\>$. This allows us to treat $\<U', \Lambda_i(\omega) \>$ as the probability of the $i^\text{th}$ transformation occurring.

As already anticipated, the GPT formalism, which may at first glance look as an ad-hoc generalization of quantum mechanics to a broader setting, can be actually deduced from first principles. Namely, the so-called Ludwig embedding theorem~\cite[Ch.~IV, Thm.~3.7]{ludwig_1985} guarantees that any physical theory that obeys some basic requirements is equivalent to a GPT. Here, a physical theory is intended as a map that associates with every preparation procedure $\omega$ of a fixed physical system and every measurement setting and outcome, together described as $E$, a number $\mu(E,\omega)\in [0,1]$ that represents the probability of obtaining the outcome $E$ when measuring the system prepared according to the procedure $\omega$. Among the axioms required for the Ludwig embedding theorem to apply, the only less intuitive one is the no-restriction hypothesis~\cite[Ch.~1]{lami_2018-1}.

\subsection{Resource theories}\label{sec:resource_theories_intro}

A resource theory is concerned with the description of the manipulation of physical systems under {a set of constraints that} single out a specific property of states as a ``resource'' which --- within the given setting --- is expensive to generate or preserve. One then defines the set of free states $\F \subseteq \Omega$ as the states which do not possess a given resource, and the free operations $\O$ as the channels which can be implemented without using any resource. Any state $\omega \notin \F$ and any channel $\Lambda \notin \O$ is considered resourceful, and is not allowed to be used for free within the physical constraints of the theory. In its general formulation, the formalism can describe virtually any aspect of manipulating physical states in any desired setting. However, making meaningful statements about general resources and their practical applications can be difficult, since different resources are built upon very different physical constraints, which means that they can have very different sets of free states and admit a variety of different classes of {free} operations. Our aim will therefore be to remain as general as possible in our investigation, only making very weak assumptions about resources in consideration, so that the results immediately apply to broad classes of resources.

We thus begin with two basic assumptions about the set $\F$: that it is closed with respect to the topology induced by the base norm $\norm{\cdot}_\Omega$, and that it is convex.

The assumption of closedness is a very natural one in any GPT. It can be understood as a consequence of the fact that the base norm distance $\norm{\omega - \sigma}_\Omega$ quantifies the distinguishability of two different states~\cite{ludwig_1983,kimura_2010,HOLEVO1973337,Helstrom}. That is, if a sequence of states $\{\omega_n\}_n$ satisfies $\norm{\omega_n - \omega}_\Omega \to 0$, the states in the sequence become indistinguishable from $\omega$ by any measurement. As a consequence, performing any experiment on states in the sequence should yield results consistent with the results one would obtain by performing the experiment on $\omega$ directly. In this sense, we can assume that a resource cannot be generated by simply taking a sequence of resourceless states, as this would contradict their indistinguishability.

Convexity, as we have mentioned, is a ubiquitous feature of GPTs. It relies only on the fundamental assumption that, having access to a preparation procedure which outputs a state $\sigma_1$ with probability $p$ and another state $\sigma_2$ with probability $1-p$, we are allowed to erase the information about which of the states was prepared, resulting in the probabilistic mixture $p \sigma_1 + (1-p) \sigma_2$. In most physical contexts, it is then intuitive to expect that simply mixing two states without any resource cannot result in a resourceful state.

While the vast majority of {physically relevant} resources are convex by definition, in some contexts the free states can form a non-convex set --- the most common example being the resource theory based on the set of Gaussian states~\cite{genoni_2010}. However, when considered in an operational setting, it is often easy to prepare mixtures of such states~\cite{takagi_2018,albarelli_2018}, meaning that such convex combinations can also be considered free as they are not able to provide practical advantages when manipulated by free operations. Motivated by this, we will thus take convexity as a basic axiom of the considered resources, and we will call any theory satisfying this property a \textit{convex resource theory} for clarity.

As for the free operations, we only make the weakest possible assumption: that no free operation can generate any resource from a resourceless state. That is, any free channel $\Lambda$ must satisfy $\sigma \in \F \Rightarrow \Lambda(\sigma) \in \F$. This class of maps is known as the maximal free operations~\cite{chitambar_2019}, as any physical class of free operations should be a subset of $\O$. In a very similar way, we deem a probabilistic instrument $\{\Lambda_i\}$ to be free whenever $\sigma \in \F \Rightarrow \Lambda_i(\sigma) \in \cone(\F)$ for all $i$, that is, each $\Lambda_i$ preserves the set of free states up to normalization.


\subsection{Resource monotones}\label{sec:measures_axioms}

There are many, in general inequivalent, ways to quantify and compare the resource content of a state~\cite{chitambar_2019,regula_2018}. An axiomatic characterization of the conditions necessary for a function $M: \Omega \to \RR_+ \cup \{\infty\}$ to constitute a valid resource measure has been considered in various works~\cite{vidal_2000,horodecki_2001,chitambar_2019}. For simplicity, we will assume that the measure is suitably normalized and shifted so that its minimal value is 1; another common option is to consider 0 as the baseline value. The minimal requirements can then be formulated as follows.
\begin{enumerate}[(1{a})]
\item Minimality on the set of free states. That is, $\omega \in \F \Rightarrow M(\omega) = 1$.
\end{enumerate}
This is a natural requirement, since free states have less resources than resourceful ones by definition. However, a problem with this definition is that it does not ensure that \textit{only} free states minimize the given measure. Therefore, a stronger constraint is often imposed:

\begin{enumerate}[(1{b})]
\item Faithfulness. That is, $M(\omega) = 1 \iff \omega \in \F$.
\end{enumerate}
The condition (1b) is sometimes eschewed in favor of ease of evaluation, as non-faithful monotones can be easier to compute than faithful ones. However, it can be seen that faithfulness is required for any monotone which precisely characterizes a given resource, since any such measure needs to be able to distinguish a free state from a non-free state.

The next condition concerns the behavior of the measure under free operations.

\begin{enumerate}[(1{a})]
\setcounter{enumi}{1}
\item Monotonicity under free operations. That is, $M(\Lambda(\omega)) \leq M(\omega)$ for any $\Lambda \in \O$.
\end{enumerate}
This is another natural requirement, since a free operation cannot generate any resources by definition. However, this condition does not capture the fact that state transformations can be performed probabilistically, where --- with some nonzero probability --- it might actually increase the resource~\cite{vidal_2000}. To account for this, a stronger requirement is often imposed.
\begin{enumerate}[(1{b})]
\setcounter{enumi}{1}
\item Strong monotonicity under free operations. That is, $\sum_i p_i M\left(\frac{\Lambda_i(\omega)}{p_i}\right) \leq M(\omega)$ for any free probabilistic protocol $\{\Lambda_i\}$ where $p_i = \< U, \Lambda_i(\omega) \>$.
\end{enumerate}
We remark that, within quantum mechanics, a weaker notion of strong monotonicity based on the Kraus operators of a channel is often employed~\cite{vidal_2000,chitambar_2019}.

Additionally, in all convex resource theories, a natural property that we previously discussed is that probabilistically mixing states should not increase the resource. This leads to the next condition.
\begin{enumerate}[(1)]
\setcounter{enumi}{2}
\item Convexity. That is, $M\left(\sum_i c_i \omega_i\right) \leq \sum_i c_i M(\omega_i)$ for any convex combination $\sum_i c_i \omega_i$.
\end{enumerate}

Our final requirement concerns the continuity properties of the monotones. Although various notions of continuity are sometimes used in the description of finite-dimensional resource measures~\cite{horodecki_2001,chitambar_2019}, such requirements are often too strong in infinite dimensions~\cite{eisert_2002}, and even in finite dimensions many useful resource monotones are not continuous~\cite{brandao_2010}. However, a fundamental property of a resource that we discussed is its closedness. It is then natural to require that a resource cannot be increased by simply taking a limit of less resourceful states. This property is reflected by the lower semicontinuity of the function.
\begin{enumerate}[(1)]
\setcounter{enumi}{3}
\item Lower semicontinuity. That is, $\displaystyle M(\omega) \leq \liminf_{n\to\infty} M(\omega_n)$ for any sequence such that $\{\omega_n\}_n \to \omega$.
\end{enumerate}
Having established these basic criteria, we can introduce a general definition of a bona fide resource measure.
\begin{defn}
A \textbf{resource measure} (or resource monotone) is a function $M: \Omega \to \RR_+ \cup \{\infty\}$ which satisfies the conditions of faithfulness (1b), strong monotonicity (2b), convexity (3), and lower semicontinuity (4).
\end{defn}

Although these axiomatic requirements are often employed in the characterization of resources, they can be insufficient to ensure that a given function is \textit{useful} as a quantifier. Take, for instance, the indicator function
\begin{equation}\begin{aligned}
	M(\omega) \coloneqq \begin{cases} 1 & \omega \in \F\\ \infty & \omega \notin \F.\end{cases}
\end{aligned}\end{equation}
Such a function satisfies all of the criteria (1)-(4), but clearly is useless in benchmarking and comparing the resource content of different states. This issue becomes even more troublesome in infinite dimensions, where the existence of infinitely resourceful states \txb{with respect to certain tasks} is a physical possibility~\cite{eisert_2002, keyl_2003}. \txb{We have to keep in mind that the amount of resource a given state contains depends on the operational task under examination. Therefore, it is perfectly possible that certain monotones yield a finite value when evaluated on a state, while others diverge.} 
\note{However, from a pragmatic standpoint, the primary function of a resource measure should be to enable a quantitative comparison of the resource contents of different states. Therefore, a monotone which assigns an infinite value to a large fraction of quantum states can often be regarded as less useful, and it is of interest to establish monotones which remain finite even when others might diverge. Such a requirement appears to be difficult to formalize in a fully general way, but we will consider it explicitly in several relevant cases.}


\subsection{Robustness measures}\label{sec:robustness_intro}

Having established the basic properties which can be expected from a measure of a resource, we now consider a general type of quantifier in any resource theory: the so-called robustness measures. Defined first for quantum entanglement~\cite{vidal_1999} and later generalized to other finite-dimensional resources in quantum theory~\cite{brandao_2015} and GPTs~\cite{takagi_2019}, the measures are defined through the basic concept of convex combinations of states. Intuitively, to quantify how robust the resources contained in a state $\omega$ are to noise, one can ask: how much noise in the form of statistical mixing with some suitable noisy state $\tau$, i.e.
\begin{equation}\begin{aligned}
	\frac{1}{1+\lambda} \,\omega + \frac{\lambda}{1+\lambda} \,\tau,
\end{aligned}\end{equation}
can $\omega$ withstand before it becomes a free state? Two common types of robustness are then defined by optimizing over the noise states $\tau$. The \textit{(generalized) robustness} can be defined for any state $\omega$ as
\begin{equation}\begin{aligned}\label{eq:rob_primal_nonlsc}
	\R_\F(\omega) \coloneqq \inf \lsetr 1+\lambda \sbarr \frac{\omega + \lambda \tau}{1+\lambda} \in \F, \; \tau \in \Omega \rsetr.
\end{aligned}\end{equation}
Note that our definition differs by a constant term 1 from the commonly used terminology~\cite{vidal_1999,takagi_2019} which defines robustness as $R_\F(\omega)-1$ (or, generally, $R_\F(x)-\<U, x\>$). This inconsequential change in notation is employed because it will make the measure more straightforward to apply also to unnormalized elements of $\V$.

The other common type of robustness is the \textit{standard robustness} (also known as \textit{free robustness}), for which the noise states are also free:
\begin{equation}\begin{aligned}\label{eq:robs_primal_nonlsc}
	\Rs_\F(\omega) \coloneqq \inf \lsetr 1+\lambda \sbarr \frac{\omega + \lambda \sigma}{1+\lambda} \in \F, \; \sigma \in \F \rsetr.
\end{aligned}\end{equation}
Both of the robustness measures find a plethora of applications in finite-dimensional resource theories, including quantum resource distillation and dilution~\cite{brandao_2008-1,brandao_2010,brandao_2011,liu_2019,regula_2020}, simulation of quantum circuits~\cite{pashayan_2015,howard_2017,seddon_2020}, as well as channel discrimination problems in quantum mechanics~\cite{piani_2016,takagi_2019-2} and broader GPTs~\cite{takagi_2019}. 

However, such results do not immediately generalize to infinite-dimensional resource theories, even within quantum mechanics. The first problem being that, in general, the definitions in Eq.~\eqref{eq:rob_primal_nonlsc}--\eqref{eq:robs_primal_nonlsc} do not even guarantee lower semicontinuity, so the measures defined in this way can fail to satisfy condition (4) that we considered in the previous section. Additionally, the methods required to characterize robustness monotones and their operational applications --- mathematical tools like convex duality, various bounds and relations with other monotones, resource-theoretic properties --- were explicitly developed only in finite-dimensional theories. Their extension to infinite-dimensional spaces is non-trivial due to the considerably richer and more complex structure of infinite-dimensional theories.

Crucially, we will see that the definition of the generalized robustness $\R_\F$ can indeed be suitably adapted to general resource theories in infinite-dimensional GPTs, extending most of the useful properties of the measure familiar from finite-dimensional spaces. The case of the standard robustness $\Rs_\F$ becomes more troublesome. Even in finite dimensions, it is known that there exist resource theories for which the standard robustness is not well-behaved and takes an infinite value for most resourceful states~\cite{napoli_2016, liu_2019, regula_2020}, which prevents it from being a useful quantifier. \txb{In the infinite-dimensional case, the measure suffers from this problem for the theory of nonclassicality, and diverges on several interesting states for entanglement theory as well. Since these are perhaps the two most important continuous-variable resources in quantum mechanics, this suggests that $\Rs_\F$ may not be well suited to be a universal resource quantifier.} Although we will also study the properties of $\Rs_\F$ when relevant, our focus from now on will be mainly on $\R_\F$.


\section{Defining robustness measures in infinite dimensions}\label{sec:defining_robustness}

In our investigation, it will be useful to consider the optimization problems $\R_\F$, $\Rs_\F$ in the form of inequalities with respect to convex cones. To this end, we notice that by defining $\tau' = \lambda \tau \in \C$ in Eq.~\eqref{eq:rob_primal_nonlsc}, we can write
\begin{equation}\begin{aligned}\label{eq:rob_primal_conic_expression}
	\R_\F(\omega) &= \inf \lset 1+\lambda \sbar \omega + \tau' = (1+\lambda)\sigma,\; \tau' \in \C,\; \sigma \in \F \rset\\
	&= \inf \lset \< U, \sigma' \> \sbar \omega \cleq_\C \sigma',\; \sigma' \in \cone(\F) \rset.
\end{aligned}\end{equation}
In a similar way, we write 
\begin{equation}\begin{aligned}
	\Rs_\F(\omega) = \inf \lset \< U, \sigma' \> \sbar \omega \cleq_\F \sigma',\; \sigma' \in \cone(\F) \rset
\end{aligned}\end{equation}
where we use $A \cleq_\F B \iff B-A\in\cone(\F)$ to denote inequality with respect to the conic hull of $\F$. We will take the above to be the definitions of $\R_\F(\tau)$ and $\Rs_\F(\tau)$ when $\tau \in \V$ is not a state. 

In this section, we will use the tools of convex optimization to investigate the properties of these quantities and to define variants of the measures which are more suited to infinite-dimensional spaces. Before progressing with the study of the robustnesses explicitly, we will consider general optimization problems of this type.


\subsection{Convex duality in Banach spaces}\label{sec:optimization_banach}

Let $\K_1 \subseteq \K_0 \subseteq \C$ be two closed convex cones. Define the \textit{primal problem} $P$ as the optimization problem whose optimal value is given by
\begin{equation}\begin{aligned}\label{eq:opt_primal}
  P(\omega) \coloneqq \inf_{\sigma} \lset \< U, \sigma \> \sbar \sigma - \omega \in \K_0,\, \sigma \in \K_1 \rset.
\end{aligned}\end{equation}
The generalized robustness $\R_\F$ is obtained when $\K_0 = \C$, $\K_1 = \cone(\F)$, and the standard robustness $\Rs_\F$ when $\K_0 = \K_1 = \cone(\F)$. We then define the corresponding \textit{dual problem} $D$ as
\begin{equation}\begin{aligned}\label{eq:opt_dual}
  D(\omega) \coloneqq \sup_{W} \lset \< W, \omega \> \sbar W \in \K_0\*,\, U - W \in \K_1\* \rset.
\end{aligned}\end{equation}
Notice that the optimal values are homogeneous functions of $\omega$ (i.e., $P(k \omega) = kP(\omega)$). The primal problem is called feasible if there exists a choice of $\sigma$ which satisfies the constraints in Eq.~\eqref{eq:opt_primal}, and infeasible otherwise; the dual is always feasible since one can choose $W = U$.

A commonly used property of such optimization problems, particularly in finite-dimensional spaces, is the strong duality: under certain conditions, it can be guaranteed that $P(\omega) = D(\omega)$~\cite{rockafellar_1970,boyd_2004}. However, this is generally much more difficult to ensure in infinite-dimensional spaces~\cite{ponstein_2004,borwein_2006}, motivating an alternative approach.

We will say that the primal problem $P$ is \textit{subfeasible}~\cite{duffin_1956} with subfeasible value $\lambda$ if there exists sequences $\{ \sigma_n \}_n \in \K_1, \{ \tau_n \}_n \in \K_0$ such that $\{ \sigma_n - \tau_n \}_n \to \omega$ in the base norm topology, and $\{ \< U, \sigma_n \> \}_n \to \lambda$ in the standard topology on $\RR$. Intuitively, this means that the choices of $\{ \sigma_n, \tau_n \}$ --- although not necessarily feasible for the primal problem $P$ --- approach feasibility. The \textit{optimal subfeasible value} is then defined as
$P'(\omega) \coloneqq \inf \lset \lambda \sbar \lambda \text{ subfeasible for } P(\omega) \rset$.

We now establish a general form of duality between the optimal primal subfeasible value and the optimal dual value. The result is based on general properties of linear and conic optimization, and many formulations of this principle can be found in the literature~\cite{duffin_1956,kretschmer_1961,yamasaki_1968,anderson_1983,shapiro_2001}. We include a self-contained proof in Appendix~\ref{app:duality}.

\begin{pro}\label{prop:banach_states}
For any $\omega \in \Omega$ it holds that
\begin{equation}\begin{aligned}
    P'(\omega) = D(\omega).
\end{aligned}\end{equation}
Furthermore, the primal problem is subfeasible if and only if there exists a dual solution $W$ which achieves the optimal value, with $\< W, \omega \> = D(\omega) = P'(\omega)$.
\end{pro}

In addition to the above duality result, it is of course of interest to ask: under what conditions is full strong duality achieved, that is, when do the values of the primal and dual problems satisfy $P(\omega) = D(\omega)$? There are several ways to establish sufficient conditions for such a property based on so-called generalized interior conditions~\cite{borwein_2003,bot_2008,ponstein_2004}, which require the existence of points in the interior (or various generalized notions thereof) of the cones $\K_0$ or $\K_1$. Such properties are very non-trivial requirements in infinite-dimensional spaces, and they are out of scope of this work.

To derive a simpler condition, often possible to verify in practice, we establish several equivalent formulations of the optimization problems $P, P'$, clarifying the definitions and the convergence requirement of the subfeasible solutions. Since $\K_0, \K_1 \subseteq \C$, each of the cones admits a bounded, closed, convex base $\B_0, \B_1 \subseteq \Omega$ --- in the case of the robustness $\R_\F$, the cones are $\C$ with the base $\Omega$ and $\cone(\F)$ with the base $\F$. We then have the following.
\begin{lem}\label{lem:rob_based}
For any $\omega \in \Omega$ it holds that
\begin{align}
    P(\omega) &= \inf \lset \lambda \sbar \omega = \lambda \sigma - (\lambda - 1) \tau, \; \tau \in \B_0,\; \sigma \in  \B_1\rset \label{eq:rob_primal_base}\\
    P'(\omega) &= \inf \Big\{ \lambda \;\Big|\; \exists \{\xi_n\}_n \to \omega \colon \xi_n = \lambda \sigma_n - (\lambda - 1) \tau_n, \nonumber\\
     & \hphantom{\inf \Big\{ \lambda \;\Big|\;} \tau_n \in \B_0,\; \sigma_n \in \B_1\Big\}.
\end{align}
In particular, it suffices to {look at} sequences of normalized elements $\xi_n \in \V$ such that $\<U, \xi_n\> = 1$ when considering subfeasibility.

Alternatively, we can write
\begin{equation}\begin{aligned}
    P(\omega) &= \inf \lset \lambda \sbar \omega \in \lambda (\B_1 - \K_0) \rset\\
    P'(\omega) &= \inf \lset \lambda \sbar \omega \in \lambda \cl (\B_1- \K_0) \rset.
\end{aligned}\end{equation}
\end{lem}
The proof follows straightforwardly by manipulating the definitions of the optimization problems, and we include it in Appendix~\ref{app:duality} for completeness.

Stemming from the above characterization, we obtain a sufficient condition for strong duality as follows.
\begin{pro}\label{prop:strong_duality_from_closedness}
If either $\conv \!\left( \B_1 \cup (-\B_0) \right)$ or $\B_1 - \K_0$ is a closed set, then $P(\omega) = P'(\omega) = D(\omega)$ holds for all $\omega \in \Omega$.
\end{pro}
\noindent In the case of the robustness $\R_\F$, the condition requires that $\conv(\F\cup(-\Omega))$ be closed.

We will later see that the condition can be used to show strong duality in several relevant cases. In particular, we will use it in Sec.~\ref{sec:strong_duality_resources} to establish the strong duality of the robustness $\R_\F$ in many continuous-variable quantum resources.


\subsection{Lower semicontinuous robustness}

As mentioned before, the definition of the robustness $\R_\F$ is not guaranteed to satisfy lower semicontinuity (l.s.c.) --- a property which reflects the fact that, if a state can be approximated by a sequence of states $\{\omega_n\}$ with $\R_\F(\omega_n) \leq \lambda$, then we should also have $\R_\F(\omega) \leq \lambda$. Mathematically, this corresponds to the fact that each set $\lset \omega \sbar \R_\F(\omega) \leq \lambda \rset$ should be closed.
We thus define the \emph{l.s.c.\ robustness} $\Rl_\F$ as the l.s.c.\ hull (also known as closure) of $\R_\F$, that is, the largest function which is lower semicontinuous and upper bounded by $\R_\F$ (see e.g.~\cite{ponstein_2004}). Specifically, we have the following.
\begin{defn}The \textbf{lower semicontinuous robustness} $\Rl_\F$ is defined as
\begin{equation}\begin{aligned}
    \Rl_\F(\omega) \coloneqq& \liminf_{\xi \to \omega}\, \R_\F(\xi)\\
    =& \lim_{\ve\to 0^+}\, \inf_{\norm{\xi - \omega}_\Omega \leq \ve} \R_\F(\xi).
\end{aligned}\end{equation}
\end{defn}
We will often refer to $\Rl_\F$ simply as the robustness, since we shall see that it is the most natural extension of $\R_\F$ applicable to general infinite-dimensional resource theories and satisfying all desirable properties.

Recall that the value of the robustness can be identified with the optimal value of the optimization problem
\begin{equation}\begin{aligned}\label{eq:robustness_repeat}
	\R_\F(\omega) = \inf \lset \< U, \sigma \> \sbar \omega \cleq_\C \sigma,\; \sigma \in \cone(\F) \rset.
\end{aligned}\end{equation}
Consider then the so-called epigraph of $\R_\F$, defined as the set
\begin{equation}\begin{aligned}
	\operatorname{epi} \R_\F \coloneqq& \lset (\omega, \lambda) \sbar \R_\F(\omega) \leq \lambda \rset \\
	=& \lset (\sigma - \tau,\, \<U, \sigma\>) \sbar \sigma \in \cone(\F),\;\tau \in \C \rset
\end{aligned}\end{equation}
in the space $\V \times \RR$. The l.s.c.\ robustness $\Rl_\F$ can be understood as the function whose epigraph corresponds to the closure of $\operatorname{epi}  \R_\F$~\cite[2.1.3]{barbu_2012}. 
It is then not difficult to show that $\Rl_\F$ is precisely the optimal subfeasible value of the optimization problem~\eqref{eq:robustness_repeat} as we have considered in the previous section. 
This argument allows us to use the duality relation in Prop.~\ref{prop:banach_states} to identify the value of the l.s.c.\ robustness with the optimal value of the dual problem. We thus obtain several equivalent definitions of this measure, which we summarize in the following.

\begin{cor}\label{cor:rob_lsc_equal}
Let $\F \subseteq \Omega$ be any closed and a convex set. For any state $\omega$ it holds that
\begin{equation}\begin{aligned}\label{eq:rob_lsc_definition}
    &\Rl_\F(\omega) = \lim_{\ve\to 0^+}\, \inf_{\norm{\xi - \omega}_\Omega \leq \ve} \R_\F(\xi)\\
    &= \inf_{\{\xi_n\}_{n} \in \V} \lset 1+\lambda \sbar \frac{\xi_n + \lambda \tau_n}{1 + \lambda} \in \F, \; \tau_n \in \Omega,\; \{\xi_n\}_n \to \omega \rset\\
    &= \inf \lset \lambda \sbar \omega \in \cl(\lambda \F - \C) \rset\\
    &= \sup \lset \<W, \omega\> \sbar W \in \C\*,\; \<W, \sigma\> \leq 1 \; \forall \sigma \in \F \rset
\end{aligned}\end{equation}
and the supremum in the last line is achieved whenever ${\Rl_\F(\omega) < \infty}$.
\end{cor}
In other words, the function $\Rl_\F(\omega)$ can be understood in three different ways: (1) as the tightest lower semicontinuous approximation to $\R_\F(\omega)$; (2) as the least value of $\R_\F$ achievable by sequences of points converging to $\omega$; (3) as the optimization over dual witnesses, which generalizes an equivalent definition of finite-dimensional robustness~\cite{brandao_2005,regula_2018,takagi_2019}.

We will show in the following section that $\Rl_\F$ satisfies requirements (1)-(4) from Sec.~\ref{sec:measures_axioms} to be considered a bona fide resource measure. We will thus regard $\Rl_\F(\omega)$ as the proper definition of the robustness in infinite-dimensional GPTs. Within quantum mechanics, we will later show that the functions $\Rl_\F$ and $\R_\F$ are actually equal in many of the practically relevant cases.

One might wonder why extra care with the definition of the robustness is required here, but the distinction between $\Rl_\F$ and $\R_\F$ is not necessary in finite-dimensional theories. This is a consequence of a very strong property of finite-dimensional spaces: all closed sets of states in such theories are compact. Compactness significantly simplifies the considered optimization problems, and in particular we have the following condition for equality of the two definitions.
\begin{pro}\label{prop:compact}
If $\F$ is compact, then $\R_\F(\omega) = \Rl_\F(\omega)$ for all $\omega$.
\end{pro}
\begin{proof}
Recalling our characterization in Lemma~\ref{lem:rob_based} and Prop.~\ref{prop:strong_duality_from_closedness}, we have that
\begin{equation}\begin{aligned}
    \R_\F(\omega) &= \inf \lset \lambda \sbar \omega \in \lambda \F - (\lambda - 1) \Omega \rset,\\
    \Rl_\F(\omega) &= \inf \lset \lambda \sbar \omega \in \cl(\lambda \F - (\lambda - 1) \Omega) \rset.
\end{aligned}\end{equation}
We will then show that $\omega \in \lambda \F - (\lambda - 1) \Omega \iff \omega \in \cl(\lambda \F - (\lambda - 1) \Omega)$ for all $\lambda$. To this end, consider a convergent sequence $\{\lambda \sigma_n - (\lambda - 1) \tau_n\}_n \to \omega$ with $\sigma_n \in \F, \tau_n \in \Omega$. Compactness of $\F$ gives that there exists a subsequence $\{\sigma_{n_s}\}_{n_s}$ which converges to some $\sigma \in \F$. Noting that $\{ \lambda \sigma_{n_s} - (\lambda-1) \tau_{n_s} \}_{n_s}$ is a subsequence of a convergent sequence and therefore must converge to the same limit $\omega$, we have that the sequence
\begin{equation}\begin{aligned}
    \{ \lambda \sigma_{n_s} - (\lambda - 1) \tau_{n_s} \}_{n_s} - \{ \lambda \sigma_{n_s}\}_{n_s} &= -(\lambda -1) \{  \tau_{n_s} \}_{n_s}
\end{aligned}\end{equation}
is a sum of two convergent sequences and thus is convergent itself, hence $\{\tau_{n_s}\}_{n_s} \to \tau \in \Omega$ by the closedness of $\Omega$. Altogether, this gives
\begin{equation}\begin{aligned}
    \{ \lambda \sigma_{n_s} - (\lambda - 1) \tau_{n_s} \}_{n_s} \to \lambda \sigma - (\lambda - 1) \tau
\end{aligned}\end{equation}
and so $\omega = \lambda \sigma - (\lambda - 1) \tau \in \lambda \F - \C$.
\end{proof}

We remark that a lower semicontinuous variant of $\Rs$ can be defined in full analogy with the results above:
\begin{equation}\begin{aligned}
	&\Rsl_\F(\omega) = \lim_{\ve\to 0^+}\, \inf_{\norm{\xi - \omega}_\Omega \leq \ve} \Rs_\F(\xi)\\
    &= \inf_{\{\xi_n\}_n \in \V} \lset 1+\lambda \sbar \frac{\xi_n + \lambda \delta_n}{1 + \lambda} \in \F, \; \delta_n \in \F,\; \{\xi_n\}_n \to \omega \rset\\
    &= \sup \lset \<W, \omega\> \sbar 0 \leq \<W, \sigma\> \leq 1 \; \forall \sigma \in \F \rset.
\end{aligned}\end{equation}
An analogous compactness condition clearly suffices for the equality $\Rs_\F(\omega) = \Rsl_\F(\omega)$ as well.


\section{Robustness as a resource monotone}\label{sec:robustness_properties_monotone}

In order to fully motivate the use of $\Rl_\F$ as a resource quantifier, we now characterize its properties in the axiomatic framework of resource monotones discussed in Sec.~\ref{sec:measures_axioms}.

For completeness, we describe the properties of both the usual definition of the robustness $\R_\F$ as well as the lower semicontinuous robustness $\Rl_\F$.

\begin{lem}\label{thm:rob_measure_properties}
For any closed and convex set $\F \subseteq \Omega$, the robustness $\R_\F$ is:
\begin{enumerate}[(i)]
\item convex,
\item faithful, i.e., any $\omega \in \Omega$ satisfies $\R_\F(\omega) \geq 1$ and $\R_\F(\omega) = 1 \iff \omega \in \F$,
\item monotonic on average under free operations, i.e., for any probabilistic protocol $\{\Phi_i\}_i$ consisting of subchannels which map $\omega$ to $\Phi_i(\omega)$ with corresponding probability $p_i = \Tr \Phi_i(\omega)$ and each $\Phi_i$ satisfying $\Phi_i[\F] \subseteq \cone(\F)$, we have
\begin{equation}\begin{aligned}
    \R_\F(\omega) \geq \sum_i p_i \R_\F\left(\frac{\Phi_i(\omega)}{p_i}\right).
\end{aligned}\end{equation}
\end{enumerate}
\end{lem}
\begin{proof}
Point (i) is obvious from convexity of $\F$. To show faithfulness in point (ii), write
\begin{equation}\begin{aligned}
    \R_\F(\omega) = \inf \lset \< U, \sigma' \> \sbar \omega \cleq_\C \sigma', \; \sigma' \in \cone(\F) \rset.
\end{aligned}\end{equation}
Clearly, if $\omega \in \F$, then $\omega$ itself is a feasible solution and thus $\R_\F(\omega) \leq 1$; the fact that $\R_\F(\omega) \geq \Rl_\F(\omega) \geq 1$ follows by taking $U$ as a feasible solution in the dual~\eqref{eq:rob_lsc_definition}. Conversely, $\R_\F(\omega) = 1$ means that there exists a sequence $\{\sigma'_k\}_k \in \cone(\F)$ satisfying $\sigma'_k - \omega \cgeq_\C 0$ such that
\begin{equation}\begin{aligned}
    \lim_{k\to\infty} \< U, \sigma'_k \> = 1.
\end{aligned}\end{equation}
But since $\< U, \sigma'_k - \omega \> = \< U, \sigma'_k \> - 1$, we have that 
\begin{equation}\begin{aligned}
    0 = \lim_{k \to \infty} \< U, \sigma'_k - \omega \> = \lim_{k \to \infty} \norm{\sigma'_k - \omega}_{\Omega}
\end{aligned}\end{equation}
where the last equality follows since $\sigma'_k - \omega \in \C$ and $\Omega$ is a base for $\C$~\cite{hartkamper_1974} (cf.~\cite[1.37]{lami_2018-1}). We conclude that $\omega \in \cone(\F)$, which can only be the case when $\omega \in \F$.

For (iii), consider any $\sigma' \in \cone(\F)$ such that $\omega \cleq \sigma'$. For any $\Phi_i$ as in the Theorem, we have that
\begin{equation}\begin{aligned}
    \Phi_i(\omega) \cleq \Phi_i(\sigma') \in \cone(\F)
\end{aligned}\end{equation}
due to the fact that each $\Phi_i$ preserves the state cone $\C$. Then,
\begin{equation}\begin{aligned}
    \R_\F\left(\frac{\Phi_i(\omega)}{p_i}\right) \leq \frac{\< U, \Phi_i(\sigma') \>}{p_i}
\end{aligned}\end{equation}
which gives
\begin{equation}\begin{aligned}
    \sum_i p_i \R_\F\left(\frac{\Phi_i(\omega)}{p_i}\right) &\leq \sum_i p_i \left(\frac{\< U, \Phi_i(\sigma') \>}{p_i}\right)\\
    &= \sum_i \< U,  \Phi_i(\sigma') \>\\
    & = \< U, \sigma' \>,
\end{aligned}\end{equation}
where the last equality follows from the fact that $\sum_i \Phi_i$ must be normalization-preserving for the whole protocol to constitute a valid physical transformation. Since this lower bound on $\< U,  \sigma'\>$ holds for any feasible $\sigma'$, it must also hold for the the greatest lower bound $\R_\F(\omega)$, which concludes the proof.
\end{proof}

We proceed to show that the l.s.c.\ robustness $\Rl_\F$ also satisfies the resource-theoretic requirements, in addition to the desirable property of lower semicontinuity.

\begin{thm}\label{thm:rob_lsc_measure_properties}
For any closed and convex set $\F \subseteq \Omega$,  the l.s.c.\ robustness $\Rl_\F$ is:
\begin{enumerate}[(i)]
\item convex,
\item lower semicontinuous, i.e., $\displaystyle \Rl_\F(\omega) \leq \liminf_{n\to\infty} \Rl_\F(\omega_n)$ for any sequence $\{\omega_n\}_n \to \omega$,
\item faithful,
\item monotonic on average under free operations.
\end{enumerate}
\end{thm}
\begin{proof}
(i) Convexity follows from the convexity of $\R_\F$, or can be straightforwardly shown using the dual formulation of $\Rl_\F$. Lower semicontinuity (ii) follows by definition.

To see that faithfulness (iii) holds, notice that the existence of a sequence $\{\sigma_n - \tau_n\}_n \to \omega$ with $\sigma_n \in \F, \tau_n \in \C$ and $\<U, \sigma_n \> \to \<U, \omega \>$ means that $0 = \lim_{n \to \infty} \< U, \tau_n \> = \lim_{n \to \infty} \norm{\tau_n}_\Omega$, so $\{\tau_n\}_n \to 0$. This gives that, in fact, $\{\sigma_n\}_n = \{\sigma_n - \tau_n\} _n + \{\tau_n\}_n \to \omega$. We thus have that $\Rl_\F(\omega) = 1 \Rightarrow \omega \in \F$, and the other direction follows directly from the faithfulness of $\R_\F$.

To establish the strong monotonicity (iv) of $\Rl_\F$, take any feasible $\lambda$ such that $\omega \in \cl(\lambda \F-\C)$, which means that there exists a sequence $\{\lambda \sigma_k - \tau_k\}_k \to \omega$ with $\sigma_k \in \F, \tau_k \in \C$. Now, since $\Phi_i$ preserves the state cone and necessarily satisfies $\< U, \Phi_i(x) \> \leq \< U, x \>$ due to the normalization of $\sum_i \Phi_i$, each $\Phi_i$ cannot increase the base norm. This can be seen by taking any feasible $a_\pm$ in $\norm{x}_\Omega = \inf \lset \< U, a_+ \> + \<U, a_- \> \sbar x = a_+ - a_-,\, a_\pm \in \C \rset$ and noticing that
\begin{equation}\begin{aligned}
    \norm{\Phi_i(x)}_\Omega \leq \< U, \Phi_i (a_+) \> + \< U, \Phi_i (a_-) \> &\leq \< U, a_+ \> + \< U, a_i \>.
\end{aligned}\end{equation}
Then
\begin{equation}\begin{aligned}
    0 &= \lim_{k \to \infty} \norm{ \omega - \lambda \sigma_k + \tau_k }_\Omega\\
    &\geq \lim_{k \to \infty} \norm{ \Phi_i(\omega) - \lambda \Phi_i(\sigma_k) + \Phi_i(\tau_k)}_\Omega\\
    &= \lim_{k \to \infty} \norm{ \Phi_i(\omega) - \lambda \< U, \Phi_i(\sigma_k) \> \frac{\Phi_i(\sigma_k)}{\<U, \Phi_i(\sigma_k) \>} + \Phi_i(\tau_k)}_\Omega\\
    &\geq 0
\end{aligned}\end{equation}
which means that, for each $\Phi_i(\omega)$ there exists a sequence $\{ \lambda'_k \sigma'_k - \tau'_k \}_k$ converging to $\Phi_i(\omega)$ with $\lambda'_k = \lambda \< U, \Phi_i(\sigma_k) \>$.
We then get
\begin{equation}\begin{aligned}
    \sum_i p_i \Rl_\F\left(\frac{\Phi_i (\omega)}{p_i}\right) &= \sum_i \Rl_\F\left(\Phi_i(\omega)\right)\\
    &\leq \sum_i \liminf_{k \to \infty} \Rl_\F(\lambda'_k \sigma'_k - \tau'_k)\\
    &\leq \sum_i \liminf_{k \to \infty} \lambda \< U, \Phi_i(\sigma_k) \>\\
    &\leq \liminf_{k \to \infty} \sum_i \lambda \< U, \Phi_i(\sigma_k) \>\\
    & = \lambda
\end{aligned}\end{equation}
where the first inequality is due to the lower semicontinuity of $\Rl_\F$, the second inequality follows since $\Rl_\F(\lambda'_k \sigma'_k - \tau'_k) \leq \R_\F(\lambda'_k \sigma'_k - \tau'_k) \leq \lambda'_k$, the third inequality is a straightforward consequence of the definition of $\liminf$, and the last equality follows since each $\sigma_k$ is a normalized state and $\sum_i \Phi_i$ preserves normalization.
\end{proof}

We remark that the (strictly weaker) notion of strong monotonicity under selective measurements, often considered in quantum theory~\cite{vidal_2000,chitambar_2019}, can be recovered as a special case of our result by identifying each subchannel $\Phi_i$ with a single Kraus operator $K_i \cdot K_i^\dagger$.


\section{Robustness as the advantage in discrimination tasks}\label{sec:discrimination}

The task of state discrimination forms the foundation of the operational aspects of quantum theory \cite{HOLEVO1973337,Helstrom,kitaev_1997,chefles_2000-1,childs_2000,acin_2001-1,bae_2015,jencova_2014,watrous_2018} as well as GPTs~\cite{ludwig_1985,hartkamper_1974,kimura_2010,bae_2016-1,lami_2017}. The setup of the problem is as follows: after a state is randomly chosen from an ensemble $\{p_i, \omega_i\}$, where $\omega_i \in \Omega$ and each $p_i$ denotes the corresponding probability, the goal is to correctly guess which of the states has been selected by performing a measurement $\{M_i\}$ on the output state. As the figure of merit, we consider the expected probability of successfully distinguishing the states,
\begin{equation}\begin{aligned}
    p_{\mathrm{succ}} (\{p_i, \omega_i\}, \{M_i\}) = \sum_i p_i \< M_i, \omega_i \>.
\end{aligned}\end{equation}
To characterize the optimal guessing probability, the quantity $p_{\mathrm{succ}}$ can then be maximized over all choices of $\{M_i\} \in \M_n$, where $\M_n$ denotes the set of all $n$-effect measurements on the given space. Importantly, state discrimination reveals a fundamental relation between the base norm of the given GPT and state distinguishability ---  for a fixed ensemble of two states, the base norm distance quantifies the maximal probability of successfully discriminating the states~\cite{ludwig_1983,kimura_2010}:
\begin{equation}\begin{aligned}
    \sup_{\mbM \in \M_2} p_{\mathrm{succ}} \left( \{p_i, \omega_i\}_{i=1}^2, \mbM\right) = \frac{1}{2} \big( \!\norm{p_1 \omega_1 - p_2 \omega_2}_\Omega + 1 \big).
\end{aligned}\end{equation}
However, the discrimination of more than two states is much less well understood, and there are no general quantitative results which describe it as concisely.

Channel discrimination is a related task, where the goal is to distinguish between channels sampled from an ensemble $\{p_i, \Lambda_i\}$ with $\Lambda_i : \V \to \V'$. Here, one can use a fixed state $\omega$ as a probe system, which is sent through the randomly chosen channel and afterwards measured with a measurement $\{M_i\}$. Hence, the average probability of success for a given channel ensemble, input state, and output measurement is given by
\begin{equation}\begin{aligned}
	p_{\mathrm{succ}} \left(\{p_i, \Lambda_i\}, \{M_i\} , \omega\right) = \sum_i p_i \< M_i, \Lambda_i(\omega) \>.
\end{aligned}\end{equation}
We remark that this can be regarded as a state discrimination of the ensemble $\{p_i, \Lambda_i(\omega)\}$.
In addition to the fundamental applications of state discrimination tasks, infinite-dimensional channel discrimination finds use in applications such as quantum illumination~\cite{lloyd_2008,tan_2008} and sensing~\cite{pirandola_2018}.

It is therefore important to characterize the advantages achievable in channel discrimination tasks in various constrained settings motivated by different physical considerations. The resource-theoretic framework lends itself well to such investigations, and we will thus ask a general and broadly applicable question: \textit{how much} better can we discriminate channels by using a resourceful probe state $\omega \notin \F$, compared to a restricted setting in which we are limited to using only the free states $\sigma \in \F$?

We will consider the discrimination of channels between two GPTs in spaces $\V$ and $\V'$, each endowed with a positive cone $\C, \C'$, respectively, and a relevant closed and convex set of free states $\F \subseteq \Omega$, $\F' \subseteq \Omega'$. Let $\M_n$ denote the set of $n$-effect measurements $\{M_i\}_{i=1}^n \subset \V'\*$, and $\T_n$ denote the set of all $n$-element channel ensembles of the form $\{p_i, \Lambda_i\}_{i=1}^{n}$ with $\sum_i p_i = 1$. Given any choice of an ensemble $\A \in \T_n$ and an output measurement $\mbM \in \M_n$, the quantity that we focus on is the advantage provided by a given state over resourceless states. We quantify this through the ratio
\begin{equation}\begin{aligned}
	\frac{p_{\mathrm{succ}} (\A, \mbM, \omega)}{\sup_{\sigma \in \F}  p_{\mathrm{succ}} (\A, \mbM, \sigma)}.
\end{aligned}\end{equation}
Note here that, in order to single out the advantage provided by the state $\omega$, we require that the same measurement $\mbM$ be used in both cases. 
We then obtain a characterization of the l.s.c.\ robustness as follows.

\begin{thm}\label{thm:robustness_discrimination}
For any $\omega \in \Omega$ and any closed and convex set $\F$ it holds that
\begin{equation}\begin{aligned}
  \sup_{\substack{\mbM \in \M_n\\\A \in \T_n\\n \in \mathbb{N}}} \frac{p_{\mathrm{succ}} (\A, \mbM, \omega)}{\sup_{\sigma \in \F}  p_{\mathrm{succ}} (\A, \mbM, \sigma)} = \Rl_\F(\omega),
\end{aligned}
\label{eq:robustness_discrimination}
\end{equation}
where the maximization is over all possible channel discrimination tasks. If $\Rl_\F(\omega) < \infty$, there exists a choice of $\mbM \in \M_2$, $\A \in \T_2$ which achieves the outermost supremum on the left-hand side.
\end{thm}
\begin{proof}
Consider any sequence {$\{\lambda \sigma_k - \tau_k\}_k \to \omega$ such that $\sigma_k \in \F, \tau_k \in \C$.} Then, for any $n \in \mathbb{N}$, any $\mbM \in \M_n$, and any $\A \in \T_n$, it holds that
\begin{equation}\begin{aligned}
    p_{\mathrm{succ}} (\A, \mbM, \omega) &= \sum_{i=1}^n p_i \< M_i, \Lambda_i (\omega) \>\\
    &= \lim_{k \to \infty} \sum_{i=1}^n p_i \< M_i, \Lambda_i (\lambda \sigma_k - \tau_k) \>\\
    &\leq \limsup_{k \to \infty} \sum_{i=1}^n p_i \< M_i, \Lambda_i (\lambda \sigma_k) \>\\
    &\leq \lambda \sup_{\sigma' \in \F} \sum_{i=1}^n p_i \< M_i, \Lambda_i (\sigma') \>\\
    &= \lambda \sup_{\sigma' \in \F} p_{\mathrm{succ}} (\A, \mbM, \sigma')
\end{aligned}\end{equation}
where in the second line we used the continuity of $p_{\mathrm{succ}} (\A, \mbM, \cdot)$ for fixed $\A$ and $\mbM$, and in the third line we used that $\omega \in \C \Rightarrow \Lambda_i(\omega) \in \C'$ and each $M_i \in \C'\*$, hence $\< M_i, \Lambda_i (\lambda \sigma_k - \tau_k) \> \leq \<M_i, \Lambda_i(\lambda \sigma_k) \> \, \forall k$. Since this holds for any feasible $\lambda$, we have
\begin{equation}\begin{aligned}
    &p_{\mathrm{succ}} (\A, \mbM, \omega) \\
    &\qquad \leq \inf \lset \lambda \sup_{\sigma' \in \F}  p_{\mathrm{succ}} (\A, \mbM, \sigma') \sbar \omega \in \cl(\lambda \F - \C) \rset\\
    &\qquad = \Rl_\F(\omega) \, \sup_{\sigma' \in \F} p_{\mathrm{succ}} (\A, \mbM, \sigma')
\end{aligned}\end{equation}
by definition of the robustness. This gives an upper bound on $\displaystyle \frac{p_{\mathrm{succ}} (\A, \mbM, \omega)}{\sup_{\sigma' \in \F} p_{\mathrm{succ}} (\A, \mbM, \sigma')}$ for any $n$, $\mbM$, and $\A$, which means that the least upper bound obeys
\begin{equation}\begin{aligned}
    \sup_{\substack{\mbM \in \M_n\\\A \in \T_n\\n \in \mathbb{N}}} \frac{p_{\mathrm{succ}} (\A, \mbM, \omega)}{\sup_{\sigma' \in \F} p_{\mathrm{succ}} (\A, \mbM, \sigma')} \leq \Rl_\F(\omega).
\end{aligned}\end{equation}

For the other inequality, let $W \in \V\*$ be any feasible dual solution such that $W \in \C\*$ and $\<W, \sigma\> \leq 1 \; \forall \sigma \in \F$. Consider then the measurement $\mbM' \in \M_2$ defined as $\{ W / \norm{W}_\Omega^\circ, \, U - W / \norm{W}_\Omega^\circ \}$  and the channel ensemble $\A'$ given by the probability distribution $\{ 1, 0 \}$ and channels $\{ \mathrm{id}, \Lambda' \}$ with $\Lambda'$ arbitrary. This gives
\begin{equation}\begin{aligned}
    \sup_{\substack{\mbM \in \M_n\\\A \in \T_n\\n \in \mathbb{N}}} \frac{p_{\mathrm{succ}} (\A, \mbM, \omega)}{\sup_{\sigma \in \F}  p_{\mathrm{succ}} (\A, \mbM, \sigma)} &\geq \frac{p_{\mathrm{succ}} (\A', \mbM', \omega)}{\sup_{\sigma \in \F}  p_{\mathrm{succ}} (\A', \mbM', \sigma)}\\
    &= \frac{ \frac{1}{\norm{W}_\Omega^\circ} \< W, \omega \>}{\frac{1}{\norm{W}_\Omega^\circ} \sup_{\sigma \in \F} \< W, \sigma \>}\\
    &\geq \< W, \omega \>.
\end{aligned}\end{equation}
Optimizing over all feasible $W$, we have that the supremum on the left-hand side must equal $ \Rl_\F(\omega)$. Furthermore, from Prop.~\ref{prop:banach_states} we know that when $\Rl_\F(\omega) < \infty$, there exists an optimal choice of $W$ such that $\< W, \omega \> = \Rl_\F(\omega)$, which allows us to construct the measurement $\mbM'$ achieving this value.
\end{proof}

The result establishes the l.s.c.\ generalized robustness $\Rl_\F$ as a precise quantifier of the advantage provided by a given state over all resourceless states $\sigma \in \F$. This gives the measure a direct operational application, elevating $\Rl_\F$ from a monotone defined through geometric considerations to a physically meaningful quantity, and revealing a general connection between discrimination tasks and resource quantifiers based on conic optimization problems. This not only extends the previous relations of this type found in general finite-dimensional resource theories of states~\cite{takagi_2019-2,takagi_2019}, but adds to an increasing list of discrimination settings in which robustness-based quantifiers play a vital role~\cite{piani_2015,piani_2016,takagi_2019,uola_2019-1,uola_2019,yuan_2020,skrzypczyk_2019-1,oszmaniec_2019-1}.

Several extensions of the result in Thm.~\ref{thm:robustness_discrimination} can be considered. In some settings in quantum mechanics, entanglement can be employed to increase the probability of success in discrimination tasks --- by only sending part of an entangled probe system $\omega$ through the channel, the correlations at the output can be exploited to distinguish channels more effectively~\cite{kitaev_1997}. In order to single out the resource theory corresponding to the set $\F$ as the source of all advantages, we have not considered possible trade-offs between entanglement and other resources here.

Furthermore, the advantage provided by a given state $\omega$ could change if one allows different measurements to be used in the discrimination of $\{\Lambda_i(\omega)\}$ and in the discrimination of $\{\Lambda_i(\sigma)\}$. Again, we do not consider this modified setting here explicitly in order to maintain full generality of our results.
\txb{To illustrate why allowing different measurements in \eqref{eq:robustness_discrimination} could lead to significantly different behavior, consider a state $\omega$ for which $2<\Rl_\F(\omega)<\infty$. By Theorem~\ref{thm:robustness_discrimination}, there exist $\A \in \T_2$ and $\mbM \in \M_2$ satisfying that
\begin{equation}
	\frac{p_{\mathrm{succ}} (\A, \mbM, \omega)}{\sup_{\sigma \in \F}  p_{\mathrm{succ}} (\A, \mbM, \sigma)} > 2.
\end{equation}
Since $p_{\mathrm{succ}} (\A, \mbM, \omega)\leq 1$ always holds, it follows that $\sup_{\sigma \in \F}  p_{\mathrm{succ}} (\A, \mbM, \sigma)< 1/2$. But a success probability of $1/2$ can always be achieved with a maximum likelihood guess. What happens is that the guessing strategy $\mbM$ is so poor --- when applied to free states --- that even random guessing outperforms it. In this case, it would therefore be more natural to replace the $\mbM$ in the denominator on the left-hand side of \eqref{eq:robustness_discrimination} with a new variable $\mbM'$. Rather than compromising the applicability of Theorem~\ref{thm:robustness_discrimination}, this argument shows that it expresses the best of itself when the channel ensemble has more than $1/\Rl_\F(\omega)$ elements. Luckily enough, this is often the case in applications --- the interested Reader can find an example showcasing this at the end of Section~\ref{subsubsec:squeezed_states}.
}

Interestingly, for some finite-dimensional resources, the robustness measure $\R_\F$ was shown to remain the figure of merit --- that is, the maximal advantage provided by a given state --- in both of the extended settings that we described above. This holds in particular for discrete-variable quantum entanglement~\cite{takagi_2019-2,bae_2019}. It would be an interesting question to investigate cases where this can be established also in the infinite-dimensional regime.


\section{Quantum mechanics}\label{sec:quantum}

We will now fix a separable Hilbert space $\H$. Let $\T(\H)$ be the Banach space of trace-class operators on $\H$, and $\B(\H)$ its continuous dual, the Banach space of bounded linear operators on $\H$. Let $\D(\H) \subseteq \T(\H)$ be the subset of density operators. Here the cone $\C$ corresponds to the cone of positive operators in $\T(\H)$, while $\C\*$ is the cone of positive operators in $\B(\H)$; we will use $\cgeq$ to denote inequality with respect to either of these. The base norm is given by the trace norm $\norm{\cdot}_1$. We use the notation $\<A,B\>$ for the Hilbert-Schmidt inner product $\Tr(A^\dagger B)$. Given a state $\ket\psi \in \H$, we will write $\psi$ for $\proj{\psi}$.

This section complements the dedicated paper~\cite{our_main}, in which we focus on the case of continuous-variable quantum resources. In particular, the discussion below constitutes a technical companion to Ref.~\cite{our_main}, providing a detailed derivation and several extensions of the results mentioned there. Since all of our results for general GTPs immediately apply to quantum mechanics as a special case, the results of previous sections already serve as proofs of some of the results of~\cite{our_main}. Specifically, Thm.~1 in \cite{our_main} is Thm.~\ref{thm:robustness_discrimination} here; Thm.~2 in \cite{our_main} is Thm.~\ref{thm:rob_lsc_measure_properties} here; Thm.~3 in \cite{our_main} is a consequence of our characterization in Sec.~\ref{sec:defining_robustness} and in particular the duality result in Prop.~\ref{prop:banach_states}. The other results will be established below.


\subsection{Strong duality}\label{sec:strong_duality_resources}

Returning to our discussion in Sec.~\ref{sec:optimization_banach}, it is of interest to ask when strong duality holds, that is, when the two optimization problems $\R_\F(\omega)$ and $\Rl_\F(\omega)$ have the same optimal value. This would allow one to simplify the application of the robustness, as the optimization over sequences $\{\xi_n\}_n \to \omega$ in the definition of $\Rl_\F$ (see Cor.~\ref{cor:rob_lsc_equal}) would no longer be necessary.

First, we remark that insight from finite-dimensional spaces does not easily generalize here. In finite-dimensional theories, the existence of a single state $\sigma \in \sint(\C) \cap \F$ is a necessary and sufficient condition to ensure that the robustness $\R_\F$ remains finite for all states~\cite{datta_2009}, and indeed also ensures the strong duality of the measure~\cite{takagi_2019}. In infinite-dimensional spaces, the conditions become more complicated, and one can easily see that an analogous requirement cannot be fulfilled.
\begin{lem}
If $\H$ is infinite dimensional, there is no state $\sigma$ for which it holds that $\forall \rho \in \D(\H) \; \exists \, \lambda \geq 0 \text{ s.t. } \rho \cleq \lambda \sigma$.
\end{lem}
\begin{proof}
This is a consequence of the fact that a choice of $\lambda \in \RR_+$ such that $\rho \cleq \lambda \sigma$ exists if and only if $\ran(\sqrt{\rho}) \subseteq \ran(\sqrt{\sigma})$~\cite{douglas_1966}, but since each $\sqrt{\sigma}$ is a compact operator, its range cannot be the whole space $\H$~\cite[p. 177]{conway_1985}.

To see this more explicitly, assume such a state $\sigma$ exists. Clearly, its support must be the whole $\H$. Let $\sigma=\sum_{n=0}^\infty p_n \ketbra{e_n}{e_n}$ be its spectral decomposition, with $p_n>0$ for all $n$ and $\{\ket{e_n}\}_n$ an orthonormal set. Set $\ket{\psi}\coloneqq \sum_n \sqrt{p_n}\ket{e_n}$, and for $N\in \NN$ consider the truncated states $\ket{\phi_N}\coloneqq \left( \sum_{n=0}^N p_n^{-1} \right)^{-1/2} \sum_{n=0}^N p_n^{-1/2} \ket{e_n}$. If $\ketbra{\psi}{\psi} \cleq \lambda \sigma$ were to hold for some real $\lambda$, then evaluating this on $\ket{\phi_N}$ would yield $N\leq \lambda$. Since this would need to hold for all $N\in \NN$, we arrive at a contradiction.
\end{proof}
As a consequence, one might expect the robustness $\Rl_\F$ to be infinite for some states --- indeed, the existence of infinitely resourceful states is often a natural property for infinite-dimensional systems~\cite{eisert_2002,keyl_2003}, so it is not a surprising fact.

The non-existence of interior points of $\C$ motivates us to study alternative approaches to strong duality. In order to establish a useful sufficient condition for this property, recall that the space $\T(\H)$ of trace-class operators can be regarded as the Banach dual space~\cite[Definition~1.10.1]{megginson_1998} of $\K(\H)$, the space of compact operators on $\H$. Regarding the spaces as the dual pair $\sigma(\K(\H), \T(\H))$, trace-class operators can then be endowed with the weak* topology induced by the seminorms of the form $\left|\< \cdot, K \>\right|$ for all $K \in \K(\H)$~\cite[Definition~2.6.1]{megginson_1998}. This gives us the following condition.
\begin{thm}\label{thm:strong_duality_from_weakstar}
If the conic hull $\cone(\F)$ is closed in the weak* topology, then strong duality holds and we have $\R_\F(\rho) = \Rl_\F(\rho)$ and $\Rs_\F(\rho) = \Rsl_\F(\rho)$ for all $\rho \in \D(\H)$.
\end{thm}
\begin{proof}
We will show that $\B \coloneqq \conv(\F\cup(-\D(\H)))$ is closed in the trace norm topology and invoke Prop.~\ref{prop:strong_duality_from_closedness} to prove that strong duality for $R_\F$ is satisfied. To this end, define the subnormalized sets $\F_\leq \coloneqq \lset \lambda \sigma \sbar \lambda \in [0,1],\; \sigma \in \F \rset$ and $\D_\leq \coloneqq \lset \lambda \rho \sbar \lambda \in [0,1],\; \rho \in \D(\H) \rset$. Recall that the unit ball $\U = \lset X \in \T(\H) \sbar \norm{X}_1 \leq 1 \rset$ is weak*-compact by the Banach-Alaoglu theorem~\cite[Theorem~2.6.18]{megginson_1998}. Noting that $\cone(\F)$ is weak*-closed by assumption and $\C = \cone(\D(\H))$ is readily verified to be weak*-closed, we have that $\F_\leq = \cone(\F) \cap \U$ and $\D_\leq = \C \cap \U$ are both intersections of a weak*-compact and a weak*-closed set, meaning that they are both weak*-compact. But then $\B = \conv(\F_\leq \cup (-\D_\leq))$ is the convex hull of a union of two convex, weak*-compact sets, and is thus weak*-compact itself, which in particular implies that it is weak*-closed. As the weak* topology is coarser than the norm topology on $\T(\H)$, we have that $\B$ is norm-closed as desired.

An analogous statement holds for $\conv(\F\cup(-\F))$, showing that also $\Rs_\F$ satisfies strong duality.
\end{proof}

Using the above result and our previously established conditions, we will see that most resources of practical relevance do indeed satisfy strong duality. Examples are:
\begin{enumerate}[(i)]
\item Nonclassicality theory (see Sec.~\ref{sec:nonclassicality}). Here, the cone $\cone(\F)$ was recently shown to be weak*-closed in Ref.~\cite{relentNC}, so we immediately get strong duality using Thm.~\ref{thm:strong_duality_from_weakstar}. 
\item Entanglement theory. We will explicitly show this using Thm.~\ref{thm:strong_duality_from_weakstar} in Sec.~\ref{sec:entanglement}.
\item Coherence theory. We will explicitly show this using Thm.~\ref{thm:strong_duality_from_weakstar} in Sec.~\ref{sec:coherence}.
\item Athermality (thermodynamics). Here, the set $\F$ consists of a single state (the Gibbs state), which is clearly compact in any of the considered topologies, with strong duality following by Prop.~\ref{prop:compact} or Thm.\ref{thm:strong_duality_from_weakstar}.
\item Energy-constrained resource theories. Specifically, for a general closed and convex set $\F \subseteq \D(\H)$, one can define the restriction $\F' = \lset \sigma \in \F \sbar \<H, \sigma\> \leq E \rset$ where $H$ is an unbounded positive self-adjoint operator with a discrete spectrum of finite multiplicity. Any such set $\F'$ is in fact compact~\cite{holevo_2003}, hence strong duality for $\R_{\F'}$ is ensured by Prop.~\ref{prop:compact}.
\end{enumerate}

The constraints on the energy of the system as in (v) are a common way to avoid discontinuities in infinite dimensions~\cite{wehrl_1978,eisert_2002,holevo_2003,winter_2016-1}. The assumptions on the allowed Hamiltonian are standard --- any Hamiltonian obeying the so-called Gibbs hypothesis~\cite{winter_2016-1}, \txb{i.e., such that $\Tr e^{-\beta H}<\infty$ for all $\beta>0$,} will indeed be an unbounded operator with a discrete spectrum of finite multiplicity.

A complete characterization of resources for which $\R_\F(\rho) = \Rl_\F(\rho)$ is an interesting question which we leave for future work.


\subsection{General bounds for the robustness}

We will now present a couple of handy results that will allow us to compute the robustness of pure states.

\begin{lem}\label{lem:robustness_pure}
Let $\ket{\psi}\in \H$ and an arbitrary set of free states $\F$, it holds that
\begin{equation}
    \R_\F(\psi)  = \inf_{\sigma\in \F} \braket{\psi|\sigma^{-1}|\psi} ,
    \label{rob_pure_alt}
\end{equation}
where $\braket{\psi|\sigma^{-1}|\psi} = \left\|\sigma^{-1/2} \ket{\psi}\right\|^2$, and the infimum is restricted to those free states $\sigma$ such that $\ket{\psi}\in \mathrm{dom}\left(\sigma^{-1/2}\right) = \ran\left(\sigma^{1/2}\right)$, with $\mathrm{dom}$ denoting the domain.
\end{lem}

\begin{proof}
We employ the characterization in the second line of \eqref{eq:rob_primal_conic_expression}. For a fixed state $\sigma$, it is easy to see that $\min\left\{ \lambda: \ketbra{\psi}{\psi}\leq \lambda \sigma \right\} = \braket{\psi|\sigma^{-1}|\psi}$. In fact, on the one hand, by the Cauchy--Schwarz inequality for every vector $\ket{x}$ we have that
\begin{equation}\begin{aligned}
    \left|\braket{x|\psi}\right|^2 &= \left|\braket{x|\sigma^{1/2}\sigma^{-1/2}|\psi}\right|^2 \\&\leq \left\|\sigma^{1/2} \ket{x}\right\|^2 \left\| \sigma^{-1/2}\ket{\psi}\right\|^2\\& = \braket{\psi|\sigma^{-1}|\psi} \braket{x|\sigma|x} \\&= \braket{x|\left( \braket{\psi|\sigma^{-1}|\psi} \sigma \right) |x} ,
\end{aligned}\end{equation}
implying that $\ketbra{\psi}{\psi}\leq \braket{\psi|\sigma^{-1}|\psi} \sigma$. On the other hand, diagonalize $\sigma$ as $\sigma = \sum_{n=0}^\infty p_n \ketbra{e_n}{e_n}$, where $p_n>0$ for all $n$ without loss of generality, and set $\ket{\phi_N} \coloneqq \left(\sum_{n=0}^{N-1} p_n^{-1} \ketbra{e_n}{e_n}\right) \ket{\psi}$. Let $\lambda$ be such that $\ketbra{\psi}{\psi}\leq \lambda \sigma$. Taking the overlap of both sides with the vector $\ket{\phi_N}$ yields 
\begin{equation}\begin{aligned}
\braket{\psi| \left(\sum_{n=0}^{N-1} p_n^{-1} \ketbra{e_n}{e_n}\right) |\psi}^2 &= \left|\braket{\phi_N|\psi}\right|^2\\& \leq \lambda \braket{\phi_N|\sigma|\phi_N} \\&= \lambda \braket{\psi| \left(\sum_{n=0}^{N-1} p_n^{-1} \ketbra{e_n}{e_n}\right) |\psi} .
\end{aligned}\end{equation}
Note that since $\ket{\psi}\in \mathrm{dom}\left(\sigma^{-1/2}\right)$ we must have that $\braket{e_n|\psi}\neq 0$ for some $n$, in turn implying that $\braket{\psi| \left(\sum_{n=0}^{N-1} p_n^{-1} \ketbra{e_n}{e_n}\right) |\psi}>0$. By simplifying we obtain that $\lambda \geq \braket{\psi| \left(\sum_{n=0}^{N-1} p_n^{-1} \ketbra{e_n}{e_n}\right) |\psi} = \sum_{n=0}^{N-1} p_n^{-1} \left|\braket{\psi|e_n}\right|^2$. Taking the limit $N\to\infty$ yields $\lambda \geq \braket{\psi|\sigma^{-1}|\psi}$, completing the proof.
\end{proof}

\begin{lem}\label{lem:robustness_lowerbound}
For any states $\rho, \omega \in \D(\H)$, it holds that
\begin{equation}\begin{aligned} \label{rob_lower_bound}
\Rl_\F(\rho) \geq \frac{\< \rho, \omega \>}{\sup_{\sigma \in \F} \<\sigma, \omega \>}.
\end{aligned}\end{equation}
In particular,
\begin{equation}\begin{aligned} \label{rob_pure_lower_bound}
    \Rl_\F(\psi) \geq \frac{1}{\sup_{\sigma \in \F} \cbraket{\psi|\sigma|\psi}^2}
\end{aligned}\end{equation}
for any pure state $\ket{\psi}\in \H$.
\end{lem}

\begin{proof}
The first claim is easily proved by noting that any $\frac{\omega}{\sup_{\sigma \in \F} \<\sigma, \omega \>}$ is a feasible dual solution in \eqref{eq:rob_lsc_definition}. The second claim follows by taking $\omega = \ketbra{\psi}{\psi}$.
\end{proof}

An interesting property of the pure-state bound in Eq.~\eqref{rob_pure_lower_bound} is that it can relate the robustness $\R_\F$ to the minimal relative entropy distance from the set of free states. Indeed, the robustness always provides an upper bound to the minimal relative entropy as~\cite{datta_2009,berta_2015}
\begin{equation}\begin{aligned}
\inf_{\sigma \in \F} D(\rho \| \sigma) \leq \log \R_\F(\rho),
\end{aligned}\end{equation}
where $D(\rho\|\sigma) = \< \rho, \log\rho - \log\sigma \>$ is the quantum (Umegaki) relative entropy. In certain cases, this bound can be tight.
\begin{cor}\label{cor:relative_entropy}
Whenever $ \R_\F(\psi)$ equals $\displaystyle\frac{1}{\sup_{\sigma \in \F} \cbraket{\psi|\sigma|\psi}^2}$, 
it holds that $\displaystyle \inf_{\sigma \in \F} D(\psi \| \sigma) = \log \R_\F(\psi)$.
\end{cor}
\begin{proof}
The result follows since the condition is equivalent to
\begin{equation}\begin{aligned}
   \inf_{\sigma \in \F} \log \frac{1}{\<\psi, \sigma \>} \eqqcolon \inf_{\sigma \in \F} D_{\min} (\psi\| \sigma) &= \inf_{\sigma \in \F} D_{\max} (\psi \| \sigma)\\ &\coloneqq \log \R_\F(\psi),
\end{aligned}\end{equation}
and $D_{\min}(\psi \| \cdot)$ (R\'{e}nyi relative entropy of order 0) and $D_{\max}(\psi \| \cdot)$ (sandwiched R\'{e}nyi relative entropy of order $\infty$) provide, respectively, a lower and upper bound to the Umegaki relative entropy~\cite{datta_2009}.\end{proof}

In finite-dimensional spaces, the condition of the Corollary is always satisfied for maximally resourceful states in any convex resource theory~\cite{liu_2019,regula_2020}, and can also hold for broader classes of states obeying certain symmetries~\cite{bravyi_2019,seddon_2020}. Although it is unclear if states of this kind can always be found in infinite-dimensional theories, we will find such examples for relevant resources.

\subsection{Robustness and seminorms}\label{sec:seminorms}

Many resources of practical relevance have a structure defined by free \textit{pure} states and convex combinations thereof. Examples include the resource theory of entanglement, nonclassicality, non-Gaussianity, and coherence. Although the results we considered above apply to more general theories, it will be useful to study additional properties that this class of resources enjoys.

We then consider a norm-closed set $\f \subseteq \H$ of free pure states in the underlying Hilbert space, and define the set of free states $\F_\f$ to consist of convex combinations of such free states:
\begin{equation}\begin{aligned}
    \F_\f \coloneqq \cl \conv \lset \proj{v} \sbar \ket{v} \in \f \rset.
\label{eq:free states closure}
\end{aligned}\end{equation}
Here, the closure can be taken either in the weak topology or the norm topology, as the two coincide for convex sets. This definition is equivalent to~\cite{holevo_2005}
\begin{equation}\begin{aligned}
    \F_\f = \lset \int_{\f} \proj{v}\, \mathrm{d}\mu(\proj{v}) \sbar \mu \in P(\f) \rset.
\label{eq:free states integral}
\end{aligned}\end{equation}
where $P(\f)$ is the set of Borel probability measures supported on the closed set $\lset \proj{v} \sbar \ket{v} \in \f \rset$.
\begin{remark}
Taking the closure is, in general, necessary. In particular, when $\f$ denotes the set of product states on the tensor product of two infinite-dimensional Hilbert spaces, Holevo, Shirokov, and Werner~\cite{holevo_2005} showed the existence of states $\sigma \in \F_\f$ which cannot be written as a convex combination $\sum_{k=1}^{r} p_k \proj{v_k}$ with $\sum_{k=1}^r p_k = 1$ and $\ket{v_k} \in \f$, even when one allows infinite discrete combinations with $r = \infty$.
\end{remark}


We will assume without loss of generality that $\f$ is balanced, i.e., $\ket{v} \in \f \Rightarrow \lambda \ket{v} \in \f$ for any $\lambda \in \mathbb{C}$ s.t. $|\lambda| = 1$. This condition can be easily fulfilled for any non-balanced set $\f'$ by defining $\f = \lset \lambda \ket{v} \sbar |\lambda| = 1,\; \ket{v} \in \f' \rset$, since the corresponding sets $\F_\f$ are the same.

Given any such set, define the seminorm $\norm{\cdot}_\f^\circ \colon \H \to \mathbb{R}_+$ by
\begin{equation}\begin{aligned}
    \norm{\ket{x}}_\f^\circ \coloneqq \sup \lset \left|\braket{x | v}\right| \sbar \ket{v} \in \f \rset,
\end{aligned}\end{equation}
and $\norm{\cdot}_\f \colon \H \to \mathbb{R}_+ \cup \{\infty\}$ by
\begin{equation}\begin{aligned}\label{eq:arveson_seminorm}
    \norm{\ket{y}}_\f \coloneqq \sup \lset \left|\braket{y | x}\right| \sbar \norm{\ket{x}}_\f^\circ \leq 1 \rset.
\end{aligned}\end{equation}
Arveson~\cite{arveson_2009} characterized this function and explicitly showed that it is convex, absolutely homogeneous, and lower semicontinuous. In fact, it corresponds to the gauge function (Minkowski functional) of the set $\cl \conv \f$, i.e.
\begin{equation}\begin{aligned}\label{eq:norm_minkowski}
    \norm{\ket{y}}_\f = \inf \lset \mu > 0 \sbar \ket{y} \in \mu \cl \conv \f \rset.
\end{aligned}\end{equation}
However, although we follow \cite{arveson_2009} in using norm-like notation for both functions, $\norm{\cdot}_\f$ defines a valid norm on $\H$ only if $\cl \conv \f$ has nonempty interior, which is generally not the case. Indeed, we will later see that $\norm{\cdot}_\f$ can be infinite in some cases.

Ref.~\cite{arveson_2009} then extends $\norm{\cdot}_\f^\circ$ to a seminorm on $\B(\H)$:
\begin{equation}\begin{aligned}
    \norm{Z}_{\F_\f}^\circ \coloneqq \sup \lset \left|\braket{v | Z | u}\right| \sbar \ket{v}, \ket{u} \in \f \rset.
\end{aligned}\end{equation}
With a slight modification of what was done in~\cite{arveson_2009}, we define a function $\T(\H) \to \RR_+ \cup \{\infty\}$ as
\begin{equation}\begin{aligned}
    \norm{X}_{\F_\f} \coloneqq \sup \lset \left|\< Z, X \>\right| \sbar \norm{Z}_{\F_\f}^\circ \leq 1 \rset,
\end{aligned}\end{equation}
where the optimization is over operators $Z \in \B(\H)$. The optimization can be restricted to self-adjoint $Z$ whenever $X$ is self-adjoint.

We will now relate this function with the robustness.
\begin{pro}\label{thm:robustness_pure_norm}
For any $\rho \in \D(\H)$, it holds that
\begin{equation}\begin{aligned}
    \Rl_\F(\rho)  \leq \norm{\rho}_{\F_\f}.
\end{aligned}\end{equation}
For any rank-one state $\psi = \proj{\psi}$, it further holds that
\begin{equation}\begin{aligned}
    \Rl_\F(\psi)  = \norm{\psi}_{\F_\f} = \norm{\ket{\psi}}_\f^2.
\label{eq:robustness_pure_norm}
\end{aligned}\end{equation}
\end{pro}
\begin{proof}
We begin by noticing that, for any positive self-adjoint operator $Z \in \B(\H)$ and pure states $\ket{u}, \ket{v}$, the Cauchy-Schwarz inequality for the inner product on $\H$ gives
\begin{equation}\begin{aligned}
    \left|\braket{u | Z | v}\right|^2 &= \left|\< \sqrt{Z} \ket{u}, \sqrt{Z} \ket{v} \>\right|^2\\
    &\leq \norm{\sqrt{Z} \ket{u}}^2 \norm{\sqrt{Z} \ket{v}}^2\\
    &= \braket{u | Z | u} \braket{v | Z | v}
\end{aligned}\end{equation}
where $\sqrt{Z}$ is the positive square root of $Z$. This implies that, for any $Z \cgeq 0$, we get
\begin{equation}\begin{aligned}
    \norm{Z}_{\F_\f}^\circ = \sup \lset \left|\braket{v | Z | v}\right| \sbar \ket{v} \in \f \rset
\end{aligned}\end{equation}
since optimizing over $\ket{u}, \ket{v} \in \f$ cannot achieve a higher value of $\left|\braket{u | Z | v}\right|$. This gives
\begin{equation}\begin{aligned}\label{eq:rob_norm_ineq}
    \norm{\rho}_{\F_\f} &= \sup \lset \left|\< Z, \rho \>\right| \sbar \norm{Z}_\F^\circ \leq 1 \rset\\
    &\geq \sup \lset \left|\< Z, \rho \>\right| \sbar Z\cgeq 0,\, \norm{Z}_\F^\circ \leq 1 \rset\\
    &=\sup \lset \< Z, \rho \> \sbar Z\cgeq 0,\; \braket{v | Z | v} \leq 1 \; \forall \ket{v} \in \f \rset\\
    &=\Rl_\F(\rho)
\end{aligned}\end{equation}
where the last equality follows from the dual formulation of the robustness (see Cor.~\ref{cor:rob_lsc_equal}).

Using an argument based on \cite[7.2]{arveson_2009}, we can then relate $\norm{\psi}_{\F_\f}$ with $\norm{\ket{\psi}}_\f^2$ for any pure state. Notice that for any $\ket\psi$ we have $\ket\psi / \norm{\ket\psi}_\f \in \cl \conv \f$, so for all $Z \in \B(\H)$ it holds that
\begin{equation}\begin{aligned}
\left| \frac{\bra{\psi}}{\norm{\ket\psi}_\f} Z \frac{\ket{\psi}}{\norm{\ket\psi}_\f} \right| \leq \norm{Z}_{\F_\f}^\circ.
\end{aligned}\end{equation}
This then implies that
\begin{equation}\begin{aligned}
\norm{\psi}_{\F_\f} = \sup_{\norm{Z}_{\F_\f}^\circ \leq 1} \cbraket{\psi|Z|\psi} \leq \norm{\ket\psi}^2_\f.
\end{aligned}\end{equation}
Together with Eq.~\eqref{eq:rob_norm_ineq} we have thus shown that $\Rl_\F(\psi) \leq \norm{\ket\psi}_{\f}^2$, and so it remains to show the other inequality. To this end, consider any feasible $\ket{x} \in \H$ such that $\norm{\ket{x}}_\f^\circ \leq 1$. The operator $W = \proj{x}$ is then clearly positive, and 
\begin{equation}\begin{aligned}
    \sup_{\sigma \in \F} \< W, \sigma \> = \sup_{\ket{v} \in \f} \braket{ v | W | v} = (\norm{\ket{x}}_\f^\circ)^2 \leq 1
\end{aligned}\end{equation}
which means that $\Rl_\F(\psi)  \geq \< W, \psi \> = \left|\braket{\psi | x}\right|^2$. Optimizing over all feasible $\ket{x}$ gives $\Rl_\F(\psi)  \geq \norm{\ket{\psi}}_\f^2$.
\end{proof}

The Theorem extends a general relation from finite-dimensional spaces~\cite{regula_2018}, which includes several well-known correspondences: the generalized robustness of entanglement was previously shown to equal the squared sum of Schmidt coefficients of a pure state~\cite{vidal_1999,steiner_2003,harrow_2003} (which is indeed the norm $\norm{\cdot}_\f$ in this case~\cite{rudolph_2001}), and the generalized robustness of coherence was shown to equal the $\ell^1$ norm for pure states~\cite{piani_2016}.

A useful way to view the result of the Theorem is as follows: for any rank-one state $\psi$, it suffices to optimize over rank-one witnesses in the dual formulation of Eq.~\eqref{eq:rob_lsc_definition}.


\section{Examples and applications}\label{sec:examples}

We now consider explicit applications of our results in the characterization of several important quantum resources: optical nonclassicality, entanglement, coherence, and genuine non-Gaussianity.

\subsection{Nonclassicality}\label{sec:nonclassicality}

Coherent states \cite{Schroedinger1926-coherent, Klauder1960, glauber_1963, sudarshan_1963} and their probabilistic mixtures are widely recognized as the most classical among all quantum states of a quantum harmonic oscillator, and hence defined as classical states. From a resource theory perspective, the identification of this particular set of states is motivated by the fact that they can be easily produced and manipulated with standard techniques in quantum optical settings. Moreover, states which are nonclassical according to this distinction can be exploited to obtain operational advantages in applications such as linear optical quantum computation \cite{knill_2001}, quantum metrology \cite{yadin_2018, kwon_2019} and entanglement generation \cite{Aharonov1966, kim_2002, wang_2002, asboth_2005}.

Formally, we consider the quantum theory of a single harmonic oscillator. The corresponding Hilbert space is then $\H_1\coloneqq L^2(\RR)$, i.e., the space of square-integrable function on the real line. The annihilation and creation operators, denoted with $a,a^\dag$, respectively, satisfy the commutation relations $[a,a^\dag] = I$. Fock states are defined by $\ket{n}\coloneqq (a^\dag)^n \ket{0} /\sqrt{n!}$, where $\ket{0}$ is the vacuum state. For a complex number $\alpha\in \CC$, the corresponding coherent state is given by $\ket{\alpha}\coloneqq e^{-|\alpha|^2/2}\sum_{n=0}^\infty \frac{\alpha^n}{\sqrt{n!}} \ket{n} = \mathcal{D}_\alpha \ket{0}$, where $\mathcal{D}_\alpha \coloneqq e^{\alpha a^\dag - \alpha^* a}$ is a displacement operator. We deem free all the so-called classical states~\cite{Bach1986, yadin_2018}. Mathematically, we set $\F = \C \coloneqq \cl \conv\left\{ \ketbra{\alpha}{\alpha}: \alpha\in \CC \right\}$; this can equivalently be understood as the set of states whose Glauber--Sudarshan $P$ representation yields a valid probability distribution~\cite{glauber_1963,sudarshan_1963,Bach1986}. Among the simplest and most useful classical state is undoubtedly the thermal state with mean photon number $N\in [0,\infty)$, defined by
\begin{equation}
    \tau_N\coloneqq \frac{1}{N+1} \sum_{n=0}^\infty \left( \frac{N}{N+1} \right)^n \ketbra{n}{n}.
    \label{thermal}
\end{equation}

Notable examples of nonclassical states, instead, include the Fock states themselves, as well as the so-called squeezed states, obtained by letting a squeezing operations act on the vacuum, i.e.,
\begin{align}
S(r) &\coloneqq \exp\left[ \frac{r}{2}\left((a^\dag)^2 - a^2 \right) \right] , \label{squeezing}\\
\ket{\zeta_r} &\coloneqq S(r) \ket{0} = \frac{1}{\sqrt{\cosh(r)}} \sum_{n=0}^\infty \frac{1}{2^n} \sqrt{\binom{2n}{n}} \tanh(r)^n \ket{2n} , \label{squeezed_state}
\end{align}
where we assume that $r\geq 0$~\cite[Eq.~(3.7.2) and~(3.7.5)]{BARNETT-RADMORE}.

Recalling from the result in Sec.~\ref{sec:strong_duality_resources} that strong  duality holds in this resource theory, we have $\R_\C(\rho) = \Rl_\C(\rho)$ as well as $\Rs_\C(\rho) = \Rsl_\C(\rho)$ for all states.

\subsubsection{Infinite standard robustness}

We begin by establishing that the standard robustness $\Rsl_\C$ is in fact infinite for most physically accessible states in this resource theory.

\begin{lem}\label{unbounded chi lemma}
Any $m$-mode state $\rho$ with finite standard robustness of nonclassicality $\Rsl_\C (\rho)<\infty$ has a bounded normal-ordered characteristic function $\chi^\rho_1:\CC \to \CC$, where $\chi_1^\rho(\alpha)\coloneqq e^{|\alpha|^2/2}\Tr[\rho\mathcal{D}_\alpha]$. Specifically,
\begin{equation}
    \left\|\chi_1^\rho \right\|_{L^\infty} \leq 2 \Rsl_\C (\rho) -1\, ,
    \label{Linfty bound}
\end{equation}
where $\left\|\chi_1^\rho \right\|_{L^\infty}\coloneqq \sup_{\alpha\in \CC} \left|\chi_1^\rho(\alpha)\right|$.
\end{lem}

\begin{proof}
If $r\coloneqq \Rs_\C(\rho) = \Rsl_\C(\rho)<\infty$, for every $\epsilon>0$ there exists a free state $\sigma\in\F$ such that $\omega=\frac{\rho+(r+\epsilon-1)\sigma}{r+\epsilon}\in\F$. Both $\sigma$ and $\omega$ have positive $P$-functions, and hence their Fourier transform, i.e.\ $\chi_1^\sigma(\alpha)$ and $\chi_1^\omega(\alpha)$, are positive-definite~\cite{Bochner1933, Richter2002, Bohmann2020}. In this context, a function $f:\CC\to \CC$ is called positive definite if the matrix $\left(f(\alpha_\mu - \alpha_\nu)\right)_{\mu,\nu=1,\ldots, N}$ is positive semi-definite for all collections $\alpha_1,\ldots, \alpha_N\in \CC$. This in turn implies $|f(\alpha)|\leq|f(0)|$, and in particular that $f$ is bounded. The relation which defines the characteristic functions is linear, and hence
\begin{equation}
    \chi^\omega_1(\alpha)=\frac{\chi^\rho_1(\alpha)+(r+\epsilon-1)\chi^\sigma_1(\alpha)}{r+\epsilon}\,.
\end{equation}
Hence, also $\chi^\rho_1(\alpha)$ has to be bounded, and moreover
\begin{equation}\begin{aligned}
\left\|\chi_1^\rho \right\|_{L^\infty} &\leq (r+\epsilon) \left\|\chi_1^\omega\right\|_{L^\infty} + (r+\epsilon -1) \left\|\chi_1^\sigma\right\|_{L^\infty} \\
&= (r+\epsilon) \left|\chi_1^\omega(0)\right| + (r+\epsilon -1) \left|\chi_1^\sigma(0)\right| \\
&= 2(r+\epsilon) -1 \, ,
\end{aligned}\end{equation}
where in the last step we used that $\chi_1^\tau(0) = \Tr[\tau \mathcal{D}_0] = \Tr[\tau]=1$ for all density operators $\tau$. Since $\epsilon>0$ was arbitrary, we deduce \eqref{Linfty bound}. 
\end{proof}

Examples of states with unbounded $\chi_1$ functions (and hence infinite standard robustness of nonclassicality) comprise Fock, squeezed and cat states, \txb{but also some physically achievable approximations thereof. Consider for instance a squeezed thermal state
\begin{equation}
    \rho_{N,r} \coloneqq S(r) \tau_N S^\dag(r)\, ,
    \label{squeezed_thermal}
\end{equation}
where $S(r)$ is given by \eqref{squeezing} and $\tau_N$ by \eqref{thermal}. A state such as that in \eqref{squeezed_thermal} is a reasonable physical approximation to a pure squeezed state, one that can actually be achieved in a laboratory. It is well known that $\rho_{N,r}$ is nonclassical if and only if $e^{2r} > 2N+1$. It is also straightforward to show that its normal-ordered characteristic function evaluates to $\chi_1^{\rho_{N,r}}(\alpha) = e^{|\alpha|^2/2} e^{-\frac12 (2N+1) \left(e^{-2r} \alpha_R^2 + e^{2r} \alpha_I^2\right)}$, where $\alpha_R\coloneqq \Re \alpha$ and $\alpha_I \coloneqq \Im \alpha$. Therefore, $\Rsl_\C(\rho_{N,r}) = +\infty$ as soon as $\rho_{N,r}$ is nonclassical. It is not difficult to generalize this statement to all Gaussian states, so that $\Rsl_\C$ becomes trivial (namely, either $0$ or $+\infty$) on this whole set of states.

But there is more: the next Proposition proves that a large class of pure (possibly non-Gaussian) states also have unbounded $\chi_1$ function.}

\begin{pro}\label{prop:infinite_noncl_robustness}
Let $\ket{\psi}$ be a nonclassical pure state having a vanishing overlap with a finite (possibly empty) set of coherent states. Then, $\Rsl_\C(\psi)=\infty$.
\end{pro}

\begin{proof}
Let us write $\ket{\psi}=\sum_{n=0}^\infty c_n\ket{n}$. The function $f(\alpha)=e^{|\alpha|^2/2}\braket{\psi|\alpha}=\sum_{n=0}^\infty \frac{c^*_n}{\sqrt{n!}}\alpha^n$ is a complex entire function of order at most 2 (otherwise, $|\braket{\psi|\alpha}|$ would diverge). If $N<\infty$ is the number of zeros of $f(\alpha)$, its Hadamard factorization \cite{ComplexVariablesHandbook} becomes $f(\alpha)=e^{a\alpha^2+b\alpha}P_N(\alpha)$, where $P_N(\alpha)$ is a polynomial of degree $N$ and $|a|<\frac{1}{2}$ in order for $|\braket{\psi|\alpha}|$ to be bounded. The Husimi Q-function is then
\begin{equation}\begin{aligned}\label{q function expression}
Q^\psi(\alpha)&=e^{-|\alpha|^2}|f(\alpha)|^2\\&=e^{-|\alpha|^2}e^{2\Re[a\alpha^2+b\alpha]}|P_N(\alpha)|^2\\&=e^{-\mathbf{r}(\alpha)^T A\mathbf{r}(\alpha)+\boldsymbol{\beta}^T\mathbf{r}(\alpha)}|P_N(\alpha)|^2\,,
\end{aligned}\end{equation}
with
\begin{equation*}
\mathbf{r}(\alpha)=
\begin{pmatrix}
\Re\alpha\\\Im\alpha
\end{pmatrix},\quad
A=\begin{pmatrix}
1-2\Re a& 2\Im a\\
2\Im a & 1+2\Re a
\end{pmatrix},\quad
\boldsymbol{\beta}=
\begin{pmatrix}
2\Re b\\-2\Im b
\end{pmatrix}\,.
\end{equation*}
It is easy to see that the matrix $A$ has eigenvalues $1\pm2|a|$. The Fourier transform of a function in the form \eqref{q function expression} has again the same structure, but with $A^{-1}$ in place of $A$. Now, let us suppose for the moment that $a>0$. In this case, $A$ has an eigenvalue strictly bigger than 1, and hence $A^{-1}$ has one strictly smaller than 1. Thus, $\chi_1^\psi=e^{|\alpha|^2}\chi_{-1}^\psi$, where $\chi_{-1}^\psi$ is the Fourier transform of $Q^\psi$, is necessarily unbounded and, by virtue of Lemma \ref{unbounded chi lemma}, we conclude that $\Rsl_\C(\psi)=\infty$. If instead $a=0$, we have $A=A^{-1}$. At this point, we have to make a subsequent distinction: if $P_N$ is the trivial polynomial, $\psi$ is a coherent state and hence is obviously classical; if $P_N$ is not trivial we end up once again with a divergent $\chi_1^\psi$ (in this case the divergence is polynomial instead of exponential).
\end{proof}

Among the states which fulfil the hypothesis of the last result there are e.g. any finite superposition of Fock states and any nonclassical Gaussian state. It is worth noticing that the celebrated cat state $\ket{\alpha_+}\propto\ket{\alpha}+\ket{-\alpha}$ is a pure state with vanishing overlap with \txb{infinitely many} coherent states, but with unbounded $\chi_1$ and hence infinite standard robustness of nonclassicality.

\note{%
We remark that the number of coherent states for which $\braket{\psi|\alpha}$ vanishes has been considered in~\cite{chabaud_2020} as the so-called stellar rank, related to the degree of non-Gaussianity of a state.}

Let us briefly clarify an apparent similarity between the standard robustness of nonclassicality and the notion of Glauber--Sudarshan P representation~\cite{glauber_1963,sudarshan_1963}. Specifically, recalling that any state $\sigma \in \C$ can be written as $\int_{\C'} \proj{\alpha}\, \mathrm{d}\mu(\alpha)$ for some Borel probability measure $\mu$ on the set $\C' \coloneqq \{\proj{\alpha}\}_{\alpha \in \CC}$~\cite{Bach1986,holevo_2005}, we can write the standard robustness $\Rs_\C$ as the least coefficient $\lambda$ such that
\begin{equation}\begin{aligned}
\rho = \lambda \int_{\C'} \proj{\alpha}\, \mathrm{d}\mu_+(\alpha) - (\lambda - 1) \int_{\C'} \proj{\alpha}\, \mathrm{d} \mu_-(\alpha)
\end{aligned}\end{equation}
for some two probability measures $\mu_+, \mu_-$. Notice that any two measures $\mu_\pm$ give rise to a signed measure (i.e., one allowing negative values) $\mu' = \mu_+ - \mu_-$; conversely, one can use the Hahn--Jordan decomposition theorem to write any signed measure on $\C'$ as the difference of two non-negative measures $\mu_\pm$~\cite[Section~6.6]{RUDIN-ANALYSIS}. This shows that the standard robustness (or, specifically, $\Rs_\C - 1$) admits a natural interpretation as the \textit{negativity} of the resource (cf.~\cite{tan_2020}), in the sense that it quantifies the minimal negative part of a signed measure $\mu'$ such that
 \begin{equation}\begin{aligned}\label{eq:signed_measure}
     \rho = \int_{\C'} \proj{\alpha}\, \mathrm{d} \mu'(\alpha).
 \end{aligned}\end{equation}
An important difference between this expression and the Glauber--Sudarshan P representation is that the latter is based on quasiprobability distributions over the set $\C'$, which are a strictly more general concept than signed Borel measures.
In particular, we can see from Prop.~\ref{prop:infinite_noncl_robustness} that there exist states which cannot be written as in~\eqref{eq:signed_measure} for any signed measure $\mu$, even though it is known that any state admits a representation in terms of a quasiprobability distribution on $\C'$ in the form of the P representation~\cite{glauber_1963,sudarshan_1963}.


\subsubsection{Fock states}

\begin{pro}For any $n>0$,
\begin{equation}\begin{aligned}
\Rl_\C(\proj{n}) = e^{n} \frac{n!}{n^n}.
\label{eq:robustness Fock nonclassicality}
\end{aligned}\end{equation}
\end{pro}
\begin{proof}
Define
\begin{equation}\begin{aligned}
    \gamma_n &\!\coloneqq \sup_{\alpha \in \CC}\, \abs{\braket{\alpha | n}}^2\\
    &= \sup_{a \geq 0}\, e^{-a} \frac{a^n}{n!}\\
    &= e^{-n} \frac{n^n}{n!}.
\end{aligned}\end{equation}
Lemma~\ref{lem:robustness_lowerbound} then immediately gives
\begin{equation}\begin{aligned}
    \Rl_\C(\proj{n}) \geq \gamma_n^{-1}.
\end{aligned}\end{equation}
Inspired by the methods used to compute the nonclassical trace distance of Fock states~\cite{hillery_1987,nair_2017}, consider now the phase-randomized coherent state
\begin{equation}\begin{aligned}\label{eq:phase-random}
    \sigma_n &\!\coloneqq \frac{1}{2\pi} \int_0^{2\pi} \proj{\sqrt{n} e^{i\theta}} \mathrm{d}\theta\\
    &= e^{-n} \sum_{k=0}^{\infty} \frac{n^k}{k!} \proj{k}.
\end{aligned}\end{equation}
Using Lemma~\ref{lem:robustness_pure}, we get
\begin{equation}\begin{aligned}
    \Rl_\C(\proj{n})  &\leq \braket{n | \sigma_n^{-1} | n}\\
    &= e^n \frac{n!}{n^n}\\
    &= \gamma_n^{-1}.
\end{aligned}\end{equation}
The upper and lower bounds thus coincide, and we have $\Rl_\C(\proj{n}) = \gamma_n^{-1}$ for any $n>0$.
\end{proof}
We note that the bound from Lemma~\ref{lem:robustness_lowerbound} that we employed in the proof can be rephrased in terms of a geometric measure of nonclassicality based on the Husimi Q function, which was previously studied in several works~\cite{wuensche_2001,malbouisson_2003,nair_2017,marian_2019}. In particular, we have $\displaystyle \Rl_\C(\rho) \geq \frac{\< \rho, \omega \>}{\pi Q_{\max}(\omega)}$ for any state $\rho, \omega$, where $Q_{\max} (\omega) = \sup_{\alpha \in \CC} Q_\omega(\alpha)$ with $Q_\omega(\alpha) = \frac{1}{\pi} \braket{\alpha | \omega | \alpha}$ being the Husimi Q-function.

We also have from Cor.~\ref{cor:relative_entropy} that the above immediately implies
\begin{equation}\begin{aligned}
   \inf_{\sigma \in \C} D(\proj{n} \| \sigma) = \log \Rl_\C(\proj{n}) = - \log \gamma_n.
\end{aligned}\end{equation}
This expression for the relative entropy of nonclassicality has been independently shown in~\cite{relentNC}.

We can furthermore compute the robustness for the noisy Fock state
\begin{equation}\begin{aligned}
    \rho_{n,t} \coloneqq t \proj{n} + (1-t) \sigma_n.
\end{aligned}\end{equation}
On the one hand, convexity of $\Rl_\C$ gives
\begin{equation}\begin{aligned}
    \Rl_\C(\rho_{n,t}) \leq t \gamma_n^{-1} + (1-t),
\end{aligned}\end{equation}
and on the other hand using the feasible dual solution $W = \frac{\proj{n}}{\gamma_n}$ gives
\begin{equation}\begin{aligned}
    \Rl_\C(\rho_{n,t}) &\geq \< W, \rho_{n,t} \>\\
    &= t \gamma_n^{-1} + (1-t) \frac{\braket{n|\sigma_n|n}}{\gamma_n}\\
    &= t \gamma_n^{-1} + (1-t).
\end{aligned}\end{equation}

In a similar way, we can bound the robustness of the mixed state
\begin{equation}\begin{aligned}
    \rho_{n,q} \coloneqq q \proj{n} + (1-q) \proj{0}
\end{aligned}\end{equation}
as
\begin{equation}\begin{aligned}
  q \gamma_n^{-1} \leq \Rl_\C(\rho_{n,q}) \leq q \gamma_n^{-1} + (1-q).
\end{aligned}\end{equation}


\subsubsection{Squeezed states} \label{subsubsec:squeezed_states}

\begin{pro}\label{prop:rob_squeezed_states}
The robustness of nonclassicality of the squeezed states \eqref{squeezed_state} is given by
\begin{equation}
   \Rl_\C(\zeta_r) = e^r
\end{equation}
for all $r\geq 0$.
\end{pro}

\begin{proof}
For any squeezed state $\ket{\zeta_q}$ with $q \geq 0$, we have
\begin{equation}
    \sup_{\alpha\in \CC} \left|\braket{\alpha|\zeta_q}\right|^2 = \sup_{\alpha\in \CC} \frac{1}{\cosh(q)}\, e^{-|\alpha|^2+ \tanh(q) \Re(\alpha^2)} = \frac{1}{\cosh(q)}.
\end{equation}
Employing the lower bound from Lemma~\ref{lem:robustness_lowerbound} with the choice $\omega = \zeta_q$ then gives
\begin{equation}\begin{aligned}
    \Rl_\C(\zeta_r) &\geq \cbraket{\zeta_r|\zeta_q}^2 \cosh(q)\\
    &= \frac{\cosh(q)}{\cosh(q-r)}
\end{aligned}\end{equation}
for any $q$, where we used the well-known expression for the overlap of two squeezed states (see e.g.~\cite[3.7]{BARNETT-RADMORE}). In the limit $q \to \infty$, this gives
\begin{equation}\begin{aligned}
    \Rl_\C(\zeta_r) \geq e^r.
\end{aligned}\end{equation}

We now move on to the proof of the upper bound. For $s\geq 0$, let us construct the state
\begin{equation}
    \sigma_s \coloneqq S(s) \tau_{N(s)} S^\dag (s), \qquad N(s) \coloneqq \frac{e^{2s}-1}{2} ,
\label{eq:thermal ansatz}
\end{equation}
where $S$ and $\tau_N$ are defined by \eqref{squeezing} and \eqref{thermal}, respectively. Note that $\sigma_s$ is a Gaussian state with quantum covariance matrix\footnote{The quantum covariance matrix of an $m$-mode state $\rho$ is given by $(V_\rho)_{ij}\coloneqq \Tr\left[\rho \{ R_i - r_i, R_j - r_j\}\right]$, where $R\coloneqq (x_1,\ldots, x_m, p_1, \ldots, p_m)^\intercal$ is the vector of canonical operators, and $r\coloneqq \Tr[\rho R]\in \RR^{2m}$ is the vector of first moments of the state \cite{BUCCO}.} $V_{\sigma_s} = \left(\begin{smallmatrix} 1 & \\ & e^{4s} \end{smallmatrix}\right)$. Since $V_{\sigma_s}\geq \id$, the state $\sigma_s$ is classical. In fact, one can check that
\begin{equation}
    \sigma_s = \sqrt{\frac{2}{\pi (e^{4s}-1)}} \int_{-\infty}^{+\infty} dt\, e^{-\frac{2t^2}{e^{4s}-1}}\, \ketbra{it}{it} ,
\end{equation}
where $\ket{it}$ is a coherent state. 

We now employ this state as an ansatz for Lemma \ref{lem:robustness_pure}. Using the identity $\sum_{n=0}^\infty \frac{1}{4^n} \binom{2n}{n} t^n = \frac{1}{\sqrt{1-t}}$ for $t\in [0,1)$, which in turn can be easily retrieved from the fact that the squeezed state \eqref{squeezed_state} is normalized, we find that 
\begin{equation}
    \Rl_\C(\zeta_r)  \leq g(r,s) ,
\end{equation}
where
\begin{align*}
    g(r,s) \coloneqq&\ \braket{\zeta_r | \sigma_s^{-1} | \zeta_r} \\
    =&\ \braket{0 | S^\dag (r) S(s) \tau_{N(s)}^{-1} S^\dag (s) S(r) | 0} \\
    =&\ \braket{\zeta_{r-s} | \tau_{N(s)}^{-1} | \zeta_{r-s}} \\
    =&\ \frac{\sinh (s)}{(1-\tanh (s)) \sqrt{\sinh (r) \sinh (2s-r)}} ,
\end{align*}
where to obtain a finite result we assumed that $2s>r$. A lengthy yet straightforward calculation shows that, for a fixed $r$, the above function achieves the minimum in $s$ for $s=s_0(r)\coloneqq \frac{r}{2} + \frac14 \ln \left(2-e^{-2r}\right)$. Note that $2s_0(r)>r$ as long as $r>0$. A remarkable simplification occurs when we plug this value of $s$ inside the function $g$. Namely, upon elementary manipulations one obtains that $g(r,s_0(r)) = e^r$.
\end{proof}

\txb{
Before we move on, we wish to remark that squeezed states can be used to demonstrate a particularly simple application of Theorem~\ref{thm:robustness_discrimination} to the theory of nonclassicality. Consider in fact a channel ensemble composed of several displacement operators, whose action is to translate the input state by a different amount along the same line in phase space; this can be understood as ``classical message encoders'' that encode a message on the input states in the form of a displacement parameter. Using a probe whose initial state is highly squeezed along the same direction allows for almost perfect discrimination of this ensemble. On the contrary, if the initial state is restricted to be classical, even the best measurement at the output will still yield a significant error. The robustness of the squeezed state, which we just computed in Proposition~\ref{prop:rob_squeezed_states}, quantifies exactly the ratio between the two success probabilities of decoding~\cite{our_in_preparation}, in accordance with Theorem~\ref{thm:robustness_discrimination}.
}


\subsubsection{Photon-added and photon-subtracted squeezed states}

We define single-mode single-photon-added and single-photon-subtracted squeezed vacuum states as 
\bal
 \ket{\zeta_{r,\theta}}_+:=a^\dagger \ket{\zeta_{r,\theta}}/\sqrt{\braket{\zeta_{r,\theta}|aa^\dagger|\zeta_{r,\theta}}}= a^\dagger\ket{\zeta_{r,\theta}}/\cosh(r)
\eal
\bal
 \ket{\zeta_{r,\theta}}_-:=a \ket{\zeta_{r,\theta}}/\sqrt{\braket{\zeta_{r,\theta}|a^\dagger a|\zeta_{r,\theta}}}= a\ket{\zeta_{r,\theta}}/\sinh(r).
\eal
where $\ket{\zeta_{r,\theta}}=R(\theta)S(r)\ket{0}$ with $R(\theta)=e^{ia^\dagger a \theta}$ is the phase rotation unitary. Since $R(\theta)$ is a passive Gaussian unitary, it does not affect the degree of nonclassicality. 
Thus, we take $\theta=0$ and define $\ket{\zeta_{r}}_\pm:=\ket{\zeta_{r,0}}_\pm$. 
Next, we note that $\ket{\zeta_{r}}_+$ and $\ket{\zeta_{r}}_-$ are actually identical. To see this, let us express these in the number representation.
\bal
 \ket{\zeta_{r}}_+ = \frac{1}{\cosh^{3/2}(r)} \sum_{n=0}^\infty \frac{\sqrt{2n+1}}{2^n} \sqrt{\binom{2n}{n}} \tanh^n(r) \ket{2n+1}
 \label{eq:photon added number}
\eal
\bal
 \ket{\zeta_{r}}_- = \frac{1}{\sinh(r)\sqrt{\cosh(r)}} \sum_{n=1}^\infty \frac{\sqrt{2n}}{2^n} \sqrt{\binom{2n}{n}} \tanh^n(r) \ket{2n-1}.
 \label{eq:photon subtracted number}
\eal
Then, \eqref{eq:photon subtracted number} may be further computed as 
\bal
 \ket{\zeta_{r}}_- &= \frac{1}{\sinh(r)\sqrt{\cosh(r)}} \sum_{n=0}^\infty \frac{\sqrt{2(n+1)}}{2^{n+1}} \sqrt{\binom{2(n+1)}{n+1}}\\ &\hspace{4cm}\times\tanh^{n+1}(r) \ket{2n+1}\\
 &= \frac{\tanh(r)}{\sinh(r)\sqrt{\cosh(r)}} \sum_{n=0}^\infty \frac{\sqrt{2n+1}}{2^n} \sqrt{\binom{2n}{n}} \tanh^n(r) \ket{2n+1},
 \label{eq:photon subtracted number 2}
\eal
which is identical to $\ket{\zeta_r}_+$.
Hence, it suffices to only consider photon-added states. 
Then, we have the following bounds for the robustness for this class of states.

\begin{pro}\label{prop:photon_added}
Let $\ket{\zeta_r}_\pm$ be single-photon-added/subtracted squeezed vacuum states as defined above. Then,
\begin{equation}\label{eq:photon_added_lower}
    e^{-r+1}\cosh^2(r) \leq \Rl_\C({\zeta_r}_\pm) \leq \frac{4e^{2r}}{3\sqrt{3}\sinh(r)}
\end{equation}
for all $r\geq 0$.

A tighter lower bound can be obtained for $r \leq \ln\sqrt{2} \approx 0.35$ as $\displaystyle \Rl_\C({\zeta_r}_\pm) \geq e \cosh(r)^{-3}$
and for $r \geq \ln\sqrt{2}$ as $\displaystyle \Rl_\C({\zeta_r}_\pm) \geq \frac{4}{27} e^{1+2r} \sinh(r)^{-1}$.
\end{pro}
\begin{proof}
We sketch the proof here. Since the exact calculations are lengthy, we defer the complete details to Appendix~\ref{app:nonclassicality}.

The lower bounds all follow from Lemma~\ref{lem:robustness_lowerbound} with suitable choices of $\omega = \ket{\zeta_q}_+\!\bra{\zeta_q}$ in Eq.~\eqref{rob_lower_bound}. In particular, we compute
\bal
 \frac{\left|{}_+\braket{\zeta_q|\zeta_r}_+\right|^2}{\sup_{\alpha\in \CC} \left|\braket{\alpha|\zeta_q}_+\right|^2} = \frac{e^{-q+1}\cosh^2(q)}{\cosh^3(r-q)}.
\eal
The general lower bound in Eq.~\eqref{eq:photon_added_lower} follows by choosing $q=r$. For $r \leq \ln\sqrt{2}$, we get a tighter bound with the choice of $q = 0$, and for $r \geq \ln\sqrt{2}$ we can choose $\displaystyle q = \frac{1}{2} \ln\left(2e^{2r}-3\right)$.

To obtain the upper bound, we consider the ansatz in Eq.~\eqref{eq:thermal ansatz}. Computing
\bal
    g_{\rm PA}(r,s) \coloneqq&\ {}_+\braket{\zeta_r | \sigma_s^{-1} | \zeta_r}_+ \\
    =&\ \frac{\cosh(s)\sinh^2(s)}{(1-\tanh(s))\left(\sinh(r)\sinh(2s-r)\right)^{3/2}} ,
\eal
we see that this expression is minimized for $s=\frac14 \ln (4e^{2r}-3)$. Plugging this value of $s$ into the above function yields the claimed bound.
\end{proof}


\subsubsection{Cat states}

\begin{pro}\label{prop:cat_states}
For the cat states $\ket{\alpha_\pm} = \frac{1}{\sqrt{2 c_\pm}} (\ket\alpha \pm \ket{-\alpha})$, where $\alpha > 0$ and $c_\pm = 1 \pm e^{-2\alpha^2}$, we have
\begin{equation}\begin{aligned}
    \Rl_\C(\alpha_\pm) \leq \frac{2}{c_\pm}
\end{aligned}\end{equation}
and corresponding lower bounds give $\Rl_\C(\alpha_\pm) \to 2$ as $\alpha \to \infty$.

In the low $\alpha$ regime, we can furthermore obtain improved analytical bounds. For the even cat state $\ket{\alpha_+}$, it holds that
\begin{equation}\begin{aligned}
    \Rl_\C(\alpha_+) \geq \frac{\cosh(\alpha)^2}{\cosh(\alpha^2)}.
\end{aligned}\end{equation}
For the odd cat state $\ket{\alpha_-}$, we have
\begin{equation}\begin{aligned}
    \frac{\alpha^2 e}{\sinh(\alpha^2)} \leq \Rl_\C(\alpha_-) \leq \frac{\alpha^2 e}{(1-\alpha^4) \sinh(\alpha^2)},
\end{aligned}\end{equation}
where the upper bound applies only to $\alpha < 1$. The lower bound is tight for $\alpha = 1$.
\end{pro}
\begin{proof}
Defining the state $\sigma_\alpha \coloneqq \frac{1}{2} ( \proj{\alpha} + \proj{-\alpha}) \in \C$, one can notice that each $\ket{\alpha_\pm}$ is an eigenvector of this operator with eigenvalue $\frac{c_\pm}{2}$, respectively. Lemma~\ref{lem:robustness_pure} then gives
\begin{equation}\begin{aligned}
     \Rl_\C(\alpha_\pm) \leq \braket{\alpha_\pm| \sigma_\alpha^{-1} | \alpha_\pm} = \frac{2}{c_\pm}.
\end{aligned}\end{equation}
This can be alternatively noticed by recalling from Prop.~\ref{thm:robustness_pure_norm} that $\Rl_\C(\alpha_\pm) = \norm{\ket{\alpha_\pm}}_\f^2$ and using that
\begin{equation}\begin{aligned}
    \norm{\ket{\alpha_\pm}}_\f &= \frac{1}{\sqrt{2c_\pm}} \norm{\ket{\alpha} \pm \ket{-\alpha}}_\f\\
    &\leq \frac{1}{\sqrt{2c_\pm}} \left(\norm{\ket\alpha}_\f + \norm{\ket{-\alpha}}_\f\right)\\
    &=\sqrt{\frac{2}{c_\pm}}.
\end{aligned}\end{equation}

Now, from Lemma~\ref{lem:robustness_lowerbound} we have that
\begin{equation}\begin{aligned}\label{eq:cat_lowerbounds}
    \Rl_\C(\alpha_+) &\geq \inf_{\beta \in \CC} \frac{1}{\cbraket{\beta|\alpha_+}^2}\\
    &=\inf_{\beta \geq 0} \frac{ e^{\beta^2} \cosh(\alpha^2) }{\cosh(\alpha\beta)^2},\\
    \Rl_\C(\alpha_-) &\geq \inf_{\beta \geq 0} \frac{ e^{\beta^2} \sinh(\alpha^2) }{\sinh(\alpha\beta)^2}.
\end{aligned}\end{equation}
Although the bounds do not seem amenable to an analytical expression in general, they can be verified to approach $2$ in the large $\alpha$ limit (cf.~\cite{malbouisson_2003,nair_2017}).

Alternatively, for the even cat state $\ket{\alpha_+}$ we can take the cat state $\ket{\gamma_+}$ with parameter $\gamma = 1$ as an ansatz to the lower bound in Lemma~\ref{lem:robustness_lowerbound}. For this state, we compute $\sup_{\beta \in \CC} \cbraket{1_+|\beta}^2 = \cosh(1)^{-1}$. Using that $\cbraket{\alpha_+|\gamma_+}^2 = \frac{\cosh(\alpha\gamma)^2}{\cosh(\alpha^2)\cosh(\gamma^2)}$, we get
\begin{equation}\begin{aligned}\label{eq:cat_even_lower_one}
\Rl_\C(\alpha_+) \geq \frac{\cosh(\alpha)^2}{\cosh(\alpha^2)}.
\end{aligned}\end{equation}
For the odd cat state $\ket{\alpha_-}$ we can choose the Fock state $\ket{1}$ as an ansatz to Lemma~\ref{lem:robustness_lowerbound}, giving
\begin{equation}\begin{aligned}\label{eq:cat_odd_fock}
 \Rl_\C(\alpha_-) \geq \frac{\alpha^2 e}{\sinh(\alpha^2)}.
\end{aligned}\end{equation}
The upper and lower bound then coincide for $\alpha = 1$, in which case we have $\displaystyle \Rl_\C(1_-) = \frac{2e^2}{e^2 - 1}$. For an alternative upper bound, we can choose the phase-randomized state $\sigma_1$ (see Eq.~\ref{eq:phase-random}). Since for all $0 < \alpha < 1$, $\ket{\alpha_-} \in \ran(\sigma_1^{1/2})$,  Lem.~\ref{lem:robustness_pure} gives
\begin{equation}\begin{aligned}\label{eq:cat_odd_upper_sigmaone}
    \Rl_\C(\alpha_-) \leq \braket{\alpha_-| \sigma_1^{-1} | \alpha_-} = \frac{\alpha^2 e}{(1-\alpha^4) \sinh(\alpha^2)}.
\end{aligned}\end{equation}
\end{proof}

We note that the lower bounds in Eq.~\eqref{eq:cat_lowerbounds} can be easily evaluated numerically --- we visualize this in Fig.~\ref{fig:bounds_cat}. In particular, the convergence of the lower bound to $2$ is fast, and already for $\alpha = 2$ we obtain
\begin{equation}\begin{aligned}
    1.9990 &\approx \frac{1}{\sup_{\beta \in \CC} \cbraket{\alpha_+|\beta}^2} \leq \Rl_\C(\alpha_+) \leq \frac{2}{c_+} \approx 1.9993\\
    2.0003 &\approx \frac{1}{\sup_{\beta \in \CC} \cbraket{\alpha_-|\beta}^2} \leq \Rl_\C(\alpha_-) \leq \frac{2}{c_-} \approx 2.0007.
\end{aligned}\end{equation}

An interesting property emerges by performing an explicit numerical optimization over the choices of even cat states $\ket{\gamma_+}$ in the bound from Lemma~\ref{lem:robustness_lowerbound},
\begin{equation}\begin{aligned}\label{eq:cat_optimized_lower}
    \Rl_\C(\alpha_+) \geq \sup_{\gamma \in \CC} \frac{\cbraket{\alpha_+|\gamma_+}^2}{\sup_{\beta \in \CC} \cbraket{\beta|\gamma_+}^2}.
\end{aligned}\end{equation}
As we see in Fig.~\ref{fig:bounds_cat}, the bound actually matches the upper bound $\frac{2}{c_+}$ perfectly for all $\alpha$. In a similar way, we obtain an exact match of the analytical upper and numerical lower bound for $\ket{\alpha_-}$ for all $\alpha \geq 1$.

For the odd cat state $\ket{\alpha_-}$, we additionally plot a bound obtained from numerically optimizing the parameter $x$ in the phase-randomized state $\sigma_x$ over $x > \alpha^2$, which can provide a slight improvement over the bound from $\sigma_1$. As $\alpha \to 0$, all of the bounds approach $\Rl_\C(\proj{1}) = e$.

\begin{figure*}[t]
\centering
\subfloat[$\ket{\alpha_+}$]{\includegraphics[width=.48\textwidth]{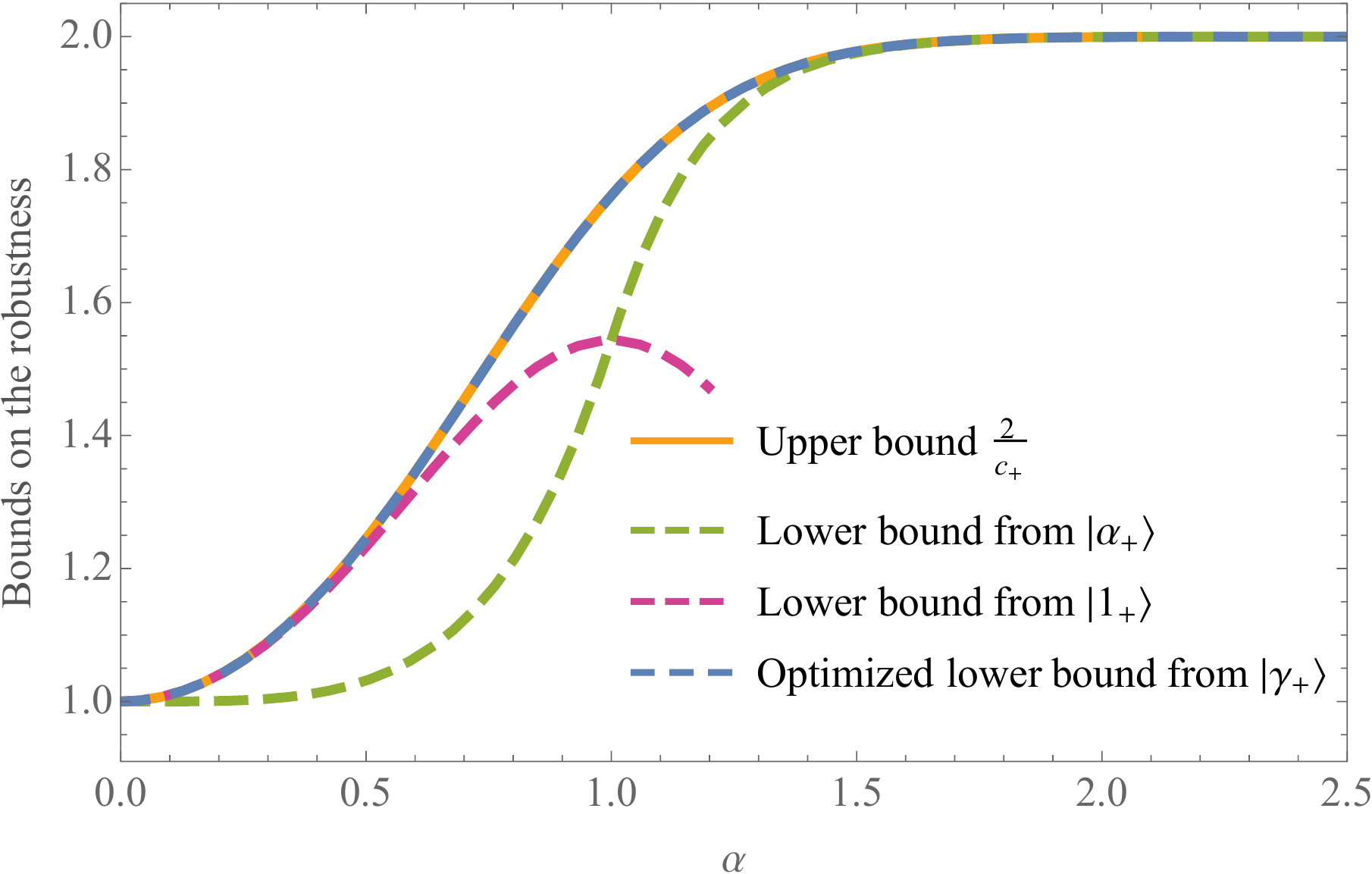}}\qquad
\subfloat[$\ket{\alpha_-}$]{\includegraphics[width=.48\textwidth]{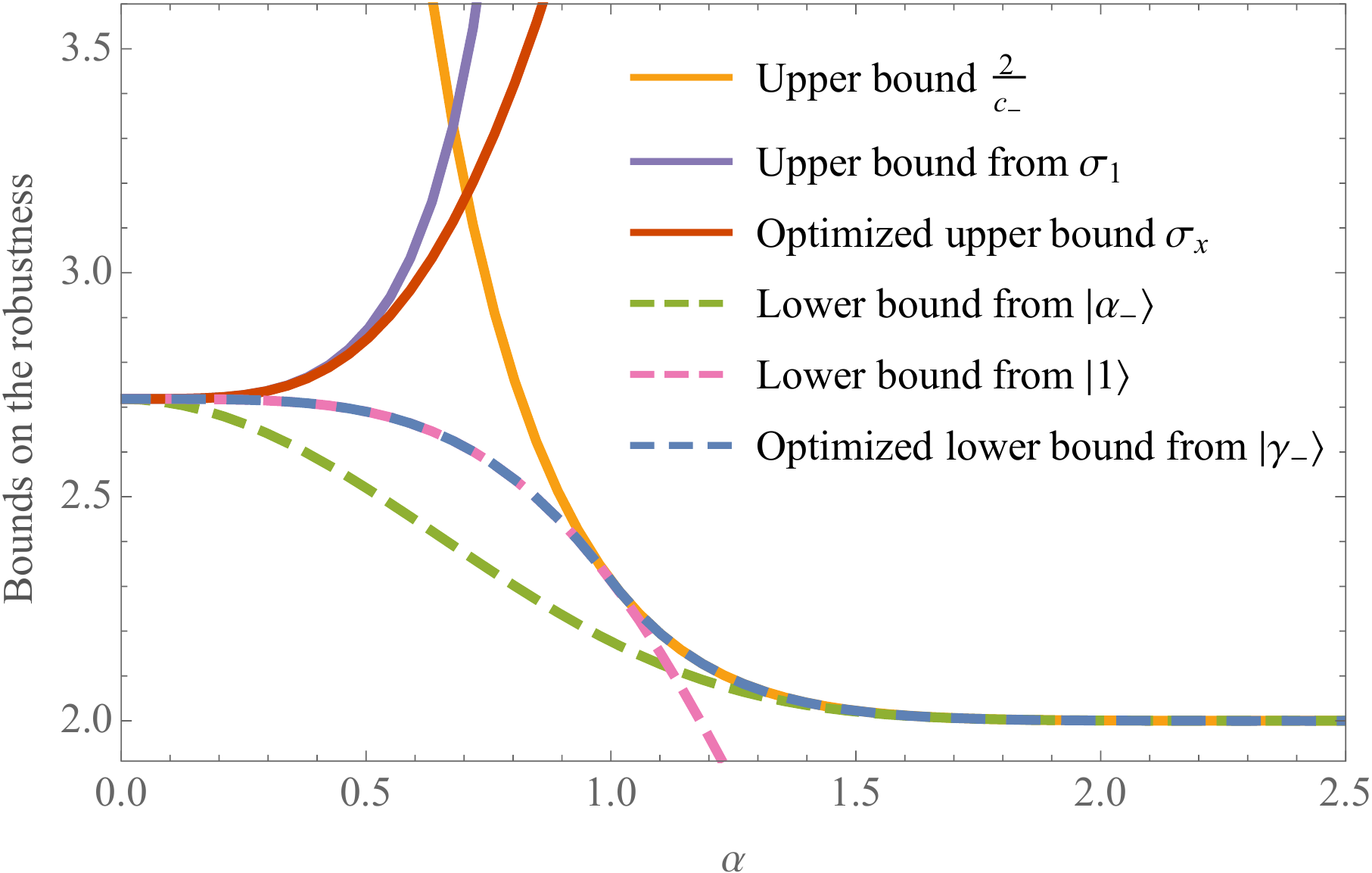}}
\caption{Evaluating bounds for the robustness of nonclassicality of the Schr\"odinger cat states (a) $\ket{\alpha_+}$ and (b) $\ket{\alpha_-}$. The plots show the various analytical and numerical bounds: (a) the lower bound obtained by choosing $\ket{\alpha_+}$ (Eq.~\eqref{eq:cat_lowerbounds}) and $\ket{1_+}$ (Eq.~\eqref{eq:cat_even_lower_one}) in Lemma~\ref{lem:robustness_lowerbound}, as well as the one obtained by numerically optimizing over the choice of $\ket{\gamma_+}$ (Eq.~\eqref{eq:cat_optimized_lower}); (b) the upper bounds obtained by choosing $\sigma_1$ (Eq.~\eqref{eq:cat_odd_upper_sigmaone}) or the optimized state $\sigma_x$ in Lemma~\ref{lem:robustness_pure}, and the lower bounds obtained by choosing $\ket{\alpha_-}$ (Eq.~\eqref{eq:cat_lowerbounds}), $\ket{1}$~(Eq.~\eqref{eq:cat_odd_fock}), or the optimized state $\ket{\gamma_-}$ in Lemma~\ref{lem:robustness_lowerbound}.}
\label{fig:bounds_cat}
\end{figure*}


\subsubsection{Multi-mode nonclassicality}

The Hilbert space modelling $m$ harmonic oscillators (or modes) is $\H = L^2(\RR^m)$, which can be written as a tensor product of the single-mode spaces $\H = L^2(\RR) \otimes \cdots \otimes L^2(\RR)$. 
The theory of multi-mode nonclassicality admits a remarkably simple structure, in the sense that any $m$-mode coherent state is of the form $\ket{\boldsymbol\alpha} = \ket{\alpha_1} \otimes \cdots \otimes \ket{\alpha_m}$ with $\boldsymbol\alpha = (\alpha_1, \ldots, \alpha_m) \in \CC^m$. This leads to the property that the robustness of nonclassicality is, in fact, multiplicative for products of single-mode states.
\begin{pro}\label{prop:multiplicativity_noncl}
Let $\{\rho_i\}_{i=1}^m$ be a collection of single-mode quantum states $\rho_i \in \D(L^2(\RR))$. Then
\begin{equation}\begin{aligned}
    \Rl_\C \left(\bigotimes_{i=1}^m \rho_i\right) = \prod_{i=1}^m \Rl_\C(\rho_i).
\end{aligned}\end{equation}
\end{pro}
\begin{proof}
We will use $\C_i$ to denote the set of classical states of the given mode. Let $\sigma_i \in \C_i$ be feasible states such that $\rho_i \cleq \lambda_i \sigma_i$. Then, $\bigotimes_i \rho_i \cleq \left(\prod_i \lambda_i\right) \bigotimes_i \sigma_i$ --- this can be seen for $m=2$ from the fact that
\begin{equation}\begin{aligned}
    0 &\cleq (\lambda_1 \sigma_1 - \rho_1) \otimes \lambda_2 \sigma_2 + \rho_1 \otimes (\lambda_2 \sigma_2 - \rho_2)\\
    &= \lambda_1 \lambda_2 (\sigma_1 \otimes \sigma_2 ) - \rho_1 \otimes \rho_2
\end{aligned}\end{equation}
where we used that the tensor product of positive operators is positive, and an extension to arbitrary $m$ is immediate. Since $\bigotimes_i \sigma_i \in \C$, this shows in particular that feasible solutions for $\R_\C(\rho_i) = \Rl_\C(\rho_i)$ can be used to construct a feasible solution for $\Rl_\C\left(\bigotimes_i \rho_i\right)$, leading to $\Rl_\C\left(\bigotimes_i \rho_i\right) \leq \prod_i \lambda_i$. As this holds for any feasible choice of $\lambda_i$, it must also hold for their infimum, which implies the desired relation $\Rl_\C\left(\bigotimes_i \rho_i\right) \leq \prod_i \Rl_\C(\rho_i)$.

Let $W_i$ be feasible dual solutions for $\Rl_\C(\rho_i)$, with each $W_i \cgeq 0$ and $\sup_{\sigma \in \C_i} \<W_i, \sigma\> \leq 1$. Then $\bigotimes_i W_i \cgeq 0$ and
\begin{equation}\begin{aligned}
    \sup_{\sigma \in \C} \< \bigotimes_i W_i, \sigma \> &= \sup_{\boldsymbol\alpha \in \CC^m} \< \bigotimes_i W_i, \proj{\boldsymbol\alpha} \> \\
    &= \sup_{\alpha_i \in \CC} \< \bigotimes_i W_i, \bigotimes_i \proj{\alpha_i} \> \\
    &= \sup_{\alpha_i \in \CC} \prod_i \braket{\alpha_i|W_i|\alpha_i}\\
    &\leq 1,
\end{aligned}\end{equation}
where in the first line we used that the supremum of a continuous linear functional over a set $\S$ is equal to its supremum over $\cl\conv\S$, and in the second line we used that every $m$-mode coherent state is a product of single-mode coherent states. This means that $\bigotimes_i W_i$ is a feasible witness for $\Rl_\C\left(\bigotimes_i \rho_i\right)$. Therefore $\Rl_\C\left(\bigotimes_i \rho_i\right) \geq \< \bigotimes_i W_i, \bigotimes_i \rho_i \> = \prod_i \< W_i, \rho_i \>$, and since this holds for any feasible $W_i$, it must also hold for their least upper bounds $\prod_i \Rl_\C(\rho_i)$.
\end{proof}

One can notice that the submultiplicativity of the robustness $\R_\F$ in fact holds in any resource theory such that $\sigma_1 \in \F(\H_1), \sigma_2 \in \F(\H_2) \Rightarrow \sigma_1 \otimes \sigma_2 \in \F(\H_1 \otimes \H_2)$, but the full multiplicativity relies on the particular structure of pure states in the theory of nonclassicality.


\subsection{Entanglement}\label{sec:entanglement}

The foremost example of a quantum resource theory is entanglement theory~\cite{werner_1989,bennett_1996-1, bennett_1996-3, horodecki_2009} (cf.~\cite[Sec.~IV.A]{chitambar_2019}). In spite of its historical importance, comparatively little is known on the properties of entanglement in infinite-dimensional systems~\cite{eisert_2002, holevo_2005, shirokov_2010, Shirokov-AFW-1, relentNC}. Throughout this section, we will specialize our results to this case, and establish the corresponding robustness as a valid entanglement measure in the infinite-dimensional setting.

We consider a tensor product of two separable Hilbert spaces $\H = \H_A \otimes \H_B$. The set of free pure states is defined to consist of product vectors:
\begin{equation*}
\f = \lset \ket{\phi} \otimes \ket{\tau} \sbar \ket{\phi} \in \H_A,\; \ket{\tau} \in \H_B,\, \braket{\phi|\phi}=1=\braket{\tau|\tau} \rset .
\end{equation*}
The closure of all convex combinations of such states defines the set of separable states $\S = \cl \conv \lset \proj \psi \sbar \ket\psi \in \f \rset$.

\subsubsection{Strong duality}

We start our investigation by showing that Thm.~\ref{thm:strong_duality_from_weakstar} applies, thus implying that $\R_\S=\Rl_\S$ holds in this case as well.

\begin{lem}
The cone $\cone(\S)$ of separable operators is closed in the weak* topology. Thus, for all states $\rho\in \D(\H_A\otimes \H_B)$ it holds that $\Rl_\S (\rho) = \R_\S(\rho)$ and $\Rsl_\S (\rho) = \Rs_\S(\rho)$.
\end{lem}

\begin{proof}
Let $\{\ket{n}_A\}_n$ and $\{\ket{m}_B\}_m$ be orthonormal bases of the local Hilbert spaces $\H_A$ and $\H_B$. For $N\in \NN$, consider the subspaces $\H_A^N\coloneqq \mathrm{span}\{\ket{0}_A,\ldots, \ket{N}_A \}$ and $\H_B^N\coloneqq \mathrm{span}\{\ket{0}_B,\ldots, \ket{N}_B \}$, with corresponding projectors $P_A^N$ and $P_B^N$. Define the mappings $\mathcal{P}_N:\T(\H_A\otimes \H_B)\to \T\left(\H_A^{N}\otimes \H_B^{N}\right)$ given by
\begin{equation}
    \mathcal{P}_N(\cdot) \coloneqq P_A^N \otimes P_B^N (\cdot) P_A^N \otimes P_B^N ,
\end{equation}
and construct the sets
\begin{align}
    \f^N &\coloneqq \lset \ket{\phi}\! \otimes\! \ket{\tau} \sbar \ket{\phi}\! \in\! \H_A^N,\; \ket{\tau}\! \in\! \H_B^N,\, \braket{\phi|\phi}=\!1\!=\braket{\tau|\tau} \rset , \\
    \S^N &\coloneqq \conv\left( \lset \ketbra{\psi}{\psi} \sbar \ket{\psi}\in \f^N \rset\right) .
\end{align}
Now, it is not difficult to verify that
\begin{equation}
    \cone(\S) = \bigcap_{N\in\NN} \mathcal{P}_N^{-1} \left( \cone(\S_N) \right) .
    \label{eq:cone_sep_projections}
\end{equation}
In fact, on the one hand any $X\in \cone(\S)$ clearly satisfies that $\mathcal{P}_N(X)\in \cone(\S_N)$. On the other, let $X$ be a non-zero trace-class operator such that $\mathcal{P}_N(X)\in \cone (\S_N)\subseteq \cone(\S)$ for all $N\in \NN$. First, it is elementary to see that $X\geq 0$ is in fact positive semidefinite. Second, it holds that $\lim_{N\to\infty} \Tr[\mathcal{P}_N(X)]=\Tr[X]$, because $X$ is trace-class. Incidentally, since $X\geq 0$ and $X\neq 0$, we have that $\Tr X > 0$. Therefore, we can apply the gentle measurement lemma~\cite[Lemma~9]{winter_1999} and conclude that $X_N\coloneqq \frac{\mathcal{P}_N(X)}{\Tr[\mathcal{P}_N(X)]}\in\S_N\subset \S$ satisfies that
\begin{equation}
    \lim_{N\to\infty} \frac{\mathcal{P}_N(X)}{\Tr\left[\mathcal{P}_N(X)\right]} = \frac{X}{\Tr[X]} ,
\end{equation}
where the limit is with respect to the trace norm topology. Since $\S$ is closed, we deduce that $\frac{X}{\Tr[X]}\in \S$ and hence that $X\in \cone(\S)$, as claimed.

Having established \eqref{eq:cone_sep_projections}, let us see why this ensures that $\cone(\S)$ is weak*-closed. Note that the range of the map $\mathcal{P}_N$ is finite-dimensional and the entries of $\mathcal{P}_N(X)$ are weak*-continuous functions of $X$ for a fixed $N$, so that the map $\mathcal{P}_N$ is itself weak*-continuous. 
Since $\cone(\S_N)$ is closed, it follows that $\mathcal{P}_N^{-1} \left( \cone(\S_N) \right)$ is weak*-closed as well. Being an intersection of weak*-closed sets, $\cone(\S)$ is itself weak*-closed. Thus, an application of Thm.~\ref{thm:strong_duality_from_weakstar} concludes the proof.
\end{proof}

\subsubsection{Pure states}

Having established the equality between the robustness and its lower semicontinuous version for the case of entanglement, we turn to the problem of computing the resulting function for pure states. Recall that, for any pure state $\ket\psi$, there exist orthonormal bases $\{\ket{i}\}_{i=0}^\infty \in \H_A, \{\ket{i}\}_{i=0}^\infty \in \H_B$ such that $\ket\psi$ can be written as
\begin{equation} \label{eq:Schmidt} \begin{aligned}
    \ket\psi = \sum_{i=0}^{\infty} \mu_i \ket{ii}
\end{aligned}\end{equation}
where $\{\mu_i\}_{i=1}^\infty$ is a monotonically non-increasing, square-summable sequence in $\RR_+$ referred to as the Schmidt coefficients of $\ket\psi$. When at least one of $\H_A$ and $\H_B$ is finite-dimensional, there are at most $d = \min\{\dim \H_A, \dim \H_B\}$ non-zero Schmidt coefficients, and we have~\cite{rudolph_2001}
\begin{equation}\begin{aligned}
    \norm{\ket\psi}_\f = \sum_{i=0}^d \mu_i.
\end{aligned}\end{equation}
We will extend the above also to the infinite-dimensional case, leading to a computable expression for the robustness $\Rl_\S(\psi)$ of any pure state.

The case of entanglement theory is peculiar because for pure states this same expression yields also the standard robustness $\Rs_\S$. While this has been established by Vidal and Tarrach~\cite{vidal_1999} for the finite-dimensional case, we now extend their result to our general infinite-dimensional setting.

\begin{pro}\label{thm:rob_ent_pure}
Consider any pure state  $\ket\psi \in \H_A \otimes \H_B$ and let $\{\mu_i\}_{i=0}^\infty \in \ell^2(\RR)$ be the sequence of its Schmidt coefficients. Then
\begin{equation}\begin{aligned}
\Rsl_\S(\psi) = \Rl_\S(\psi) = \left(\sum_{i=0}^\infty \mu_i\right)^2.
\end{aligned}\end{equation}
In particular, a pure state has $\Rl_\S(\psi) < \infty$ if and only if the sum $\sum_i \mu_i$ converges.
\end{pro}

\begin{proof}
Set $\mu\coloneqq \sum_{i=0}^\infty \mu_i$. Since the inequality $\R_\S(\psi)\leq \Rs_\S(\psi)$ holds by construction, we need to show that $\R_\S(\psi)\geq \mu^2$ and $\Rs_\S(\psi)\leq \mu^2$. 
To prove the latter inequality, define the family of vectors $\ket{w_\xi} \coloneqq \sum_{i=0}^{\infty} \xi^{i} \ket{ii}$, where $\xi \in (0,1)$. Note that $\{\xi^i\}_{i=0}^\infty \in \ell^1(\RR)$, and consider the positive bounded operator $W_\xi = \proj{w_\xi}$. The fact that $\norm{W_\xi}_\S^\circ \leq 1$ can be shown from a result by Shimony~\cite{shimony_1995}, but we will prove this explicitly for completeness. Since any product state can be written as $\ket{v} = \ket{v_A} \otimes \ket{v_B} \coloneqq (\sum_{k=0}^\infty a_k \ket{k}) \otimes (\sum_{l=0}^\infty b_l \ket{l})$ with $\norm{\ket{v_A}} = \norm{\ket{v_B}} = 1$ using the Schmidt bases of $\ket\psi$, for any $\ket{v} \in \f$ we have
\begin{equation}\begin{aligned}
   \< W_\xi, v \> &= \< v | W_\xi | v \>\\
    &= \sum_{i,j=0}^\infty \xi^{i} \xi^{j} a_i a_j^* b_i b_j^*\\
    &= \left(\sum_{i=0}^\infty \xi^{i} a_i b_i\right)\left(\sum_{i=0}^\infty \xi^{i} a_i b_i\right)^*\\
    &= \left|\< \Delta(W_\xi), \ketbra{v_B}{v_A^*} \>\right|^2\\
    &\leq \norm{\Delta(W_\xi)}^2_\infty \norm{\ket{v_A}}^2 \norm{\ket{v_B}}^2\\
    &\leq 1
\end{aligned}\end{equation}
where we used $\Delta(W_\xi)$ to denote the positive operator $\sum_{i=0}^\infty \xi^{i} \proj{i}$, and in the last line we used that $\xi^i \leq 1 \; \forall i$. This shows that $\norm{W_\xi}_\S^\circ \leq 1$ and hence we can take $W_\xi$ as a feasible solution for the robustness $\Rl_\S$, giving $\Rl_\S (\psi)  \geq \< W_\xi, \psi \>$ for any $\xi \in (0,1)$. This then gives
\begin{equation}\begin{aligned}
    \Rl_\S (\psi)  &\geq \lim_{\xi \to 1^-} \< W_\xi, \psi \>\\
    &= \left(\lim_{\xi \to 1^-} \textbf{}\sum_{i=0}^\infty \xi^i \mu_i\right)^2\\
    &= \left(\sum_{i=0}^\infty \mu_i\right)^2,
\end{aligned}\end{equation}
where in the third line we used Abel's theorem, which extends also to the case when the series of $\{\mu_i\}_{i=0}^\infty$ diverges to infinity. 
We conclude that $\Rl_\S(\psi)  \geq \mu^2$, which in particular implies that $\Rl_\S(\psi) = \infty$ when $\{\mu_i\}_{i=0}^\infty \notin \ell^1(\RR)$.

It remains to show that $\Rsl_\S(\psi) = \Rs_\S(\psi)\leq \mu^2$, where we shall now assume that $\mu < \infty$. To this end, for $\theta\in [0,2\pi]$, define
\begin{equation}
    \ket{\psi_\theta} \coloneqq \frac{1}{\mu^{1/2}} \sum_{j=0}^\infty e^{i\, 2^j \theta} \mu_j^{1/2} \ket{j}
\end{equation}
in either $\H_A$ or $\H_B$, where the basis employed here is the Schmidt basis of $\ket\psi$ (Eq.~\eqref{eq:Schmidt}). The state
\begin{equation}
    \omega\coloneqq \int_0^{2\pi} \frac{d\theta}{2\pi}\, \psi_\theta\otimes \psi_\theta^\intercal
\end{equation}
is separable by construction. For indices $j,k,\ell,m\in \NN$, we compute
\begin{equation}
\begin{aligned}
    &\braket{jk|\omega|\ell m} \\
    &\quad = \frac{1}{\mu^2} \int_0^{2\pi} \frac{d\theta}{2\pi} \sqrt{\mu_j \mu_k \mu_\ell \mu_m}\, e^{i\, 2^j\theta} e^{-i\, 2^k\theta} e^{-i\, 2^\ell\theta} e^{i\, 2^m\theta} \\
    &\quad= \frac{1}{\mu^2} \sqrt{\mu_j \mu_k \mu_\ell \mu_m} \int_0^{2\pi} \frac{d\theta}{2\pi} \,\exp\left( i \left[2^j + 2^m - 2^k - 2^\ell\right] \theta\right) .
\end{aligned}
\end{equation}
Since the binary representation of any integer is unique, $2^j + 2^m \neq 2^k + 2^\ell$ unless either $j=k$ and $\ell=m$ or $j=\ell$ and $k=m$. Therefore,
\begin{equation}
    \braket{jk|\omega|\ell m} = \frac{\mu_j \mu_m}{\mu^2} \left( \delta_{j,k} \delta_{\ell,m} + \delta_{j,\ell} \delta_{k,m} - \delta_{j,k,\ell,m} \right) ,
\end{equation}
where $\delta_{j,k,\ell,m}=1$ if $j=k=\ell=m$, and $\delta_{j,k,\ell,m}=0$ otherwise.
The last state we need to define is 
\begin{equation}
\sigma\coloneqq \frac{1}{\mu^2-1}\sum_{\substack{j,m=0,1,\ldots \\ j\neq m}} \mu_j \mu_m \ketbra{jm}{jm} ,
\end{equation}
which is also separable by construction. Now, consider that
\begin{equation} \begin{aligned}
\psi + \left(\mu^2\!-\!1\right) \sigma &= \sum_{j,m} \mu_j\mu_m \ketbra{jj}{mm} + \sum_{j\neq m} \mu_j \mu_m \ketbra{jm}{jm} \\
&= \sum_{j,k,l,m} \mu_j \mu_m \left( \delta_{j,k} \delta_{\ell,m} + \delta_{j,\ell} \delta_{k,m} - \delta_{j,k,\ell,m} \right)\\
&\hspace{4cm}\times\ketbra{jk}{lm} \\
&= \mu^2 \omega .
\end{aligned} \end{equation}
Therefore, thanks to \eqref{eq:robs_primal_nonlsc} we immediately see that $\Rs_\S(\psi)\leq \mu^2$, concluding the proof.
\end{proof}

We note that $\ell^1$ is dense in $\ell^2$; more generally, it can be shown that the set of states with $\Rl_\S(\rho) < \infty$ is dense in $\D(\H)$ with the trace norm topology.

\subsubsection{Comparison with standard robustness}

\txb{In finite-dimensional quantum mechanics, the two robustnesses $\R_\S$ and $\Rs_\S$ appear naturally in the quantification of the optimal rates for entanglement distillation and dilution (respectively) under asymptotically non-entangling operations~\cite{brandao_2008-1, brandao_2010, brandao_2010-1}. Since these two measures reflect complementary tasks, comparing them directly can yield useful insights into some peculiar features of infinite-dimensional entanglement theory. As we have seen, they turn out to coincide for all pure states. This} prompts the natural question of whether there exist states that satisfy $\R_\S(\rho)<\Rs_\S(\rho)$. We will now construct an extreme example of this behavior. To this end, let us start by recalling the definition of negativity, an entanglement measure given by~\cite{vidal_2002}
\begin{equation}
    N(\rho) \coloneqq  \frac12 \left( \left\|\rho^\Gamma \right\|_1 - 1\right) ,
\end{equation}
where $\Gamma$ denotes partial transposition with respect to either of the subsystems~\cite{peres_1996}. For a pure state $\ket{\psi}$ as in \eqref{eq:Schmidt}, one finds that $2N(\psi)+1=\R_\S(\psi)=\Rs_\S(\psi)$ coincides with the two robustnesses, and its value is given by Prop.~\ref{thm:rob_ent_pure}. In general, using the fact that separable states have a positive partial transpose, it is not difficult to show that
\begin{equation}
    \Rs_\S(\rho) \geq N(\rho)+1
\end{equation}
for all states $\rho\in \D(\H_A\otimes \H_B)$~\cite{vidal_2002}.

Our example is based on a famous operator first constructed and studied by Hilbert~\cite{WeylPhD}. Let $\ell^2(\CC)$ be the Hilbert space of square-summable complex sequences $(a_n)_{n\in \NN_+}$ (with index starting from $1$), and denote by $\{\ket{n}\}_{n\in \NN_+}$ its canonical basis. Construct the Hilbert operator $H_{-1}$ on $\ell^2(\CC)$ with matrix representation
\begin{equation}
(H_{-1})_{n,m} \coloneqq \left\{ \begin{array}{ll} 0 & \text{if $n=m$,} \\ \frac{1}{n-m} & \text{if $n\neq m$.} \end{array} \right.
\label{eq:Hilbert_operator}
\end{equation}
It has been shown by Hilbert himself in his lectures~\cite{WeylPhD} that $H_{-1}$ is bounded. Schur~\cite{Schur1911} later proved that in fact
\begin{equation}
\left\| H_{-1}\right\|_\infty = \pi .
\label{Hilbert-Schur}
\end{equation}
For an excellent account of this and related topics, see the book by Hardy, Littlewood and P\'{o}lya \cite[Sec.~8.12]{INEQUALITIES}, that by Steele~\cite[Ch.~10]{MASTERCLASS}, and the set of lecture notes by Jameson~\cite{Jameson-Hilbert}.

\begin{pro}\label{prop:rob_ent_infinite}
There exist a density operator $\omega$ over $\ell^2(\CC)$ such that the corresponding maximally correlated state
\begin{equation}
    \rho[\omega] \coloneqq \sum_{n,m=1}^\infty \omega_{n,m} \ketbra{n,n}{m,m}
    \label{nasty}
\end{equation}
on $\ell^2(\CC)^{\otimes 2}$ satisfies that
\begin{equation}
    \R_\S(\rho) \leq 2\, ,\qquad \Rs_\S(\rho) = N(\rho) = +\infty\, . 
\end{equation}
\end{pro}

\begin{proof}
Note that $H_{-1}^\intercal = - H_{-1}$, so that $i H_{-1}$ is self-adjoint. We proceed by setting
\begin{align}
    \omega_\pm &\coloneqq \frac{1}{c}\, D\left(\id \pm \frac{i}{\pi} H_{-1} \right) D , \label{omega} \\
    D &\coloneqq \sum_{n=1}^\infty d_n\, \ketbra{n}{n} , \label{D} \\
    d_n &\coloneqq \frac{1}{\sqrt{n} \ln(n+1)} , \\
    c &\coloneqq \sum_{n=1}^\infty d_n^2 < \infty .
\end{align}
We now verify that the state $\rho[\omega_+]$ constructed as in \eqref{nasty} with $\omega_+$ given by \eqref{omega} satisfies all desired properties. First, thanks to \eqref{Hilbert-Schur} we have that $\omega_\pm \geq 0$. By definition of $c$, we also see that $\Tr \omega_\pm =1$, so that $\omega_\pm$ are both valid density operator on $\H$. The associated states $\rho_\pm \coloneqq \rho[\omega_\pm]$ then satisfy that
\begin{equation} \begin{aligned}
\Rs_\S(\rho_\pm) &\geq N(\rho_\pm)+1 \\
&= \frac12 \left( \left\|\rho_\pm^\Gamma\right\|_1 + 1 \right) \\
&= \frac12 \left( \left\|\omega_\pm\right\|_{\ell^1} + 1 \right) ,
\end{aligned} \end{equation}
where $\|X\|_{\ell^1} \coloneqq \sum_{n,m=1}^\infty |X_{n,m}|$. Note that
\begin{align*}
\left\|\omega_\pm\right\|_{\ell^1}  &= \frac{1}{c} \sum_{n,m=1}^\infty d_n d_m \left( \delta_{n,m} + \frac1\pi \left| (H_{-1})_{n,m} \right| \right)  \\
&= \frac{1}{\pi c} \sum_{n,m=1}^\infty d_n d_m \left| (H_{-1})_{n,m} \right| +1\\
&= +\infty\, ,
\end{align*}
where the last equality is proved in \cite[Sec.~8.12, p.~214]{INEQUALITIES}. Therefore,
\begin{equation}
\Rs_\S (\rho_\pm) = N(\rho_\pm) = +\infty\, .
\end{equation}
It remains to prove that $\R_\S(\rho_\pm)\leq 2$. This follows immediately by observing that
\begin{equation}
\rho_\pm \leq \rho_+ + \rho_- = \frac{2}{c} \sum_{n=1}^\infty d_n^2\, \ketbra{nn}{nn}\, ,
\end{equation}
and the latter operator is clearly separable. This concludes the proof.
\end{proof}

The above example shows that our measure $\Rl_\S$ and the standard robustness $\Rsl_\S$ behave in general very differently, with the latter sometimes assigning an infinite value to states that have a finite amount of entanglement according to the former. \txb{From an operational standpoint, this could hint at the fact that there exist infinite-dimensional states with finite distillable entanglement but infinite entanglement cost. This intuition relies on the already mentioned fact that the generalized robustness is intimately connected with a type of entanglement distillation, while the standard robustness appears to play an analogous role for entanglement dilution~\cite{brandao_2008-1, brandao_2010, brandao_2010-1}. To make this connection rigorous, however, one would need to generalize the argument in~\cite{brandao_2008-1, brandao_2010, brandao_2010-1}, which as stated works in finite-dimensional systems only, to infinite-dimensional ones. We leave this for future works.}



\subsection{Non-Gaussianity}\label{sec:nongaussianity}
Non-Gaussianity is another type of quantum resource characterizing nonclassical features of quantum states defined in continuous-variable systems. 
In particular, it is known that non-Gaussian states are required for a number of quantum information processing tasks~\cite{giedke_2002,eisert_2002-1,niset_2009,lloyd_1999,lami_2018}, motivating its rigorous quantification in a resource-theoretic formulation that considers the set of Gaussian states free~\cite{Genoni2008quantifying,genoni_2010,Marian2013relative}.
On the other hand, another operational standpoint where free resources should be what can be easily accessible in experiments leads to a resource theory of genuine non-Gaussianity~\cite{takagi_2018,albarelli_2018}, which includes classical postprocessing and feedforwarded operations in its free operations and entails the closed convex hull of the set of Gaussian states as its free states. (See Ref.~\cite{Yamasaki2020GKP} for a recent application of this framework to non-Gaussian state transformations.)
Genuine non-Gaussian states provide advantages in protocols such as continuous-variable universal quantum computation~\cite{Mari2012positive, takagi_2018}, secure quantum communication~\cite{Lee2019secure}, and channel discrimination~\cite{takagi_2019-2}. 
The robustness measure established here can then be used to provide a quantitative account of the operational power of genuine non-Gaussianity.
Here, we provide formulas for the robustness of genuine non-Gaussianity for Fock states and single-photon-added/subtracted states. 

Let $\ket{\alpha,\xi}:=\mathcal{D}_\alpha S(\xi)\ket{0}$
with $\alpha=|\alpha|e^{i\phi}$, $\xi=|\xi|e^{i\theta}$ denote a single-mode Gaussian state where $S(\xi)= \exp\left[ \frac{1}{2}\left(\xi(a^\dag)^2 - \xi^*a^2 \right) \right]$ is a (rotated) squeezing operator, which reduces to \eqref{squeezing} when $\phi=0$ and $|\xi|=r$. 
We take the set of single-mode states without genuine non-Gaussianity as our set of free states, i.e., $\F=\G:=\cl \conv \lset \proj{\alpha,\xi} \sbar \alpha,\xi\in\mathbb{C} \rset$, as in \cite{takagi_2018}.

Although it is not clear whether the two expressions of robustness, $\Rl_\G$ and $\R_\G$, coincide for the theory of genuine non-Gaussianity in general, we show in the following that they indeed coincide for the Fock states, and the inequality \eqref{rob_pure_lower_bound} is achieved.  

\begin{pro}
\bal
 \Rl_\G(\proj{n}) &= \R_\G(\proj{n}) = \frac{1}{\sup_{\alpha,\xi\in\mathbb{C}}|\braket{n|\alpha,\xi}|^2} \\
 &= \left[\sup_{|\alpha|,r\geq 0, \theta\in[0,\pi)} P_n(|\alpha|,r,\theta)\right]^{-1}
 \label{eq:robustness nonGaussianity fock}
\eal
where 
\bal
 P_n(|\alpha|,r,\theta) = \frac{\left(\frac{\tanh{r}}{2}\right)^n}{n!\cosh{r}}\exp\left[-|\alpha|^2(1+\cos\theta\tanh{r})\right]\\
 \times \left|H_n\left[|\alpha|\left(\cosh{r}+e^{i\theta}\sinh{r}\right)\left(e^{i\theta}\sinh(2r)\right)^{-1/2}\right] \right|^2
 \label{eq:squeezed coherent photon statistics}
\eal
with $H_n(x)$ being the (physicists') Hermite polynomials satisfying the recursion relation $H_{n+1}(x)=2xH_n(x)-2nH_{n-1}(x)$. 

As a result, the min-relative entropy, relative entropy, max-relative entropy measures all collapse to $\log\Rl_\G(\proj{n})$ (see Cor.~\ref{cor:relative_entropy}).
\end{pro}

\begin{proof}
Take an arbitrary single-mode Gaussian state $\ket{\alpha,\xi}=\sum_{n=0}^\infty c_n \ket{n}$ where $c_n:=\braket{n|\alpha,\xi}$ is a coefficient for the number basis representation.  
Application of a random phase shift to this state gives the phase randomized state $\sigma_{\alpha,\xi}$ written as 
\bal
 \sigma_{\alpha,\xi} &= \frac{1}{2\pi}\int_0^{2\pi} d\theta e^{ia^\dagger a \theta} \dm{\alpha,\xi} e^{-ia^\dagger a \theta}\\
 &= \sum_{n=0}^\infty |c_n|^2 \dm{n},
\eal
which in particular implies
\bal
 \frac{1}{|\braket{n|\alpha,\xi}|^2} = \bra{n}\sigma_{\alpha,\xi}^{-1}\ket{n},\ \forall \alpha,\xi\in\mathbb{C}.
 \label{eq:overlap phase randomized}
\eal
Since $e^{ia^\dagger a \theta}$ is a Gaussian unitary, which maps Gaussian states to Gaussian states, we always have $\sigma_{\alpha,\xi}\in\F$, noting the equivalence between \eqref{eq:free states closure} and \eqref{eq:free states integral}.
Thus, considering a sequence of Gaussian states $\{\ket{\alpha_l,\xi_l}\}_l$ such that $\lim_{l\to\infty}|\braket{n|\alpha_l,\xi_l}|^2=\sup_{\sigma\in\F} \braket{n|\sigma|n}$, we have
\bal
 \lim_{l\to\infty}\frac{1}{|\braket{n|\alpha_l,\xi_l}|^2} \leq \Rl_\G(\proj{n})
 &\leq \R_\G(\proj{n})\\
 &\leq \lim_{l\to\infty}\bra{n}\sigma_{\alpha_l,\xi_l}^{-1}\ket{n}.
\label{eq:nonGaussianity fock inequalities}
\eal
where the first inequality is due to \eqref{rob_pure_lower_bound} and the third inequality is due to \eqref{rob_pure_alt}.
Then, the first two equalities in \eqref{eq:robustness nonGaussianity fock} follow by combining \eqref{eq:overlap phase randomized} and \eqref{eq:nonGaussianity fock inequalities}. 
The third equality in \eqref{eq:robustness nonGaussianity fock} is obtained by an explicit expression for the photon statistics of $\ket{\alpha,\xi}$ as a function of $\alpha=|\alpha|e^{i\phi}$ and $\xi=re^{i\theta}$~\cite{gerry2005introductory}, as well as by observing that the fidelity between two single-mode pure states can be computed by the overlap of their Wigner functions~\cite{takagi_2018} and that the Wigner functions of number states are symmetric under phase rotations, allowing us to fix $\phi=0$ and just optimize over $\theta$.
Moreover, $\theta$ only needs to be optimized over $[0,\pi)$ because of the symmetry $P_n(|\alpha|,r,\theta)=P_n(|\alpha|,r,2\pi-\theta)$, which can also be seen by the symmetry of the Wigner functions. 
\end{proof}
Although it appears difficult to further simplify the expression in \eqref{eq:robustness nonGaussianity fock}, numerical investigations suggest that the supremum is achieved at $\theta=0$, as well as $|\alpha|$ and $r$ satisfying $\bra{\alpha,\xi}a^\dagger a\ket{\alpha,\xi}|=|\alpha|^2+\sinh^2(r)=n$; this has been checked for $n\leq 10$, allowing us to obtain explicit values in this range (see Fig.~\ref{fig:robustness_NC_NG}).
\begin{figure}[t]
\centering
\includegraphics[width=0.4\textwidth]{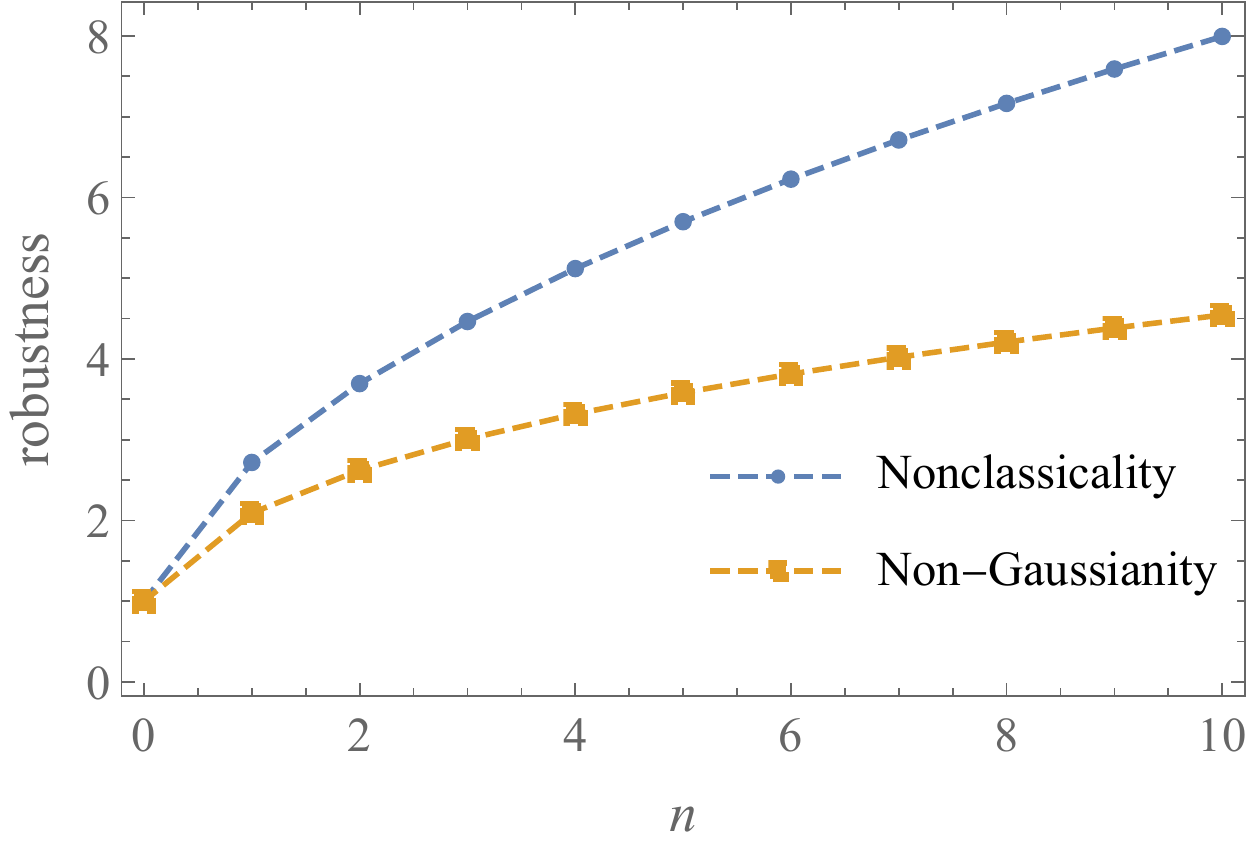}
\caption{Robustness of nonclasicality [Eq.~\eqref{eq:robustness Fock nonclassicality}] and non-Gaussianity [Eq.~\eqref{eq:robustness nonGaussianity fock}] for the Fock states $\ket{n}$.}
\label{fig:robustness_NC_NG}
\end{figure}

For the case of single-photon state, we obtain the following analytical solution.
\begin{pro}
 \bal
 \Rl_\G(\proj{1}) = \R_\G(\proj{1}) = \frac{4e}{3\sqrt{3}}
\eal
\end{pro}
\begin{proof}
For $n=1$, \eqref{eq:squeezed coherent photon statistics} becomes
\bal
 P_1(|\alpha|,r,\theta) = \frac{|\alpha|^2\left|\cosh(r)+e^{i\theta}\sinh(r)\right|^2}{\cosh^3(r)}\\
 \times\exp\left[-|\alpha|^2(1+\cos\theta\tanh(r))\right] 
\eal
By differentiating this with respect to $|\alpha|^2$, we get that $P_1(|\alpha|,r,\theta)$ achieves its maximum at $|\tilde\alpha|^2=(1+\cos\theta\tanh(r))^{-1}$ for any $r\geq 0$, $\theta\in[0,\pi)$ with
\bal
 P_1(|\tilde\alpha|,r,\theta) &= \frac{\left|\cosh(r)+e^{i\theta}\sinh(r)\right|^2}{(1+\cos\theta\tanh(r))\cosh^3(r)}e^{-1}\\
 &= \frac{\cosh(2r)+\sinh(2r)\cos\theta}{(1+\cos\theta\tanh(r))\cosh^3(r)}e^{-1}.
\eal
Then, we get 
\bal
 \frac{\partial}{\partial \cos\theta} P_1(|\tilde\alpha|,r,\theta) = \frac{\tanh(r)e^{-1}}{(1+\cos\theta\tanh(r))^2\cosh^3(r)}\geq 0,
\eal
implying that the maximum is achieved at $\theta = 0$. 
We are left to maximize $P_1(|\tilde\alpha|,r,\theta=0)$ over $r\geq 0$, which is found to give the result with the choice $r=\ln\sqrt{3}$.
\end{proof}

We can further obtain the robustness for photon-added/subtracted states defined in \eqref{eq:photon added number} and \eqref{eq:photon subtracted number}.
\begin{pro}
For single-mode photon-added/subtracted squeezed vacuum states $\ket{\zeta_r}_\pm$, we have for any $r\geq 0$,
\bal
\Rl_\G({\zeta_r}_\pm) = \R_\G({\zeta_r}_\pm)= \frac{4e}{3\sqrt{3}}.
\eal
\end{pro}
\begin{proof}
 Recall that $\ket{\zeta_r}_+=\ket{\zeta_r}_-$ so we only consider $\ket{\zeta_r}_+$.
 We first note that $\ket{\zeta_r}_+ = S(r)\ket{1}$. To see this, observe that 
 \bal
  \left((a^\dagger)^2-a^2\right)^k a^\dagger\ket{0} \propto a^\dagger \left((a^\dagger)^2-a^2\right)^k 
 \ket{0} 
 \eal
 for any $k\in \mathbb{Z}_+$, implying $S(r)a^\dagger\ket{0}\propto a^\dagger S(r)\ket{0}$ and thus $\ket{\zeta_r}_+=S(r)\ket{1}$ taking into account the normalization. 
Since $S(r)$ is a Gaussian unitary which does not affect the degree of non-Gaussianity, we get $R_\G({\zeta_r}_+)=R_\G(\proj{1})=4e/(3\sqrt{3})$.
\end{proof}

We further note that, since all coherent states $\ket\alpha$ are Gaussian, the robustness of non-Gaussianity is always upper bounded by the robustness of nonclassicality $\Rl_\C$. We can thus employ our results obtained in Sec.~\ref{sec:nonclassicality} as useful bounds. For instance, we immediately have from Prop.~\ref{prop:cat_states} that the robustness of non-Gaussianity of a cat state $\ket{\alpha_\pm}$ is bounded above by $2$ as $\alpha \to \infty$. Since $\log \R_\G(\alpha_\pm)$ upper bounds the minimal relative entropy distance from the set of genuine Gaussian states, we then have that the relative entropy distance is bounded as
\begin{equation}\begin{aligned}
	\inf_{\sigma \in \G} D(\alpha_\pm \| \sigma) \leq 1
\end{aligned}\end{equation}
in the limit of large $\alpha$. This constitutes a significant qualitative and quantitative difference from the relative entropy quantifier of non-Gaussianity considered in Ref.~\cite{albarelli_2018}, which grows to infinity as $\alpha \to \infty$.


\subsection{Quantum coherence}\label{sec:coherence}

The operational aspects of the principle of superposition have recently been formalized in a resource-theoretic framework as the resource theory of coherence~\cite{baumgratz_2014,streltsov_2017}, with an extension to the infinite-dimensional case considered explicitly in~\cite{zhang_2016}. Due to superposition being a basis-dependent concept, the study of coherence begins by identifying a countable orthonormal basis for the Hilbert space $\H$, denoted by $\{\ket{i}\}_{i=1}^\infty$, as the set of free pure states. We stress here the difference from the theory of nonclassicality (Sec.~\ref{sec:nonclassicality}), where the classical states $\{\ket{\alpha}\}_{\alpha \in \CC}$ do not form a mutually orthogonal set. The mixed free states --- called incoherent states --- are then all states diagonal in the given basis, $\F = \I \coloneqq \cl \conv \{ \proj{i} \}_{i=1}^\infty$. For any operator with matrix representation $X = \sum_{i,j=1}^\infty X_{i,j} \ketbra{i}{j}$, the $\ell^1$ norm is defined as $\norm{X}_{\ell^1} = \sum_{i,j=1}^\infty |X_{i,j}|$, which can be understood as the norm $\norm{\cdot}_\I$ for this resource theory (see Sec.~\ref{sec:seminorms}). It is already known that the standard robustness $\Rs_\I$ does not provide a meaningful quantifier of this resource, but the generalized robustness $\R_\I$ has been successfully employed in finite dimensions~\cite{napoli_2016}.

The characterization of the robustness of coherence in infinite dimensions is very similar to our previous findings in the resource theories of entanglement and nonclassicality.
We start by noticing that also in this case the cone of free states is closed with respect to the weak* topology, which makes it possible to apply Thm.~\ref{thm:strong_duality_from_weakstar} and conclude that the lower semicontinuous robustness coincides with its simplified version.

\begin{lem}
The cone $\cone(\I)$ is closed in the weak* topology. Thus, for all states $\rho\in \D(\H)$ it holds that $\Rl_\I(\rho)=\R_\I(\rho)$.
\end{lem}

\begin{proof}
For $i,j\in \NN$, define the functionals $\varphi_{i,j}:\T(\H)\to \CC$ given by $\varphi_{i,j}(X)\coloneqq \braket{i|X|j}$. Note that $\varphi_{i,j}$ is weak*-continuous, essentially because the operator $\ketbra{j}{i}$ is of finite rank and thus compact. Clearly, denoting by $\C$ the cone of positive semidefinite operators on $\H$, it holds that
\begin{equation}
    \cone(\I) = \C\cap \bigcap_{i\neq j} \varphi_{i,j}^{-1}(0).
\end{equation}
Since $\C$ is weak*-closed, and so are the sets $\varphi_{i,j}^{-1}(0)$ because the functionals $\varphi_{i,j}$ are weak*-continuous, also $\cone(\I)$ must be weak*-closed. An application of Thm.~\ref{thm:strong_duality_from_weakstar} completes the proof.
\end{proof}

We can then obtain a number of results in a very similar way to the preceding sections.

\begin{cor}
Any pure state $\ket\psi = \sum_{i=1}^\infty \psi_i \ket{i}$ satisfies
\begin{equation}\begin{aligned}
    \Rl_\I (\psi) = \norm{\psi}_{\ell^1} = \left( \sum_{i=1}^\infty |\psi_i| \right)^2.
\end{aligned}\end{equation}
In particular, a pure state has finite robustness of coherence if and only if the sum $\sum_{i=1}^\infty |\psi_i|$ converges.
\end{cor}
\begin{proof}
The fact that $\Rl_\I(\psi) = \norm{\psi}_{\ell^1}$ is an immediate consequence of Prop.~\ref{thm:robustness_pure_norm}. Alternatively, it can be explicitly shown using a similar approach to our proof of Prop.~\ref{thm:rob_ent_pure}, generalizing a construction from~\cite[Thm. 4]{piani_2016}.
\end{proof}

The pure-state formula generalizes the finite-dimensional expression~\cite{piani_2016}.

\begin{cor}
For the states $\omega_\pm$ defined in Prop.~\ref{prop:rob_ent_infinite}, we have that
\begin{equation}\begin{aligned}
    \Rl_\I(\omega_\pm) \leq 2,\qquad  \norm{\omega_\pm}_{\ell^1} = \infty.
\end{aligned}\end{equation}
\end{cor}

The result shows that coherence constitutes another example of a theory where the robustness $\Rl_\I$ can be a more well-behaved quantifier than some of the commonly used measures --- in this case, the $\ell^1$ norm of coherence --- whose value diverges for states which are not necessarily infinitely resourceful.

\begin{cor}
Let $\H_1, \ldots \H_m$ be separable Hilbert spaces, each $\H_k$ with an incoherent orthonormal basis $\{\ket{i^{(k)}}\}_{i=1}^\infty$. Taking $\left\{ \bigotimes_k \ket{i_k^{(k)}} \right\}_{i_1, \ldots, i_m=1}^\infty$ as the incoherent basis of $\bigotimes_k \H_k$, we have that the robustness of coherence is multiplicative, in the sense that for any collection of states $\{\rho_k\}_{k=1}^m$ with $\rho_k \in \H_k$ it holds that
\begin{equation}\begin{aligned}
\Rl_\I\left(\bigotimes_k \rho_k\right) = \prod_k \Rl_\I(\rho_k).
\end{aligned}\end{equation}
\end{cor}
\begin{proof}
Follows in the same way as the proof of Prop.~\ref{prop:multiplicativity_noncl}.
\end{proof}
This fact also generalizes a property known from finite dimensions~\cite{zhu_2017}.


\section{Discussion}\label{sec:discussion}

We introduced a general method of quantifying convex resources in infinite-dimensional probabilistic theories through the robustness measure $\Rl_\F$, which provides a non-trivial extension of a finite-dimensional quantifier. We showed that such a measure not only satisfies the properties desired from a bona fide resource monotone --- faithfulness, strong monotonicity, lower semicontinuity --- but it also admits a direct operational interpretation as a figure of merit in a class of channel discrimination tasks. By studying the conic optimization problems underlying the robustness, we showed that the measure enjoys a useful dual formulation and can always be computed by measuring a single, suitably chosen effect.

We investigated the robustness further in the case of continuous-variable quantum mechanics, establishing a number of results which include more robust strong duality relations as well as lower and upper bounds which aid the quantification of quantum resources. We used our results to compute the robustness exactly for a variety of states in the resource theories of nonclassicality, entanglement, coherence, and genuine non-Gaussianity. By comparing the robustness $\Rl_\F$ with the related standard robustness $\Rsl_\F$, we showed that the former can remain a useful resource quantifier even in cases when the latter diverges to infinity.

Our contribution here is twofold. First, we established a unified view of the quantification of infinite-dimensional resources and their practical applications in discrimination tasks, laying the foundations for a systematic operational investigation of general resource theories in infinite-dimensional GPTs. Second, together with Ref.~\cite{our_main}, we provided readily applicable methods for benchmarking important continuous-variable quantum resources which underlie practical technological applications. We expect the robustness to find use as a meaningful and accessible tool in the study of resources in quantum mechanics and beyond.

Interesting follow-up developments would be to further improve on the characterization of the robustness in settings of interest, in particular through comparison with other common resource measures as well as evaluation of the robustness for larger classes of states. Another question is to consider whether the quantification of the robustness simplifies in restricted settings, such as for Gaussian states or energy-constrained sets of states. Open questions also remain in the characterization of strong duality for the robustness in infinite-dimensional spaces.


\begin{acknowledgments}

L.L.\ is supported by the ERC Synergy Grant BIOQ (grant no. 319130) \txb{and by the Alexander von Humboldt Foundation}. B.R.\ is supported by the Presidential Postdoctoral Fellowship from Nanyang Technological University, Singapore. R.T.\ acknowledges the support of NSF, ARO, IARPA, AFOSR, the Takenaka Scholarship Foundation, and the National Research Foundation (NRF) Singapore, under its NRFF Fellow programme (Award No. NRF-NRFF2016-02) and the Singapore Ministry of Education Tier 1 Grant 2019-T1-002-015. Any opinions, findings and conclusions or recommendations expressed in this material are those of the author(s) and do not reflect the views of National Research Foundation, Singapore. G.F.\ acknowledges the support received from the EU through the ERASMUS+ Traineeship program and from the Scuola Galileiana di Studi Superiori.

\end{acknowledgments}

\onecolumngrid
\appendix

\section{Proof of results in Sec.~\ref{sec:optimization_banach}}\label{app:duality}

\begingroup
\renewcommand{\thethm}{\ref{prop:banach_states}}
\begin{pro}
For any $\omega \in \Omega$ it holds that
\begin{equation}\begin{aligned}
    P'(\omega) = D(\omega).
\end{aligned}\end{equation}
Furthermore, the primal problem is subfeasible if and only if there exists an optimal dual solution $W$.
\end{pro}
\endgroup
\begin{proof}
Letting $\lambda = P'(\omega)$, we consider two cases.

(Case $\lambda < \infty$). Defining the set
\begin{equation}\begin{aligned}
    \Q = \lset \big(\sigma - \tau, \, \< U, \sigma \> \big) \sbar \sigma \in \K_1, \tau \in \K_0 \rset \in \V \times \RR,
\end{aligned}\end{equation}
subfeasibility with optimal value $\lambda$ is then equivalent to $(\omega, \lambda) \in \cl \Q$, where we consider the product topology on $\V \times \RR$, and in particular $\< (A, B), (c,d) \> = \< A, c\> + \<B, d\>$. By the Hahn--Banach theorem, we can therefore strictly separate $(\omega, \lambda - \ve)$ from $\cl \Q$ for any $\ve > 0$, which means that there exist a choice of $Z \in \V\*$ and $q \in \RR$ such that
\begin{equation}\begin{aligned}\label{eq:hahn_banach}
    \<Z, \omega\> + q (\lambda - \ve) > \<Z, \sigma - \tau \> + q \< U, \sigma \> \quad \forall \sigma \in \K_1, \tau \in \K_0.
\end{aligned}\end{equation}
First, notice that $q < 0$: were this not the case, we would have
\begin{equation}\begin{aligned}
    \< Z, \omega \> + q \lambda \geq  \<Z, \omega\> + q (\lambda - \ve) > \<Z, \sigma - \tau \> + q \< U, \sigma \>
\end{aligned}\end{equation}
which would imply that $(\omega, \lambda)$ can be strictly separated from $\cl \Q$, a contradiction. Assume now that $\<Z, \tau \> < 0$ for some $\tau \in \K_0$. We can then take $\mu \in \RR_+$ large enough so that
\begin{equation}\begin{aligned}
    - \< Z, \mu \tau \> >  \<Z, \omega \> + q (\lambda - \ve),
\end{aligned}\end{equation}
which contradicts Eq.~\eqref{eq:hahn_banach} since $\mu \tau \in \K_0$ and $0 \in \K_1$. Therefore, we must have $Z \in \K_0\*$. Similarly, assume that $\<Z + q U, \sigma \> > 0$ for some $\sigma \in \K_1$. Taking $\mu \in \RR_+$ sufficiently large, we get
\begin{equation}\begin{aligned}
    \<Z + q U, \mu \sigma \> > \< Z, \omega \> + q (\lambda - \ve)
\end{aligned}\end{equation}
which contradicts Eq.~\eqref{eq:hahn_banach} since $\mu \sigma \in \K_1$ and $0 \in \K_0$, so it must hold that $Z + q U \in -\K_1\*$.

Defining $\displaystyle W = - Z/q$ we then see that $W \in \K_0\*$ and $U - W \in \K_1\*$, so that $W$ is a feasible solution to the dual problem. Choosing $\sigma = \tau = 0$ in Eq.~\eqref{eq:hahn_banach}, we obtain
\begin{equation}\begin{aligned}
    D(\omega) \geq \<W, \omega \> > \lambda - \ve,
\end{aligned}\end{equation}
and since $\ve > 0$ was arbitrary, we deduce that $D(\omega) \geq \lambda$. We will now verify that this solution is optimal. Suppose that there exists a choice of feasible $W' \in \K_0\*$ with $U-W' \in \K_1\*$ such that $\< W', \omega \> > \lambda$. It then holds that
\begin{equation}\begin{aligned}
    \< W', \omega \> - \lambda > 0 \geq \< W' - U, \sigma \>  - \< W', \tau \> \quad \forall \sigma \in \K_1\*, \tau \in \K_0\*,
\end{aligned}\end{equation}
so that the hyperplane $(W', -1)$ strictly separates $(\omega, \lambda)$ from $\cl \Q$, a contradiction. We thus have that $\lambda$ is the optimal value of the dual problem $D$.

(Case $\lambda = \infty$). This case occurs when there are no subfeasible solutions, which means that $(\omega, \mu)$ can be strictly separated from $\cl \Q$ for any $\mu \in \RR$. Taking $\mu > 0$ arbitrarily large, by the Hahn-Banach theorem we have a $Z \in \V\*$ and a $q \in \RR$ such that
\begin{equation}\begin{aligned}\label{eq:subfeas_unbounded}
    \< Z, \omega \> + q \mu > \< Z, \sigma - \tau \> + q \< U, \sigma \> \quad \forall \sigma \in \K_1, \tau \in \K_0.
\end{aligned}\end{equation}

If $q \geq 0$, a reasoning analogous to the above shows that $Z \in \K_0\*$ and $Z + q U \in -\K_1\*$. But since $U \in \C\* \subseteq \K_1\*$, we have
\begin{equation}\begin{aligned}
    \< Z , \sigma \> \leq \< Z + q U, \sigma \> \leq 0 \quad \forall \sigma \in \K_1
\end{aligned}\end{equation}
which shows that $Z \in -\K_1\*$. Now, as $\K_1 \subseteq \K_0$, the inclusion $Z \in \K_0\* \cap (-\K_1\*)$ can only hold when $\<Z, \sigma \> = 0 \; \forall \sigma \in \K_1$. But we also have that $Z + q U \in \K_0\* \cap (-\K_1\*)$, which implies that we must have $q=0$ since  $\<U, \rho\> > 0 \; \forall \rho \in \C\setminus\{0\}$. We then have that $\eta Z$ is a feasible dual solution for any $\eta \in \RR_+$, and since Eq.~\eqref{eq:subfeas_unbounded} implies that $\< Z, \omega \> > 0$, we can make the dual value arbitrarily large.

If, on the other hand, $q<0$ in Eq.~\eqref{eq:subfeas_unbounded}, then we simply follow the steps that we considered in the case $\lambda < \infty$ to obtain a feasible dual solution $W$ such that $\<W, \omega \> > \mu$. Taking $\mu \to \infty$, we have that the dual problem is unbounded with $D(\omega) = \infty$, and so no dual optimal solution can exist. 
\end{proof}

\begingroup
\renewcommand{\thethm}{\ref{lem:rob_based}}
\begin{lem}
For any $\omega \in \Omega$ it holds that
\begin{align}
    P(\omega) &= \inf \lset \lambda \sbar \omega = \lambda \sigma - (\lambda - 1) \tau, \; \tau \in \B_0,\; \sigma \in  \B_1\rset\\
    P'(\omega) &= \inf \Big\{ \lambda \;\Big|\; \exists \{\xi_n\}_n \to \omega \colon \xi_n = \lambda \sigma_n - (\lambda - 1) \tau_n, \nonumber\\
     & \hphantom{\inf \Big\{ \lambda \;\Big|\;} \tau_n \in \B_0,\; \sigma_n \in \B_1\Big\}. \label{eq:rob_subfeas_base}
\end{align}
In particular, it suffices to consider sequences of normalized elements $\xi_n \in \V$ such that $\<U, \xi_n\> = 1$ when considering subfeasibility.

Alternatively, we can write
\begin{equation}\begin{aligned}\label{eq:rob_subnormalized}
    P(\omega) &= \inf \lset \lambda \sbar \omega \in \lambda (\B_1 - \K_0) \rset\\
    P'(\omega) &= \inf \lset \lambda \sbar \omega \in \lambda \cl (\B_1- \K_0) \rset.
\end{aligned}\end{equation}
\end{lem}
\endgroup
\begin{proof}
The expression for $P(\omega)$ is obtained in the same way as for $\R_\F$ in Eqs.~\eqref{eq:rob_primal_nonlsc} and \eqref{eq:rob_primal_conic_expression}. 

For $P'(\omega)$, consider any subfeasible sequence $\{ \sigma_n - \tau_n \}_n \in \K_1 - \K_0$ converging to $\omega$ with $\< U, \sigma_n \> \to \lambda$. Notice that the fact that $\<U,\omega\> = 1$ requires that $\<U, \tau_n\> \to \lambda - 1$. Assuming that $\tau_n \neq 0$ (with the case of $\tau_n = 0$ proceeding similarly), we define
\begin{equation}\begin{aligned}
    \xi_n \coloneqq \lambda \frac{\sigma_n}{\< U, \sigma_n \>} - (\lambda - 1) \frac{\tau_n}{\< U, \tau_n \>}.
\end{aligned}\end{equation}
This gives
\begin{equation}\begin{aligned}
    \norm{\xi_n - \omega}_\Omega &= \norm{\xi_n - (\sigma_n - \tau_n) + (\sigma_n - \tau_n) - \omega}_\Omega\\
    &\leq \norm{\xi_n - (\sigma_n - \tau_n)}_\Omega + \norm{(\sigma_n - \tau_n) - \omega}_\Omega\\
    &\leq \norm{\left(\frac{\lambda}{\< U, \sigma_n \>} - 1\right) \sigma_n}_\Omega + \norm{\left(1 - \frac{\lambda - 1}{\<U, \tau_n\>}\right) \tau_n}_\Omega + \norm{(\sigma_n - \tau_n) - \omega}_\Omega\\
    &= \left|\lambda - \<U,\sigma_n\>\right| + \left|\<U,\tau_n\> - (\lambda-1)\right| + \norm{(\sigma_n - \tau_n) - \omega}_\Omega\\
    &\to 0,
\end{aligned}\end{equation}
where in the fourth line we used that $\norm{\xi}_\Omega = \<U, \xi\>$ for any $\xi \in \C$~\cite{hartkamper_1974}, and in the last line we used that $\{\sigma_n - \tau_n\}_n$ is subfeasible with subfeasible value $\lambda$, so each of the terms must converge to $0$. We have thus shown the existence of a sequence of the form in \eqref{eq:rob_subfeas_base}; since any such sequence is also a valid subfeasible sequence, the two statements of the problem are equivalent.

The characterization in Eq.~\eqref{eq:rob_subnormalized} is immediate by writing any $\sigma' \in \K_1$ as $\lambda \sigma$ with $\sigma \in \B_1$ and $\lambda = \<U, \sigma'\>$.
\end{proof}
\begin{remark}
We note an alternative formulation which can be easier to characterize. The set $\B_1 - \K_0$ is in fact equivalent to $\B_{1\leq} - \K_0$ where $\B_{1\leq}$ denotes subnormalized states in the cone $\K_1$, i.e., $\B_{1\leq} \coloneqq \lset  \sigma \in \K_1 \sbar \<U, \sigma \> \leq 1 \rset$. Indeed, the inclusion $\subseteq$ is obvious; for $\supseteq$, consider that any $x \in \B_{1\leq} - \K_0$ can be written as $x = \mu \sigma - \tau$ for some $\mu \in [0,1]$, $\sigma \in \B_1$, $\tau \in \K_0$. Since $\sigma \in \K_1$ and $\K_1 \subseteq \K_0$, we can then define $\tau' = \tau + (1-\mu) \sigma \in \K_0$ and write $x = \sigma - \tau' \in \B_1 - \K_0$.
\end{remark}


\section{Details of computations in Sec.~\ref{sec:nonclassicality}}\label{app:nonclassicality}

\begingroup
\renewcommand{\thethm}{\ref{prop:photon_added}}
\begin{pro}
Let $\ket{\zeta_r}_\pm$ be single-photon-added/subtracted squeezed vacuum states. Then,
\begin{equation}\label{eq:photon_added_lower2}
    e^{-r+1}\cosh^2(r) \leq \Rl_\C({\zeta_r}_\pm) \leq \frac{4e^{2r}}{3\sqrt{3}\sinh(r)}
\end{equation}
for all $r\geq 0$.

A tighter lower bound can be obtained for $r \leq \ln\sqrt{2} \approx 0.35$ as \begin{equation}
\Rl_\C({\zeta_r}_\pm) \geq e \cosh(r)^{-3},
\end{equation}
and for $r \geq \ln\sqrt{2}$ as
\begin{equation}
\Rl_\C({\zeta_r}_\pm) \geq \frac{4  e^{1+2r} }{27 \sinh(r)}.
\end{equation}
\end{pro}
\endgroup
\begin{proof}
The lower bounds all follow from \eqref{rob_lower_bound} with suitable choices of $\omega = \ket{\zeta_q}_+\!\bra{\zeta_q}$. Note that
\bal
    \sup_{\alpha\in \CC} \left|\braket{\alpha|\zeta_q}_+\right|^2 &= \sup_{\alpha\in \CC} \frac{|\braket{\alpha|a^\dagger|\zeta_q}|^2}{\cosh^2(q)} \\
    &= \sup_{\alpha\in \CC} \frac{1}{\cosh^3(q)}\, |\alpha|^2 e^{-|\alpha|^2+ \tanh(q) \Re(\alpha^2)}\\
    &= \sup_{|\alpha|\in [0,\infty)} \frac{1}{\cosh^3(q)}\, |\alpha|^2 e^{-|\alpha|^2+ \tanh(q) |\alpha|^2}\\
    &= \frac{e^{q-1}}{\cosh^2(q)} 
\eal
where in the last equation we used that the supremum is achieved at $|\alpha|^2 = (1-\tanh(q))^{-1}$ and
noting that $\cosh^3(q)(1-\tanh(q))=\cosh^2(q)e^{-q}$.
We also have 
\bal
 \left|{}_+\braket{\zeta_q|\zeta_r}_+\right|^2 = \frac{\left|\bra{\zeta_q}aa^\dagger\ket{\zeta_r}\right|^2}{\cosh^2(q)\cosh^2(r)}=\frac{1}{\cosh^3(r-q)}
\eal
where we used 
\bal
 \left|\bra{\zeta_q}aa^\dagger\ket{\zeta_r}\right|^2 &= \frac{1}{\cosh(r)\cosh(q)}\left(\sum_{n=0}^\infty \frac{2n+1}{4^n}\binom{2n}{n}\left(\tanh(r)\tanh(q)\right)^n\right)^2\\
 &= \frac{1}{\cosh(r)\cosh(q)(1-\tanh(r)\tanh(q))^3} = \frac{\cosh^2(r)\cosh^2(q)}{\cosh^3(r-q)}.
\eal
This gives
\bal
 \frac{\left|{}_+\braket{\zeta_q|\zeta_r}_+\right|^2}{\sup_{\alpha\in \CC} \left|\braket{\alpha|\zeta_q}_+\right|^2} = \frac{e^{-q+1}\cosh^2(q)}{\cosh^3(r-q)}
\eal
The general lower bound in Eq.~\eqref{eq:photon_added_lower2} follows by choosing $q=r$. For $r \leq \ln\sqrt{2}$, we get a tighter bound with the choice of $q = 0$, and for $r \geq \ln\sqrt{2}$ we can choose $\displaystyle q = \frac{1}{2} \ln\left(2e^{2r}-3\right)$.

To obtain an upper bound, let us consider the ansatz in \eqref{eq:thermal ansatz}.
\bal
    g_{\rm PA}(r,s) \coloneqq&\ {}_+\braket{\zeta_r | \sigma_s^{-1} | \zeta_r}_+ \\
    =&\ \frac{1}{\cosh^2(r)}\braket{0 | S^\dag (r) a S(s) \tau_{N(s)}^{-1} S^\dag (s) a^\dagger S(r) | 0} \\
    =&\ \frac{1}{\cosh^2(r)}\braket{\zeta_{r-s} | S^\dag (s) a S(s) \tau_{N(s)}^{-1} S^\dag (s) a^\dagger S(s) | \zeta_{r-s}} \\
    =&\ \frac{1}{\cosh^2(r)}\left(\cosh(s){}_+\bra{\tilde\zeta_{r-s}} +\sinh(s) {}_-\bra{\tilde\zeta_{r-s}}\right) \tau_{N(s)}^{-1} \left(\cosh(s)\ket{\tilde\zeta_{r-s}}_+ +\sinh(s) \ket{\tilde\zeta_{r-s}}_-\right),
\eal
where $\ket{\tilde\zeta_{r-s}}_+ = a^\dagger\ket{\zeta_{r-s}}$, $\ket{\tilde\zeta_{r-s}}_- = a\ket{\zeta_{r-s}}$ are unnormalized single-photon added and subtracted states. 
Since $\ket{\zeta_{r-s}}_+=\ket{\zeta_{r-s}}_-$ and thus $\ket{\tilde\zeta_{r-s}}_-=\tanh(r-s)\ket{\tilde\zeta_{r-s}}_+$, we have
\bal
 g_{\rm PA}(r,s)&= \left(\frac{\cosh(s)+\sinh(s)\tanh(r-s)}{\cosh(r)}\right)^2{}_+\braket{\tilde\zeta_{r-s}|\tau_{N(s)}^{-1}|\tilde\zeta_{r-s}}_+\\
 &= \frac{1}{\cosh^2(r-s)}{}_+\braket{\tilde\zeta_{r-s}|\tau_{N(s)}^{-1}|\tilde\zeta_{r-s}}_+
\eal
Using $\tau_{N(s)}^{-1}=\sum_{n=0}^\infty\left(\frac{e^{-s}}{\cosh(s)}\tanh^n(s)\right)^{-1}\dm{n}$, we get 
\bal
 {}_+\bra{\tilde\zeta_{r-s}} \tau_{N(s)}^{-1} \ket{\tilde\zeta_{r-s}}_+  &= \frac{e^s\cosh(s)}{\cosh(r-s)\tanh(r-s)} \sum_{n=0}^\infty \frac{2n+1}{4^n} \binom{2n}{n} \left(\frac{\tanh(r-s)}{\tanh(s)}\right)^{2n+1}\\
 &= \frac{e^s\cosh(s)}{\cosh(r-s)\tanh(r-s)} \frac{\tanh(r-s)}{\tanh(s)}\frac{1}{\left(1-\left(\frac{\tanh(r-s)}{\tanh(s)}\right)^2\right)^{3/2}}\\
  &= \frac{e^s\cosh(s)}{\cosh(r-s)}\frac{\tanh^2(s)}{\left(\tanh^2(s)-\tanh^2(r-s)\right)^{3/2}}
 \eal
 where in the second equality we used $\sum_{n=0}^\infty\frac{2n+1}{4^n}\binom{2n}{n}x^{2n+1}=x\left(\sum_{n=0}^\infty\frac{1}{4^n}\binom{2n}{n}x^{2n+1}\right)'=x\left(x/\sqrt{1-x^2}\right)'=x/(1-x^2)^{3/2}$ for $|x|<1$.
Then, we get
\bal
g_{\rm PA}(r,s) &= \frac{e^s\cosh(s)}{\cosh^3(r-s)}\frac{\tanh^2(s)}{\left(\tanh^2(s)-\tanh^2(r-s)\right)^{3/2}}\\
&= \frac{\cosh(s)\sinh^2(s)}{(1-\tanh(s))\left(\sinh(r)\sinh(2s-r)\right)^{3/2}}
\label{eq:photon added upper bound}
\eal

$g_{\rm PA}(r,s)$ achieves its minimum at $s=s_0(r):=\frac{1}{4}\ln\left(4e^{2r}-3\right)$ and it can be checked (after tedious calculation) that $g_{\rm PA}(r,s_0(r))=\frac{4e^{2r}}{3\sqrt{3}\sinh(r)}$, which gives the upper bound.
\end{proof}


\twocolumngrid
\bibliographystyle{apsrev4-1a}
\bibliography{main}

\end{document}